\documentclass[reqno,12pt,letterpaper]{amsart}
\usepackage{amsmath,amssymb,amsthm,graphicx,mathrsfs,url,bbm,mathtools}
\usepackage[usenames,dvipsnames]{color}
\usepackage[colorlinks=true,linkcolor=Red,citecolor=Green]{hyperref}
\usepackage{amsxtra}
\usepackage{wasysym} 
\usepackage{graphicx}
\usepackage{subcaption}
\def\arXiv#1{\href{http://arxiv.org/abs/#1}{arXiv:#1}}

\usepackage{mathtools}
\mathtoolsset{showonlyrefs}

\def\?[#1]{\textbf{[#1]}\marginpar{\Large{\textbf{??}}}}

\def\smallsection#1{\smallskip\noindent\textbf{#1}.}
\let\epsilon=\varepsilon 

\newcommand{\RR}{{\mathbb R}}
\newcommand{\NN}{{\mathbb N}}
\newcommand{\CC}{{\mathbb C}}
\newcommand{\TT}{{\mathbb T}}
\newcommand{\ZZ}{{\mathbb Z}}

\newcommand{\ML}{\mathcal{L}}
\newcommand{\ach}{{\operatorname{ac}}}
\newcommand{\ch}{{\operatorname{c}}}
\newcommand{\msH}{{\mathscr H}}
\newcommand{\msV}{{\mathscr V}}
\newcommand{\msU}{{\mathscr U}}
\newcommand{\msG}{{\mathscr G}}

\newcommand{\thot}{{\frac{\theta}{2}}}
\newcommand{\on}{{0,n}}
\newcommand{\unt}{{u_n^\theta}}
\newcommand{\unnt}{{u_n^{-\theta}}}
\newcommand{\unnnt}{{u_{-n}^{-\theta}}}
\newcommand{\Ent}{{E_\on^\theta}}

\newcommand{\ent}{{e_\on^\theta}}
\newcommand{\ennt}{{e_\on^{-\theta}}}
\newcommand{\w}{{x_2+h^{\frac{1}{2}}x_1,\xi_2 - h^{\frac{1}{2}}D_{x_1}}}

\newcommand{\upzero}{{(0)}}
\newcommand{\upone}{{(1)}}
\newcommand{\uptwo}{{(2)}}
\newcommand{\Ep}{E_{n,+}}
\newcommand{\Em}{E_{n,-}}
\newcommand{\Epm}{E_{n,\pm}}

\newcommand{\bcR}{\bar{\indic}_R}
\newcommand{\bcRW}{\bar{\indic}_R^W}
\newcommand{\cRW}{\indic_R^W}
\newcommand{\tcRW}{\tilde{\indic}_R^W} 
\newcommand{\tcRw}{\tilde{\indic}_R^w}
\newcommand{\LH}{\mathcal{L}(\msH_1;\msH_2)} 
\newcommand{\Ltxo}{L^2(\RR_{x_1};\CC^4)} 
\newcommand{\Ltx}{L^2(\RR_{x_2};\Ltxo)} 
\newcommand{\Ltxt}{L^2(\RR_{x_2};\CC^2)} 

\newcommand{\Bkp}{B^{-k'}_{x_1}}
\newcommand{\Lt}{L^2_{x_1}}
\newcommand{\bfk}{{\mathbf k}}

\newcommand{\msW}{\mathscr{W}}

\newtheorem{theo}{Theorem}
\newtheorem{prop}{Proposition}[section]

\newtheorem{lemm}[prop]{Lemma}

\newtheorem{rem}{Remark}

\newtheorem{claim}{Claim}
\numberwithin{equation}{section}

\DeclareMathOperator{\Ave}{Ave}

\DeclareMathOperator{\diag}{diag}

\DeclareMathOperator{\Spec}{Spec}

\DeclareMathOperator{\HS}{HS}
\DeclareMathOperator{\Weyl}{W}

\let\Im=\Imag

\let\Re=\Real
\DeclareMathOperator{\sgn}{sgn}

\DeclareMathOperator{\supp}{supp}

\DeclareMathOperator{\Tr}{Tr}

\def\indic{\operatorname{1\hskip-2.75pt\relax l}}
\usepackage{scalerel}

\newcommand\reallywidehat[1]{\arraycolsep=0pt\relax%
\begin{array}{c}
\stretchto{
  \scaleto{
    \scalerel*[\widthof{\ensuremath{#1}}]{\kern-.5pt\bigwedge\kern-.5pt}
    {\rule[-\textheight/2]{1ex}{\textheight}} 
  }{\textheight} %
}{0.5ex}\\           
#1\\                 
\rule{-1ex}{0ex}
\end{array}
}

\title[TBG in magnetic fields]{Magnetic response properties of twisted bilayer graphene}

\author{Simon Becker}
\address[Simon Becker]{ETH Zurich, Zurich, CH.}
\email{sion.becker@math.ethz.ch}

\author{Jihoi Kim}
\address[Jihoi Kim]{University of Cambridge, United Kingdom.}
\email{rk614@cam.ac.uk}

\author{Xiaowen Zhu}
\address[Xiaowen Zhu]{University of Washington, Seattle, USA.} 
\email{xiaowenz@uw.edu}

\begin{document}

\begin{abstract}
In this article, we analyse the Bistritzer--MacDonald (BM) model (also known as the continuum model) of twisted bilayer graphene (TBG) with an additional external magnetic field. We provide an explicit semiclassical asymptotic expansion of the density of states (DOS) in the limit of strong magnetic fields. The explicit expansion of the DOS enables us to study magnetic response properties such as magnetic oscillations which includes Shubnikov-de Haas and de Haas-van Alphen oscillations as well as the integer quantum Hall effect. In particular, we elucidate the role played by different types of interlayer tunnelings ($AA^{\prime}$/$BB^{\prime}$ vs. $AB^{\prime}$/$BA^{\prime}$) in the study of the DOS, and magnetic properties.
\end{abstract}

\maketitle

\section{Introduction}
\label{s:intr}
It is arguably one of the most exciting recent discoveries in condensed matter physics that by twisting two sheets of graphene at certain \emph{magic angles}, the material exhibits a superconducting phase \cite{C18}. The experimental discoveries were motivated by earlier theoretical work \cite{LPN07,BM11} which introduced the continuum model of twisted bilayer graphene (TBG). From this model they predicted the first magic angle by observing the appearance of a relatively flat band of the Hamiltonian at a small twisting angle. To discuss our study of TBG in magnetic fields, we first briefly introduce the BM model (see \S \ref{subsec: BM_model}, \cite{BM11}):

The BM model is an effective 4$\times$4 matrix-valued Hamiltonian $ \begin{pmatrix} {H}^{\theta}_{D} &   {T}^{\theta}(x) \\  ({T}^{\theta}(x))^* &  H^{-\theta}_{D} \end{pmatrix}$, $x \in \RR^2$, composed of two twisted-Dirac-operators ${H}_D^\theta, {H}_D^{-\theta}$ representing two isolated graphene sheets, according to the \emph{Wallace model} \cite{W47} respectively, and a tunneling potential term $T^\theta( x) =\begin{pmatrix}
  \alpha_0 V(\tfrac{x}{\lambda_\theta}) & \alpha_1 \overline{U}(- \tfrac{x}{\lambda_\theta}) \\ \alpha_1 U( \tfrac{x}{\lambda_\theta}) & \alpha_0 V(\tfrac{x}{\lambda_\theta}) 
\end{pmatrix}$ where the diagonal potentials and off-diagonal potentials represent two different types of interlayer tunneling potentials. In fact, when two layers of graphene are twisted at an angle $\theta$, a macroscopic honeycomb structure of scale $\lambda_\theta$, called the moir\'e pattern, is formed (by a purely geometrical superposition of two sheets of graphene; see Fig.\ref{fig: Moire}). Then the two different types of interlayer tunnelings (see Fig.\ref{fig: Moire}) are respectively:
\begin{enumerate}
  \item the chiral tunnelings $U( x/\lambda_\theta)$ and $\overline{U}(- x/\lambda_\theta)$ localized near the vertices of each moir\'e hexagon, with tunneling strength $ \alpha_1$ and a stacking similar to $AB^{\prime}$ and $BA^{\prime}$-stacking;
  \item the anti-chiral tunneling $V( x/\lambda_\theta)$, localized near the centers of each moir\'e hexagon, with a tunneling strength $\alpha_0$ and a stacking similar to $AA^{\prime}$/$BB^{\prime}$-stacking.
\end{enumerate}  
Here $A$ and $B$ label the equivalence classes of vertices on the honeycomb lattice and atoms on the lower lattice are indicated by a prime, cf. Figure \ref{fig: Moire}. 
We refer to the BM model as the \emph{chiral} or \emph{anti-chiral} model in the limit of purely chiral (${\alpha}_0 = 0$) or anti-chiral (${\alpha}_1 = 0$) tunneling interaction, respectively. 
\begin{figure}[t!]
  \includegraphics[height=6cm]{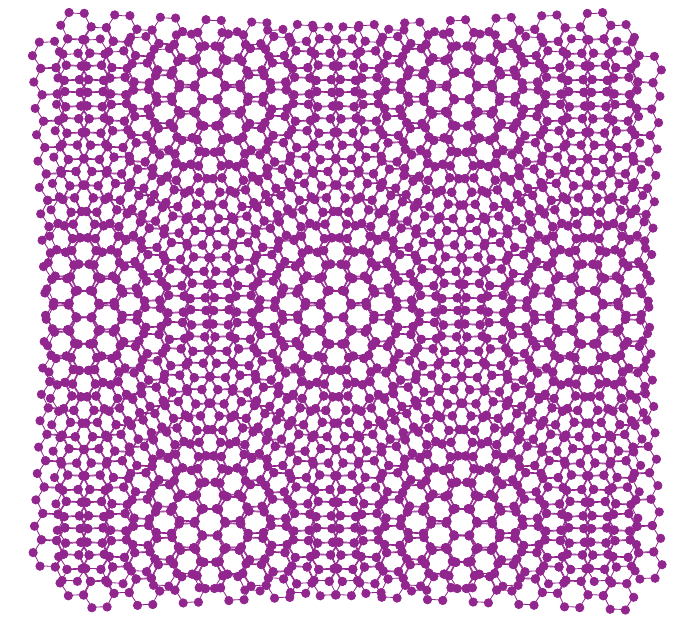}
  \includegraphics[height=6cm]{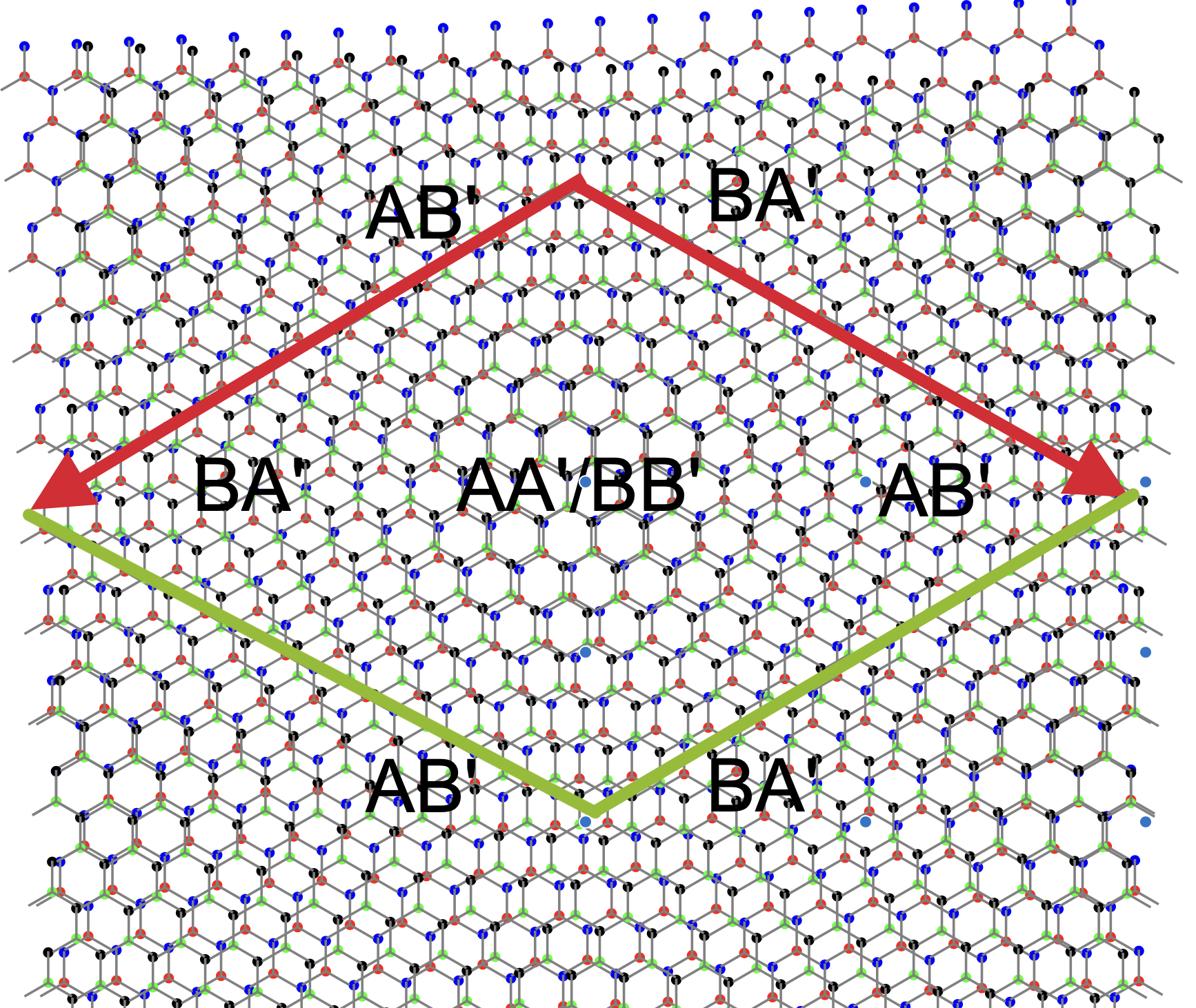}
\caption{On the left: Moir\'e pattern at twisting angle $\theta = 5^{\circ}$ with single moir\'e hexagon on the right, with (A={\color{red}red}, B={\color{blue}blue}) and (A'={\color{green}green}, B'={\color{black}black}) denote vertices of two sheets of graphene respectively.}
        \label{fig: Moire}
\end{figure}


While in the full BM model, the bands close to zero appear only approximately flat, it has been shown in \cite{TKV19,BEWZ20a,BEWZ20b,N21,NL22} that the chiral model exhibits a perfectly flat band at the magic angle \cite{TKV19,BEWZ20a} while the anti-chiral model does not \cite{BEWZ20b,LW21}. In our study of the magnetic response, we find that chiral and anti-chiral tunnelings exhibit intrinsically different features for the asymptotic expansion of the DOS in strong magnetic fields (see \S\ref{sec: DOS_Semiclassical}) which leads to different physical phenomena (see \S\ref{sec:applications}). 

In \S\ref{sec: DOS_Semiclassical}, we derive the explicit asymptotic expansion of the DOS in strong magnetic fields for both models. We find that the magnetic anti-chiral model has a similar behavior as a magnetic Schr\"odinger operator, where Landau levels in general split under perturbations of electric potential, while the magnetic chiral model has stable Landau levels especially at energy zero. Thus, chiral tunneling enhances the peaks of the DOS at Landau levels which leads to an enhancement of physical phenomena including magnetic oscillations and the quantum hall effect, which we discuss in \S\ref{sec:applications}, while anti-chiral tunneling weakens them. 

Our study of strong magnetic fields originates naturally from the interest in analyzing small twisting angles. In fact, as the twisting angle $\theta$ decreases to zero, the scale of the moir\'e hexagon $\lambda_\theta \sim (\sin \theta)^{-1}$ increases significantly. Thus, by rescaling coordinates the study of a fixed magnetic field at small twisting angle can be related to a fixed twisting angle in a strong magnetic field, see also \cite{D21,HA21} for further physical motivation. We denote the two scaling in the following as adiabatic (see \S\ref{Subsec: adiabatic_scaling}) and semiclassical (see \S\ref{Subsec: semiclassical_scaling}) scalings, respectively.



In particular, this means we provide the theoretical background for the study of the dependence of Landau levels on small twisting angle that have been studied for a simplified model in \cite{CHK} and numerically in \cite{MGJ20} for a tight-binding model. Furthermore, combining with the study of chiral and anti-chiral tunnelings, we put the substantially pronounced peaks of the DOS for small twisting angles at the Landau levels in \cite[Fig.\@ $2$,$3$]{MGJ20} on a rigorous footing. Furthermore, our results can also be used to understand the impact of strong pseudo-magnetic fields generated by physical strain. 




Finally, we summarize all our main results in an outline of the paper below:
\begin{itemize}
\item In Section \ref{sec:introduce_the_model}, we introduce the BM model with external magnetic field for TBG.
\item In Section \ref{sec: DOS}, we discuss general properties of the DOS.
\item In Section \ref{sec: DOS_Semiclassical}, we derive asymptotic formulae for the DOS:
\begin{itemize}
    \item of the chiral model: Theorem \ref{thm: chiral trace};
    \item of the anti-chiral model: Theorem \ref{thm: anti-chiral trace};
    \item is termwise-differentiable with respect to $B$: Prop \ref{prop: differentiability}.
\end{itemize}
\item In Section \ref{sec:applications}, we discuss physical applications of our semiclassical formulae.
\item The article also contains two technical appendices to which some of the computations and auxiliary results for the derivation of the DOS are outsourced.
\end{itemize}

 Our approach to analyze physical response properties rests on a thorough asymptotic analysis of the DOS. 
  Here, our approach is inspired by ideas of Helffer and Sj\"ostrand \cite{HS89} who studied the perturbation theory of periodic Schr\"odinger operators in strong magnetic fields and Wang \cite{W95}, who studied fine spectral asymptotics for random Schr\"odinger operators in strong magnetic fields. While Helffer and Sj\"ostrand  stopped at studying the spectral perturbation for strong magnetic fields, the so-called Grushin problem, we obtain a full asymptotic expansion of the DOS. This has also been obtained by Helffer and Sj\"ostrand for weak magnetic fields \cite{HS90} where the analysis relied on the semiclassical eigenvalue distribution close to a potential well. In our case, there is no natural well-structure and the asymptotic expansion relies on an asymptotic expansion of the parametrix with a splitting argument to overcome non-elliptic regions close to the real axis. Unlike in previous works by Helffer and Sj\"ostrand \cite{HS90} and an article on single-layer graphene by the first author and Zworski \cite{BZ19}, we resolve the issue of differentiability of the asymptotic expansion with respect to the semiclassical parameter by relating the asymptotic expansion with the one of the differentiated symbol, here. This expansion is needed for the rigorous analysis of the DOS when differentiated with respect to the magnetic field which is relevant for both the de-Haas van Alphen as well as the quantum Hall effect.

\smallsection{Acknowledgements}
We are very grateful Svetlana Jitomirskaya for initiating this collaboration and to Katya Krupchyk for valuable comments and references on the semiclassical expansion studied in this manuscript. This research was partially supported by Simons 681675, NSF DMS-2052899 and DMS-2155211.

\section{Introduction of magnetic BM model} 
\label{sec:introduce_the_model}
We start by introducing relevant notation.

 \smallsection{Notation}
\emph{Throughout this article we identify $\RR^2\simeq \CC$ by $x = (x_1,x_2) \simeq z=x_1+ix_2.$ We denote by $L$ the Lebesgue measure on $\RR^2\sim \CC$. For functions of complex variables $f(z,\bar z)$ we often just write $f(z).$ If there exists a constant $C_{\alpha}$ such that $ \| f \|_H \leq C_\alpha g $, we write $ f = \mathcal O_\alpha ( g )_H $. In particular,
$ f = \mathcal O ( h^\infty )_H  $ means that for any $ N $ there exists
$ C_N $ such that $ \|f \|_H \leq C_N h^N $. 
We also use the short notations $\langle x\rangle:=\sqrt{1+|x|^2}$, $B_r(x) = \{y: |y - x| \leq r\}$. }

\emph{We introduce the symbol class 
$ S(\mathbb R^{2n};\mathscr H):=\Big\{ p \in C^{\infty}(\mathbb R^{2n} \times \RR_{>0}; \mathscr H): \exists h_0, \text{ for all } \gamma \in \NN^2, \exists c_{\gamma}>0 \quad \text{s.t. for all } (x,\xi) \in \mathbb R^{2n} \text{ for all } h \in (0,h_0): \vert D_{(x,\xi)}^{\gamma} p(x,\xi,h) \vert \le c_{\gamma} \Big\}.$ In addition, let $S_\delta^k(\RR_{x,\xi}^2)$ denote the class of symbols $a \in C^{\infty}(\mathbb R^{2n} \times \RR_{>0})$ such that} 
\[
    |\partial_{x}^\alpha\partial_{\xi}^\beta a(x,\xi;h)| \leq C_{\alpha,\beta} h^{-k-\delta(\alpha + \beta )}, \quad \text{ for all } \alpha, \beta>0.
\]
 \emph{We denote standard partial derivatives in direction $x_i$ by $\partial_{x_i}$ and accordingly $D_{x_i}:=-i \partial_{x_i}.$ The principal symbol of a semiclassical operator $a(x,hD_{x})$ is denoted by $\sigma_0(a(x,hD_x)).$ We say a symbol $a$ has an asymptotic expansion in $ S_\delta^k$, $a\sim \sum_{j=0}^\infty a_j$, if $a\in S_\delta^k$ and there is a sequence of $a_j\in S_\delta^{k_j}$ s.t. $k_j\to -\infty$ as $j\to \infty$ and $a - \sum_{j= 0}^N a_j \in S_{\delta}^{k_{N+1}}$. When $k$ or $\delta = 0$, we omit the respective sub and superscript.  The spectrum of a linear operator $T$ is denoted by $\Spec(T).$ We also introduce rotated Pauli matrices $\sigma_k^\theta = e^{-i\frac{\theta}{4}\sigma_3} \sigma_k e^{i\frac{\theta}{4}\sigma_3},$ for $k = 1,2.$}

\subsection{Moir\'e lattices and TBG}
\label{subsec: BM_model}
We recall from the introduction that by twisting two honeycomb lattices with respect to each other, the emerging moir\'e honeycomb pattern exhibits different scales $\lambda_\theta $ at different twisting angles $\theta$. Thus it is easier to characterize such macroscopic honeycomb structures using a ``unit-size honeycomb lattice'' of side length $\frac{4\pi}{\sqrt{3}}$:

Let $\omega = \exp(\frac{2\pi i}{3})$, $\zeta_1 = 4\pi i \omega$, $\zeta_2 = 4\pi i \omega^2$. The ``unit-size honeycomb lattice'' is invariant under translations along a triangular lattice $\Gamma = \zeta_1\ZZ \oplus \zeta_2\ZZ$. We denote its unit cell, dual lattice, and the Brillouin zone of the dual lattice by $E = \CC/\Gamma$, $\Gamma^* = \eta_1\ZZ \oplus \eta_2\ZZ$, and $E^* = \CC/\Gamma^*$, where $\eta_1 = \frac{\omega^2}{\sqrt{3}}$ and $\eta_2 = -\frac{\omega}{\sqrt{3}}$. 

\subsection{Chiral and anti-chiral tunnelings} 
The chiral and anti-chiral tunneling potentials, $V$ and $U$, are smooth ``unit-size'' periodic functions (cf. \cite{BM11}) satisfying for $\mathbf a_j = \frac{4}{3} \pi i \omega^{j}$ with $j = 0,1,2$ the following symmetries
\begin{equation}
  \label{eq:symmU}
    \begin{split}
      &V ( z + \mathbf a_j ) = \bar \omega V ( z ) ,\ \ V ( \omega z ) = V ( z ) , \ \  \overline{ V ( z ) } = V ( - z ) , \ \ { V ( \bar z ) } = V (- z ) , \\
      &U ( z + \mathbf a_j ) = \bar \omega U ( z ) ,  \ \ U ( \omega z ) = \omega U ( z ) , \ \  \overline{ U ( \bar z ) } = U ( z ).
    \end{split}
\end{equation}
In particular, since $\zeta_1 = 3\mathbf a_1$, $\zeta_2 = 3\mathbf a_2$, we have $V(z+\zeta_j) = V(z)$ and $U(z + \zeta_j) = U(z)$ for $j = 1,2$. Thus $V(z)$, $U(z)$, $U_-(z):=U(-z)$ are periodic with respect to $\Gamma$. The tunneling potentials on the physical moir\'e scale are then $V({z}/\lambda_\theta), U({z}/\lambda_\theta), \overline{U_-}({z}/\lambda_\theta)$.

\begin{figure}[t!]
    \begin{subfigure}[b]{0.32\textwidth}
        \centering
        \includegraphics[width=7cm, height=1.8in]{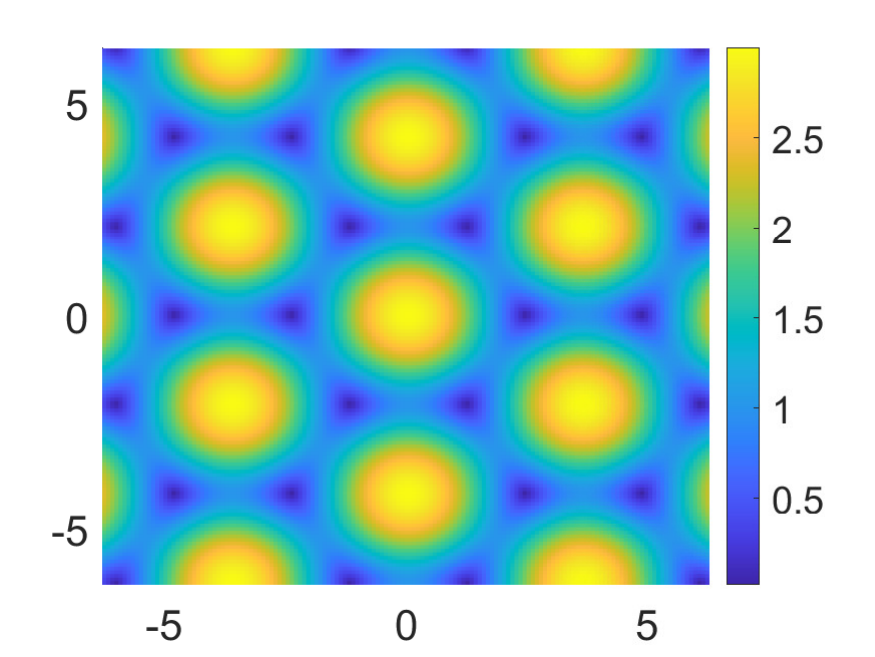}
        \caption{$\vert V \vert^2$ for $\mathrm{AA^{\prime}}/\mathrm{BB^{\prime}}$-coupling. }
       \label{fig:hcombAA}
    \end{subfigure}
        \begin{subfigure}[b]{0.32\textwidth}
        \centering      \includegraphics[width=6.5cm,height=1.8in]{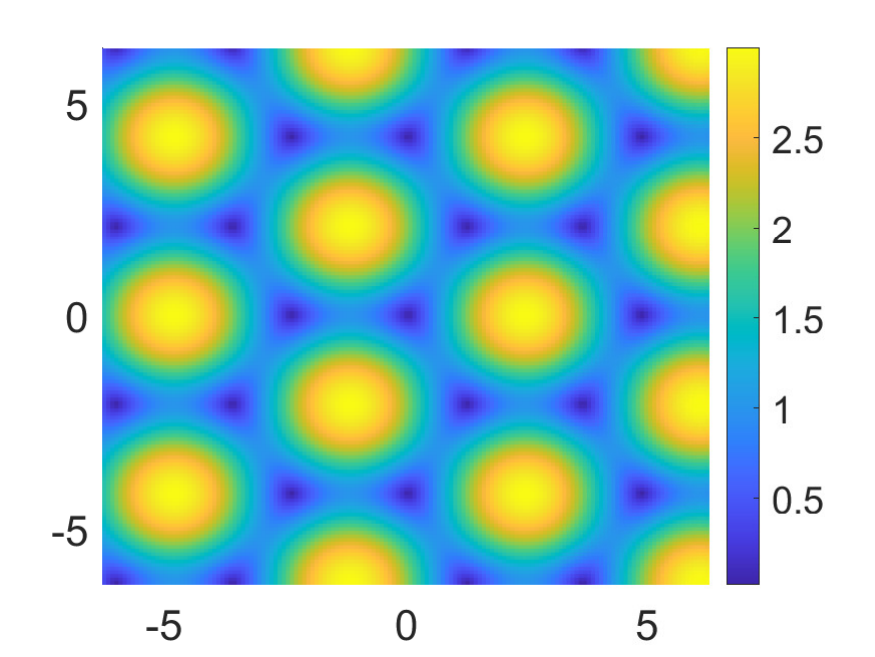}
        \caption{$\vert U \vert^2$ for $\mathrm{AB'}$-coupling.}
        \label{fig:fcellBA}
    \end{subfigure}
     \begin{subfigure}[b]{0.32\textwidth}
        \centering
        \includegraphics[width=6 cm,height=1.8in]{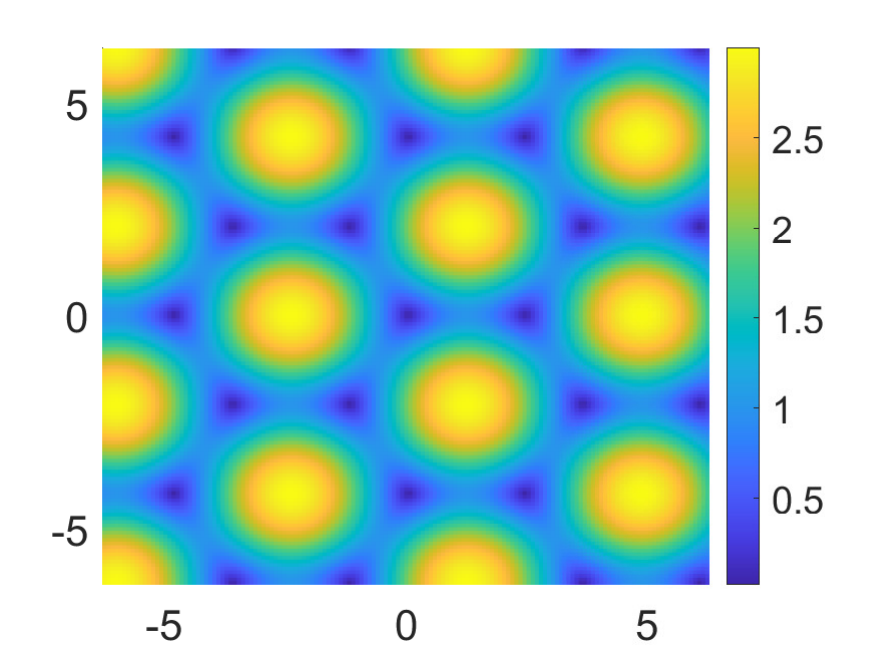}
        \caption{$\vert U_- \vert^2$ for $\mathrm{BA'}$-coupling.}
        \label{fig:fcellAB}
    \end{subfigure}%
        \caption{Modulus of tunneling potentials for various coupling types.}\label{fig:coupling}
\end{figure}

\subsection{Magnetic BM model with Adiabatic scaling} 
\label{Subsec: adiabatic_scaling}
To introduce the BM model with magnetic field we start with the physical or adiabatic scaling. Since we will immediately change to a semiclassical scaling, we denote all objects with a "$\sim$" in this paragraph to distinguish the two notations. 
Let $\tilde{A}(\tilde z) = (\tilde A_1(\tilde z), \tilde A_2(\tilde z), 0) \in C^{\infty}(\CC;\RR^3)$ be the magnetic vector potential of a magnetic field perpendicular to the TBG. The tunneling potentials, $U$ and $V$, defined on the ``unit-size honeycomb lattice'' are then rescaled to the physical moir\'e-size by rescaling coordinates by $\lambda_\theta$. Thus the magnetic BM model is $\tilde{\mathscr H}^{\theta}: D(\tilde{\msH}^\theta) \subset L^2(\CC;\CC^4) \rightarrow L^2(\CC;\CC^4)$ 
\begin{equation}
\begin{split}
\label{eq:adiabatic}
\tilde{\mathscr H}^{\theta}&:= \tilde{\msH}_{0}^\theta + \tilde{\msV} := \begin{pmatrix} \tilde{H}^{\theta}_{D}&  0 \\  0 &  \tilde{H}^{-\theta}_{D} \end{pmatrix}  + \begin{pmatrix}
  0 & \tilde{T}^{\theta}\\ (\tilde{T}^{\theta})^* & 0
\end{pmatrix}
\end{split}
\end{equation}
with $\tilde{H}^{\theta}_{D}=\sum\limits_{i=1}^2 \sigma_i^{\theta}(D_{\tilde x_i}-\tilde{A}_{i}(\tilde z))$ and 
$\tilde{T}^{\theta}(\tilde{z}) = \begin{pmatrix} \tilde \alpha_{0} V(\tilde z/{\lambda_\theta}) &  \tilde \alpha_{1} \overline{U_-}(\tilde z/\lambda_\theta) \\   \tilde \alpha_{1} U(\tilde z/\lambda_\theta) &  \tilde \alpha_{0}  V(\tilde z/\lambda_\theta) \end{pmatrix}$, where $\lambda_\theta$, $U$ and $V$ are given above and $\tilde{\alpha}_i$ represent the tunneling strength, $i = 1,2$.

\subsection{Magnetic BM model with Semiclassical Scaling}
\label{Subsec: semiclassical_scaling}
We shall now rescale the Hamiltonian in the previous paragraph to ``unit-size'' and multiply the Hamiltonian by $\lambda_\theta$ to work in another more convenient scaling called the \emph{semiclassical scaling}: Let $z = \tilde z/\lambda_\theta$, $\alpha_i = \lambda_\theta\tilde\alpha_i$, $A_i(z) = \lambda_\theta \tilde A_i(\lambda_\theta z)$ (overall represented by a unitary operator $U$), we consider 
\begin{equation}
  \label{eq:contmodel}
  \msH^\theta(z) := \lambda_\theta (U\tilde\msH^{\theta} U^{-1})(z) = \begin{pmatrix}
    H_D^\theta & 0 \\ 0 & H_D^{-\theta}
  \end{pmatrix} + \begin{pmatrix}
    0 & T(z) \\ T(z)^* & 0
  \end{pmatrix} =: \msH^\theta_0 + \msV(z),
\end{equation}
where $H_D^\theta = \sum_{i = 1}^2 \sigma_i^\theta (D_{x_i} - A_i(z))$, or equivalently, $H_D^\theta = e^{-i\frac{\theta}{4}\sigma_3} H_D e^{i\frac{\theta}{4}\sigma_3}$ where
\begin{equation}
  \label{eq: def_of_ab}
  H_D= \begin{pmatrix}
    0 & a \\ a^* & 0
  \end{pmatrix}\ \text{with~}\begin{cases}
    a = 2D_z - \overline{A(z)} \\
    a^* = 2D_{\overline{z}} - A(z)
  \end{cases}, \  T(x) = \begin{pmatrix} \alpha_0 V(z) & \alpha_1 \overline{U_-}(z) \\  \alpha_1 U(z)& \alpha_0  V(z) \end{pmatrix}.
\end{equation}
We denote the \emph{chiral model} by $\msH_{\ch}^\theta = \msH^\theta|_{\alpha_0 = 0}$ and the \emph{anti-chiral model} by $\msH_{\ach}^\theta = \msH^\theta|_{\alpha_1 = 0}$.

\begin{rem}[Why strong magnetic fields?]
  \label{rmk: strong_magnetic_field}
  Our study of strong magnetic fields in rescaled coordinates is motivated by the observation that small twisting angles naturally correspond, for constant magnetic fields, to the limiting regimes $\alpha\gtrsim 1 $ and $B \gg 1$. This provides the basis of our study of large magnetic fields which we coin the \emph{semiclassical scaling}.  
\end{rem}

\begin{figure}
  \includegraphics[height=4cm,width=5cm]{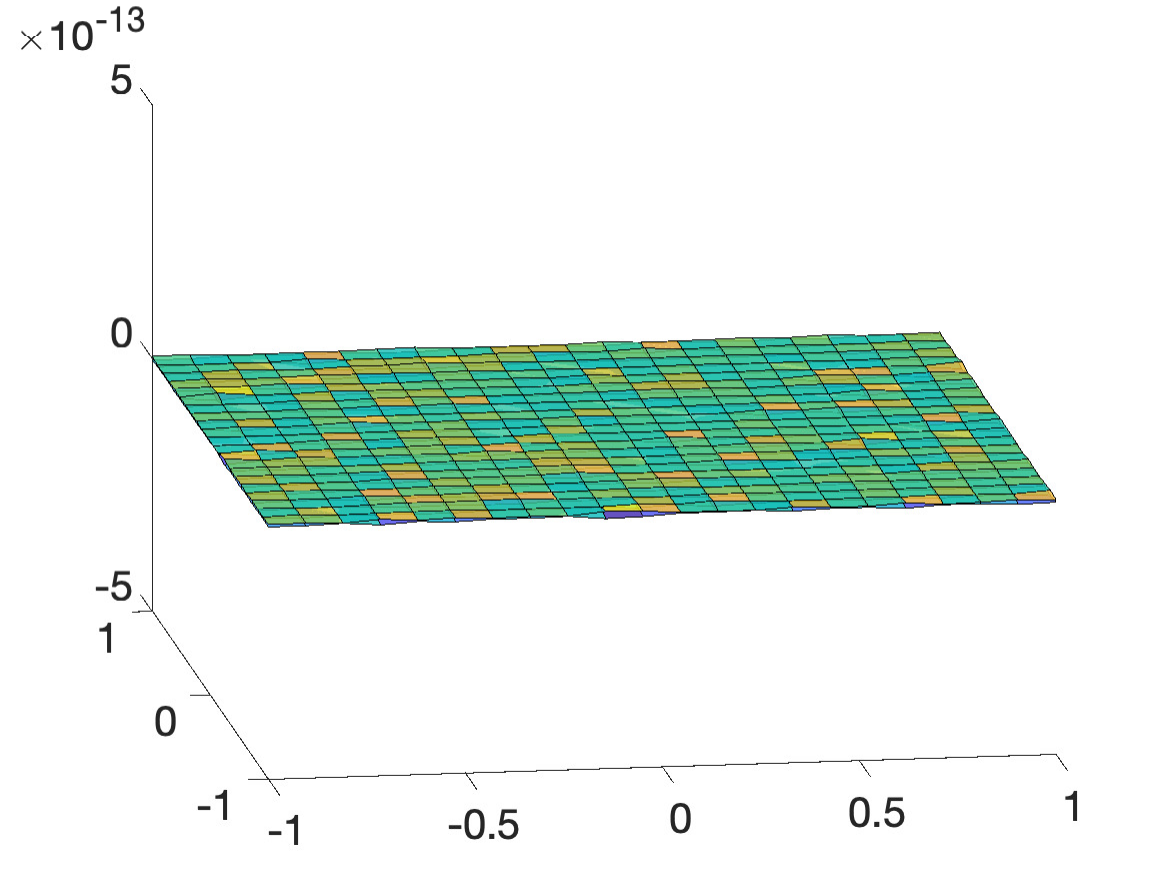}
  \includegraphics[height=4cm,width=5cm]{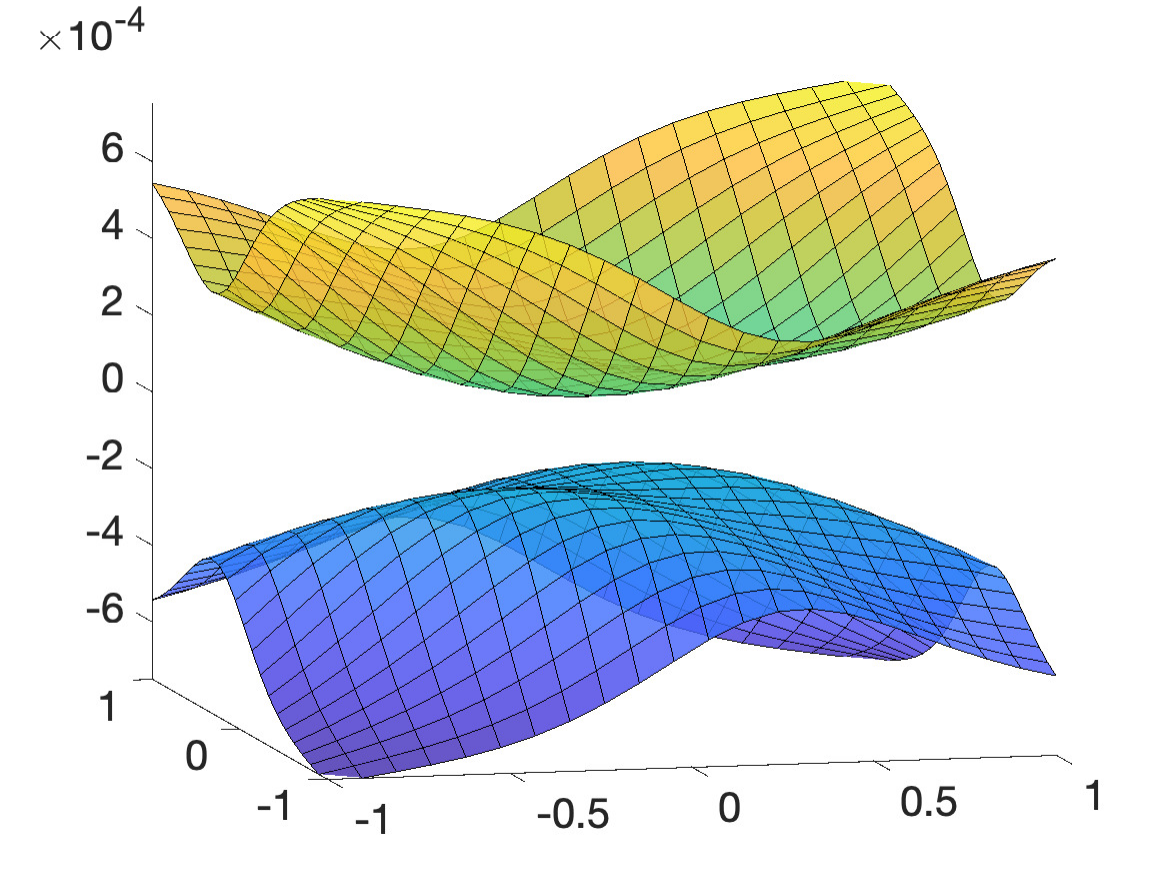}
  \includegraphics[height=4cm,width=5cm]{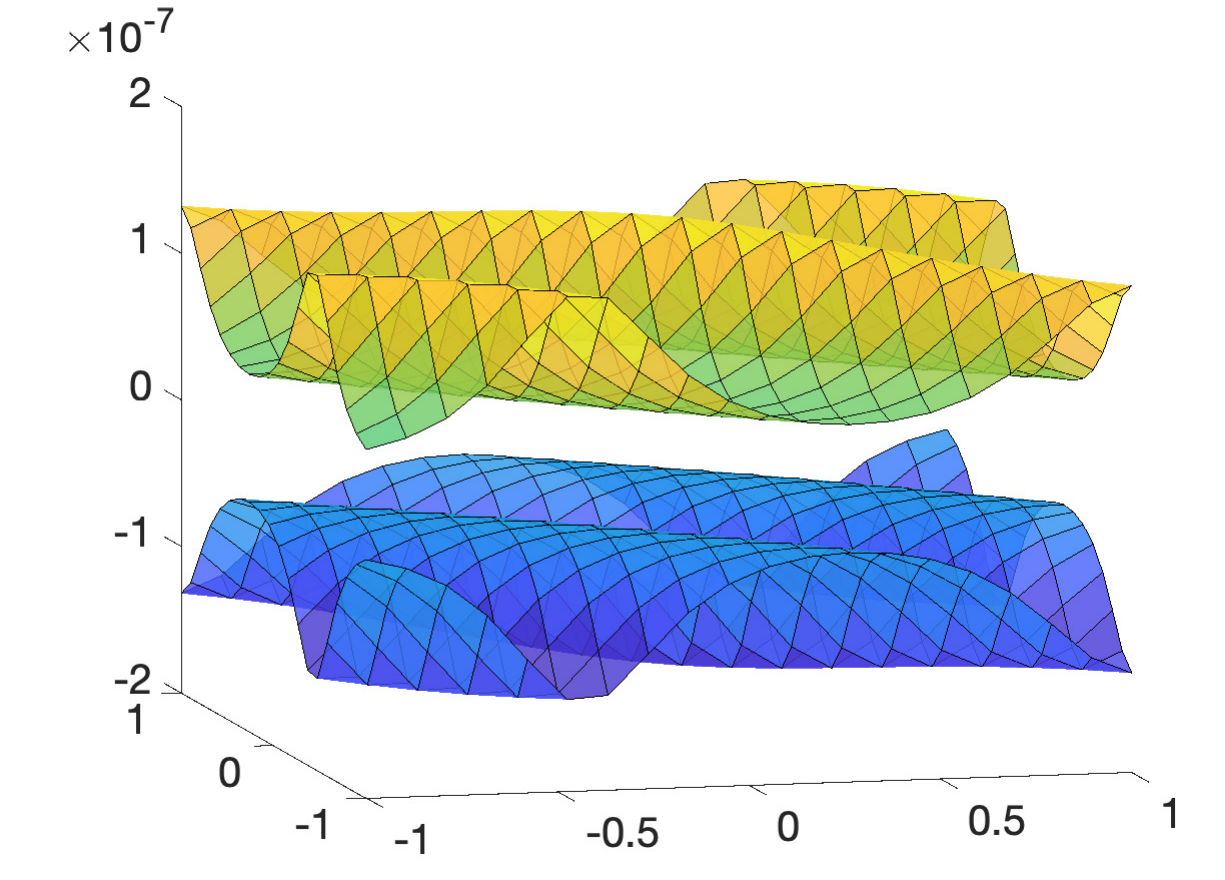}
  \caption{Constant magnetic field: On the left, flat bands for chiral model ($\alpha_1=1$); in the middle $(\theta=0)$ and on the right $(\theta=\pi)$ non-flat bands for anti-chiral model, ($\alpha_0=1$).}
  \end{figure}

\subsection{The chiral and anti-chiral model}
The chiral model is described by the Hamiltonian \eqref{eq:contmodel} for $\alpha_0=0$, such that upon conjugation by $\mathscr U=\operatorname{diag}(e^{i\theta/4},e^{-i\theta\sigma_3/4},e^{i\theta/4})$, $\mathscr H_{\operatorname{c}}=\mathscr U  \mathscr H^{\theta}\mathscr U$ it takes the form
\begin{equation}
\begin{split}
\label{eq:DcB}
\mathscr H_{\operatorname{c}}&=\begin{pmatrix}0 & (\mathcal D_c)^* \\ \mathcal D_c & 0  \end{pmatrix} \text{ with } 
\mathcal D_c= \begin{pmatrix}  2 D_{\bar z}-A_1(z)-iA_2(z)  & \alpha_1 U(z) \\ \alpha_1 U_-(z) &  2 D_{\bar z}-A_1(z)-iA_2(z)  \end{pmatrix}. 
\end{split}
\end{equation}

The anti-chiral model, with $\alpha_1=0$, can be conjugated by a unitary $\mathcal V$, with $\lambda=e^{i \frac{\pi}{4}}$,
to a Hamiltonian
\begin{equation}
\begin{split}
\label{eq:achiral}
\mathscr H^{\theta}_{\ach}&:=\mathcal V \mathscr H^{\theta}\mathcal V=\begin{pmatrix} 0& (\mathcal D^{\theta}_{\ach})^* \\ \mathcal D^{\theta}_{\ach}&0 \end{pmatrix} \text{ with } \mathcal V = \begin{pmatrix} \mathcal V_1 & \mathcal V_2 \\\mathcal V_2 & \mathcal V_1\end{pmatrix} \text{ for }\mathcal V_1 = \begin{pmatrix} i\lambda & 0 \\0 & 0 \end{pmatrix}, \mathcal V_2 = \begin{pmatrix}
  0 & 0 \\ 0 & -\bar{\lambda}
\end{pmatrix}, \\
\mathcal D_{\ach}^{\theta} &= \begin{pmatrix} \alpha_0 V(z) & e^{i \theta/2}(2 D_{\bar z}-(A_1(z)+iA_2(z)))    \\  e^{i \theta/2}(2 D_{z}-(A_1(z)-iA_2(z))) & \alpha_0 \overline{V(z)} \end{pmatrix}. 
\end{split}
\end{equation}
The off-diagonal structure implies that for both the chiral and anti-chiral model with magnetic field, the spectrum is symmetric with respect to zero. In particular, let $U:=(\sigma_3 \otimes \operatorname{id}_{\CC^2}) $ then it follows that $U\mathscr H_cU = -\mathscr H_c $ and $U \mathscr H_{\ach}^{\theta}U =-\mathscr H^{\theta}_{\ach}.$

\section{Density of states}
\label{sec: DOS}
In this section we study general properties of the density of states and study the possible values the density of states takes for the Hamiltonian of TBG.
\subsection{General properties}
In this subsection, we assume that the magnetic potential of the Hamiltonian is of the form $A =A_{\operatorname{per}} + A_{\operatorname{con}}$ where $A_{\operatorname{per}} \in C^{\infty}(E)$ and $A_{\operatorname{con}}$ is the vector potential of a constant magnetic field of strength $B.$
Let $f \in C_c(\RR)$ then we define the \emph{regularized trace}
\begin{equation}
    \label{eq: DOS_definition}
    \tilde\Tr(f(\msH^\theta)) = \lim_{r \rightarrow \infty} \frac{\Tr(\indic_{B_R} f(\msH^\theta) \indic_{B_R})}{\vert B_R \vert}
\end{equation}
where $\indic_{B_R}$ is the indicator function of the square centered at $0$ of side length $2R$. By Riesz's theorem, there exists the so-called \emph{density of states (DOS)} measure $\rho$ satisfying
\begin{equation}
    \label{eq: def_of_DOS}
    \tilde{\Tr}(f(\msH^\theta)) = \int_{\RR} f(t) \ d\rho(t).
\end{equation}
We start by showing the existence and smoothness of the DOS.
\begin{lemm}
\label{lemma: trace on E}
For $f \in C_c^{\infty}(\RR)$ the regularized trace of $f(\msH^\theta)$ exists, satisfies
\[ \tilde{\Tr}(f(\msH^\theta)) = \frac{1}{\vert E \vert} \Tr_{L^2(E)}(f(\msH^\theta))=\frac{1}{\vert E \vert} \int_{E} f(\msH^\theta)(x,x) \ dx,\]
and depends smoothly on $B \in \RR $ and $\theta \in \RR \setminus \{0\},$ with Schwartz kernel $f(\msH^\theta)(x,y)$ of $f(\msH^\theta).$
\end{lemm}

\begin{proof}
  Let $N_r, N_R \subset \Gamma$ be $ N_r := \{\zeta\in\Gamma: \zeta+E \subset  B_R\}$ and $ N_R := \{\zeta\in\Gamma: \zeta+E \subset B_R \neq \emptyset \}.$
  Then 
    \[
      S_r:= \bigcup\limits_{\zeta\in N_r} E + \zeta \subset B_R\subset \bigcup\limits_{\zeta\in N_R} E + \zeta=: S_R.    \]
  Thus for nonnegative $f$, 
  \begin{equation}
    \label{eq: squeezing_ready}
    \frac{1}{|S_R|}\Tr(\indic_{S_r}f(\msH^\theta))\leq \frac{1}{|B_R|}\Tr(\indic_{B_R}f(\msH^\theta)) \leq \frac{1}{|S_r|}\Tr(\indic_{S_R}f(\msH^\theta)).
  \end{equation}
  Furthermore, by definition, we see that for some $C,C'>0$, $\text{ for all } R$, 
  \begin{equation}
    \label{eq: difference_of_bands_near_BR}
      \#(N_R\setminus N_r) \leq CR,\quad \text{~and~}\quad |S_R\setminus S_r| \leq C'R.
  \end{equation}
By standard magnetic translation $\TT_\zeta$, which are defined e.g. in \cite[Lemma 2.1]{BKZ22} of our companion paper, satisfy $[\TT_\zeta, \msH^\theta] = 0$, therefore also $[\TT_\zeta, f(\msH^\theta)] = 0$. Furthermore, since $\TT_\zeta \indic_{E+\zeta} \TT_{-\zeta} = \indic_{E},$ thus $\Tr(\indic_{E+\zeta}f(\msH^\theta)) = \Tr(\indic_{E}f(\msH^\theta)).$
Hence,
  \[
    \Tr(\indic_{S_r}f(\msH^\theta)) = \sum\limits_{\zeta\in N_r} \Tr(\indic_{E+\zeta}f(\msH^\theta)) = (\#N_r)\Tr(\indic_E f(\msH^\theta))
  \]
and similarly $\Tr(\indic_{S_R}f(\msH^\theta)) = (\#N_R) \Tr(\indic_E f(\msH^\theta))$. Inserting this into \eqref{eq: squeezing_ready}, taking $R\to \infty$ we get by using \eqref{eq: difference_of_bands_near_BR} that
  \[
    \tilde\Tr(f(\msH^\theta)) = \frac{1}{|E|}\Tr_{L^2(E)}(f(\msH^\theta)).
  \]

To conclude the smooth dependence on $\theta$ and $B$, it suffices to adapt the arguments starting at \cite[p.251]{Sj89}. %

\end{proof}

In the next Proposition, we show that the integrated density of states of the twisted bilayer graphene Hamiltonian is stable under small perturbations of the magnetic field that do not close any spectral gaps. 

\begin{prop}
\label{prop:Sjostrand}
Let the magnetic vector potential $A = A_{\operatorname{con}}+A_{\operatorname{per}}$ be the sum of a linear potential associated with a constant field $B_0$ and $A_{\operatorname{per}} \in C^{\infty}(E)$. Assuming $t_0,t_1\notin \Spec(\mathscr H^{\theta})$, there exists a neighbourhood $\mathcal B\subset \RR$, open, connected, with $B_0 \in \mathcal B$ as well as $m = (m_1,m_2)\in \ZZ^2$ such that for any perturbation of the constant magnetic field $B\in \mathcal B$, $t_0,t_1 \notin \Spec(\mathscr H^{\theta})$ the DOS satisfies 
\[\rho((t_0,t_1))=\frac{1}{|E|}\bigg(m_1 \frac{B \vert E \vert}{2\pi}+m_2\bigg).\] 
\end{prop}
\begin{proof}
By density, we may assume that $B_0 \vert E \vert =2\pi \frac{p}{q}\in 2\pi \mathbb Q.$ This implies by choosing $\lambda=q$ that $B_0\vert E_{\lambda}\vert \in 2 \pi \ZZ$.
Let $\lambda_{n,\mathbf k}$ be the $n$-th Bloch band of $\mathscr H^{\theta}_{\mathbf k}$ for $n \in \ZZ$ on $\bfk \in E_{\lambda}^*$.  The spectrum of $\mathscr H^{\theta}$ has band structure and is given by $\Spec (\mathscr H^{\theta})=\cup_n J_n$ where $J_n=\bigcup\limits_{\bfk\in E_\lambda^*}\lambda_{n,\bfk}$. Let $t_0$, $t_1\not\in \Spec(\msH^\theta)$. We call $\mathcal I$ the set of bands fully contained in $(t_0,t_1).$ In terms of $\mathbf k\mapsto u_{n, \mathbf k}$ given by the eigenvectors associated with $\lambda_{n,\bfk}$ spectral projection of $\mathscr H_{\mathbf k}^{\theta}$ is given by
\[\indic_{(t_0,t_1)}(\msH^\theta_\bfk) v_\bfk(x)=\int_{E_\lambda} \indic_{(t_0,t_1)}(\msH^\theta_\bfk)(x,y)v_\bfk(y)dy\text{~with~} \indic_{(t_0,t_1)}(\msH^\theta_\bfk)(x,y):=\sum_{j \in \mathcal I}u_{j,\mathbf k}(x)\overline{u_{j,\mathbf k}(y)}.\]
So the spectral projection $\indic_{(t_0,t_1)}(\msH^\theta) =\mathcal U_{B_0}^{-1}\int^{\oplus}_{E_\lambda^*}\indic_{(t_0,t_1)}(\msH^\theta_{\mathbf k}) \frac{d\mathbf k}{\vert E_{\lambda}^*\vert}\mathcal U_{B_0}$ of $\mathscr H^{\theta}$ is 
\[\indic_{(t_0,t_1)}(\msH^\theta) u(x)=\int_{\RR} \indic_{(t_0,t_1)}(\msH^\theta)(x,y)u(y)dy\text{~with~}  \indic_{(t_0,t_1)}(\msH^\theta)(x,y)=\int_{E_\lambda^*}\mathcal \indic_{(t_0,t_1)}(\msH^\theta_{\mathbf k})(x,y)\frac{d\mathbf k}{|E_\lambda^*|}.\]
Since $t_0,t_1\notin \Spec(\msH^\theta)$ and let $N:=\vert \mathcal I\vert$, then by Lemma \ref{lemma: trace on E} 
\[\rho ((t_0,t_1)):=\int_{E_\lambda} \indic_{(t_0,t_1)}(\msH^\theta)(x,x) \ \frac{dx}{|E_{\lambda}|} =\int_{E_{\lambda}^*}\sum_{j \in \mathcal I} 1\,\frac{d\mathbf k}{4\pi^2}=\frac{N}{|E_{\lambda}|} \text{ s.t. } \tilde \Tr( \indic_{(t_0,t_1)}(\msH^\theta))=\frac{N}{|E_{\lambda}|}.\]
 If $f\in C_c^\infty (\RR)$, such that $f(x)=1$ for $x \in \operatorname{conv}\bigcup_n J_n$\footnote{conv is the convex hull} and $f(x)=0$ for $x \in \operatorname{Spec}(\mathscr H^{\theta}) \setminus \operatorname{conv}\bigcup_n J_n$, then
\[\rho((t_0,t_1))=\int_\RR f(t)\rho(dt)=\frac{N}{|E_{\lambda}|}.\]

Recall that $B_0\vert E\vert=\frac{B_0\vert E_\lambda\vert}{q} = 2\pi\frac{p}{q}\in 2\pi \mathbb Q$. We then introduce a new lattice $ \tilde \Gamma \subset \Gamma$ generated by $\widetilde{\zeta}_1=\zeta_1$ and $\widetilde{\zeta}_2=q\zeta_2$. Then $B_0 \vert \CC/\tilde \Gamma \vert \in 2\pi \ZZ$ and $|\Gamma/\tilde \Gamma|=q$. As before, if $t_0,t_1\notin \Spec (\mathscr H^{\theta})$, then 
\[\phi(B_0):=|E|\tilde{\Tr}(\indic_{(t_0,t_1)}(\msH^\theta)) = |E|\rho((t_0,t_1))=|E|\int_\RR f(t)d\rho(t)\in\frac{1}{q}\ZZ\subset \frac{B_0 \vert E\vert}{2\pi} \ZZ+\ZZ\]
where the last inclusion follows since $p$, $q$ are coprime i.e. there exist $ c,d\in\ZZ$ such that $cp+dq=1$. Note that if $z_0\in\RR\setminus\Spec(\mathscr H^{\theta})$, then there exists $\epsilon>0$ such that $z\notin\Spec (\mathscr H^{\theta})$ for all $|z-z_0|$ and small perturbations of the constant field $|B-B_0|<\epsilon$ and $\phi(B)$ is locally a smooth function of the constant field $B$ by Lemma \ref{lemma: trace on E}, so there exists $B_0\in\mathcal B\subset \RR$ open, connected and  $m \in \ZZ^2$ such that for $B\in \mathcal B$,
\[\rho((t_0,t_1))=\int_\RR f(t)d\rho(t)=\frac{1}{|E|}\bigg(m_1 \frac{B \vert E\vert}{2\pi}+m_2\bigg).\]
\end{proof}

\section{Semiclassical expansion of Density of states}
\label{sec: DOS_Semiclassical}

In this section, we provide explicit asymptotic expansions of the regularized trace in the semiclassical limit $B \gg 1$ for constant magnetic fields in the spirit of Remark \ref{rmk: strong_magnetic_field} for the chiral and anti-chiral model respectively. We also comment on the differentiability of the DOS at the end of this section in preparation for applications in the next section. 

We consider \eqref{eq:contmodel} with fixed $\theta$ and constant magnetic field $B$: 
\begin{equation}
  \label{eq: sec_6_Hamiltonian}
  \msH^\theta = \msH_0^\theta + \msV(x) = \begin{pmatrix}
    H_D^\theta & 0 \\ 0 & H_D^{-\theta}
  \end{pmatrix}+ \begin{pmatrix}
    0 & T(x) \\ T^*(x) & 0 
  \end{pmatrix}.
\end{equation}
Notice that the spectrum of $\msH_0^\theta$ is composed of \emph{Landau levels} $\lambda_{n,B} := \sgn(n)\sqrt{2|n|B}$ (see Lemma \ref{lemma: eigenspace}) which we will perturb by the tunnelling potential $\msV$ (see Remark \ref{rmk: perturbation}). To simplify the notation, we therefore introduce \emph{Landau bands} $\Lambda_{n,B,\msV} := (\lambda_{n-1,B} + \Vert \msV\Vert_\infty, \lambda_{n+1,B} - \Vert \msV\Vert_\infty)$ for $n\in \ZZ$, in which the spectrum of $\msH^\theta$ is contained around the $n$-th Landau level $\lambda_{n,B}$, cf. Remark \ref{rmk: gap}. 

We start by stating the main result of this section which is the asymptotic expansion of the DOS for the chiral model.
\begin{theo}[Chiral model]
  \label{thm: chiral trace}
  Let $\lambda_{n,B} = \sgn(n)\sqrt{2|n|B}$. 
  For a fixed $n\in \ZZ$, for $\epsilon>0$ small enough, for all $f\in C_c^{K}(\Lambda_{n,B,\msV}),$ with $K\geq \frac{6}{\epsilon} - 2$, we have
  \begin{equation}
    \label{eq: Asymp_chiral}
   \tilde \Tr(f(\mathscr H_{c})) = \left[\frac{B}{\pi} f(\lambda_{n,B}) + \frac{|n|}{2\pi}\Ave(\mathfrak U)f''(\lambda_{n,B})\right] + \mathcal O_{n,K,f,\msV} (B^{-\frac{1}{2}+\epsilon})
  \end{equation}
with $\mathfrak U(\eta) = \frac{\alpha_1^2}{8}\left[\alpha_1^2(|U_-(\eta)|^2 - |U(\eta)|^2)^2 + 4|\partial_{\bar{\eta}}\overline{U_-(\eta)}-\partial_\eta U(\eta)|^2\right]$, $\Ave(g) = \frac{1}{|E|}\int_{E} g(\eta) L(d\eta)$, $\eta = x_2 +i\xi_2$, and $\mathcal O_{n,K,f,\msV} =\mathcal O_n (\Vert \mathscr V \Vert_{\infty} \Vert f \Vert_{C^{K}})$.

 Furthermore, fix $N\in \NN^+$ and consider $2N+1$ Landau bands with $n \in \{-N,..,N\}$, then for all $\epsilon>0$ small enough, for any $f\in C_c^{K} ([\lambda_{-(N+1), B} + \Vert \msV\Vert_\infty , \lambda_{N+1,B} - \Vert \msV \Vert_\infty])$, with $K \geq \frac{6}{\epsilon} - 2$, we have 
\[
\tilde\Tr(f(\msH_\ch)) =  \sum\limits_{n = -N}^N \left[\frac{B}{\pi}f(\lambda_{n,B}) + \frac{|n|}{2\pi} \Ave(\mathfrak U)f''(\lambda_{n,B})  \right] +   \mathcal O_{(N),K,f,\msV} (B^{-\frac{1}{2}+\varepsilon}) 
\]
where $\mathcal O_{(N),K,f,\msV} := \sum\limits_{n=-N}^N \mathcal O_{n,K,f,\msV}$.
\end{theo}

Our proof also shows that all higher order terms, which in general have complicated expressions, in the expansion of $\tilde\Tr(f(\msH_\ch))$ are of the form $f^{(k)}(\lambda_{n,B})$ (see \eqref{eq:chiralmore}), which is different from the anti-chiral that we consider next.  

For the anti-chiral Hamiltonian the sub-leading correction in the regularized trace is already of order $\sqrt{B}.$ Since the dominant sub-leading correction in the anti-chiral case is one order higher than in the chiral case, we only state the correction up to order $\sqrt{B}.$ 

\begin{theo}[Anti-chiral model]
  \label{thm: anti-chiral trace}
Under the same assumption as in Theorem \ref{thm: chiral trace}, we have for all $\varepsilon>0$ small enough, $f\in C^K_c (\Lambda_{n,B,\msV})$ with $K \geq \frac{3}{\epsilon} - 1$
    \begin{equation}
    \label{eq: Anti-chiral_d.o.s.}
      \tilde \Tr(f(\mathscr H_{\text{ac}}^\theta)) =  \frac{B}{2\pi} t_{n,0}(f) - \frac{\sqrt{B}}{2\pi} t_{n,1}(f) +  \mathcal O_{n,K,f,\msV} (B^{\epsilon}),
    \end{equation}
where $\mathcal O_{n,K,f,\msV} =\mathcal O_n( \Vert \mathscr V \Vert_{\infty} \Vert f \Vert_{C^{K}})$ and
\[\begin{split}
     &t_{n,0}(f)=\Ave\left(f(\lambda_{n,B}+c_n)+ f(\lambda_{n,B}- c_n)\right), \  t_{n,1}(f)=\Ave\left( s_n^2 f'(\lambda_{n,B}+c_n)+ s_n^2 f'(\lambda_{n,B} - c_n)\right),\\
    &s_n(\eta;\theta) = \begin{cases}
      \alpha_0  \sin(\thot) |V(\eta)|,\\
      \alpha_0 |V(\eta)|,
    \end{cases} c_n(\eta;\theta) = \begin{cases}
      \alpha_0  \cos(\thot) |V(\eta)|, & n\neq 0, \\
      \alpha_0 |V(\eta)|, & n = 0.
    \end{cases}, \Ave(g) = \frac{1}{|E|} \int_E g(\eta)dL(\eta).
\end{split}\]

 Furthermore, fix $N\in \NN^+$ and consider $2N+1$ Landau bands with $n \in \{-N,..,N\}$. For any $\epsilon>0$, $f\in C_c^{K} ([\lambda_{-N-1, B} + \Vert \msV\Vert_\infty , \lambda_{N+1} - \Vert \msV \Vert_\infty])$ with $K\geq \frac{3}{\epsilon} - 1$, we have 
\[
\tilde\Tr(f(\msH_\ach^\theta))  =  \sum\limits_{n = -N}^N \left[\frac{B}{2\pi} t_{n,0}(f) + \frac{\sqrt{B}}{2\pi} t_{n,1}(f)\right] + \mathcal O_{(N),f,K,\msV} (B^{\varepsilon}) 
\]
where $\mathcal O_{(N),K,f,\msV} := \sum\limits_{n=-N}^N \mathcal O_{n,K,f,\msV}$.
\end{theo}

For the rest of this section, we shall temporarily stop using the identification $x = (x_1,x_2) \simeq z = x_1 + ix_2$. We will use the Landau gauge for the constant magnetic field, i.e. $A(z) = -iBx_1$ in \ref{eq: sec_6_Hamiltonian}. In this setup, Let $\Sigma_i^\theta = \diag(\sigma_i^\theta, \sigma_i^{-\theta})$. We can rewrite \eqref{eq: sec_6_Hamiltonian} as $\msH_0^\theta = \Sigma_1^\theta D_{x_1} + \Sigma_2^\theta (D_{x_2} + Bx_1)$. We will only use $x= (x_1,x_2)$ to denote the position, while $z$ is used in the resolvent $(\msH^\theta - z)^{-1}$. 

\smallsection{Quantizations}
\label{sec:quantization}
Let $x = (x_1,x_2)$, $\xi = (\xi_1,\xi_2)\in \RR^2$. For a symbol $a(x,\xi)\in S(\RR^4_{x,\xi})$, we define the $(h_1,h_2)$-Weyl quantization  $a^W(x,h_1D_{x_1}, h_2D_{x_2}): L^2(\RR_x^2) \to L^2(\RR_x^2)$ as
\begin{equation}
    \label{eq: quantization}
    (a^W(x,h_1D_{x_1},h_2D_{x_2})u)(x) = \frac{1}{2\pi } \int e^{\frac{i}{h_1}(x_1 - y_1)\xi_1 +\frac{i}{h_2}(x_2-y_2)\xi_2} a\left(\frac{x+y}{2},\xi\right)u(y) \ dy \ d\xi.
\end{equation}
In this section, we shall employ two different quantizations:  in Subsections \ref{ss: Symplectic reduction} 
and \ref{subsec: Grushin problem}, we use the $(h_1,h_2)=(1,1)$-Weyl quantization. Starting from Subsection \ref{ss: properties_of_effective_Hamiltonian}, we use the $(x_2,hD_{x_2})$-Weyl quantization of the operator-valued symbol which is related to the $(h_1,h_2)=(1,h)$-Weyl quantization (see Subsection \ref{ss: properties_of_effective_Hamiltonian} for more details). Occasionally, we denote $a^W(x,h_1D_{x_1},h_2D_{x_2})$ by $a^W$ for convenience.

\subsection{First Reduction: Symplectic reduction}
\label{ss: Symplectic reduction}
In this subsection, we first apply a symplectic reduction to $\msH^\theta$, then provide a spectral description of $\msH^\theta_0$ and $\msH^\theta$. In the end, we introduce the Helffer-Sj\"ostrand formula for our study of the regularized trace $\tilde\Tr (f(\msH^\theta))$.

\smallsection{Symplectic Reduction} Let $(h_1,h_2) =(1,1)$ for this subsection. Then the  operator $\msH_0^\theta$ and $\msV$, when viewed as a $(1,1)$-Weyl quantization, have symbols $\msH_0^\theta(x,\xi) = \Sigma_1^\theta \xi_1 + \Sigma_2^\theta (\xi_2+Bx_1)$ and $\msV(x)$ respectively. The following lemma provide the symplectic reduction of $\msH^\theta$:
\begin{lemm}
  \label{lemma: Unitary}
Let $h = 1/B$. Then there is a unitary operator $\msU$, symbols $\msG_0^\theta(x,\xi) = \Sigma_1^\theta \xi_1 + \Sigma_2^\theta x_1$ and $\msW(x,\xi) = \msV(x_2 + h^{1/2}x_1, h\xi_2 - h^{1/2} \xi_1)$, s.t. \begin{align}
     \label{eq:Unitary}
       &  \msU \msH^\theta_0(x,D_x) \msU^{-1} = \sqrt{B}\msG^\theta_0(x,D_{x}),\\
       \label{eq: Unitary_transform_for_multiplication}
       & \msU \msV(x) \msU^{-1} =  \msW^W(x,D_x).
 \end{align}
\end{lemm}
\begin{rem}
  Notice that $\msG_0^\theta(x,\xi)$ does not depend on $(x_2,\xi_2)$, thus the $(1,1)$-Weyl-quantization is $\msG_0^\theta(x,D_x) = (\Sigma_1^\theta D_{x_1} + \Sigma_2^\theta x_1) \otimes \indic_{L^2(\RR_{x_2})}$, where $\indic_{L^2(\RR_{x_2})}$ is the identity map on $L^2(\RR_{x_2})$. 
\end{rem}
\begin{rem}
  \label{rmk: perturbation}
 It follows that $\msU \msH^\theta \msU^{-1} = \sqrt{B}(\msG_0^\theta + \sqrt{h}\msW^W)$. When $B\to \infty$, we can interpret $\msG^\theta := \mathscr G_0^\theta + \sqrt{h}\msW^W$ as a small perturbation of $\mathscr G_0^\theta$.
\end{rem}
\begin{proof}
Recall that a symplectic transformation $(y,\eta) = \kappa(x,\xi)$ applying to a symbol $a(x,\xi) = a\circ \kappa^{-1}(y,\eta)\in S(\RR^4)$, induces a unitary operator $U_\kappa:L^2(\RR^2_x)\to L^2(\RR_y^2)$ s.t.
 \begin{equation}
  \label{eq: Symplectic_transformation}
  U_\kappa a^W(x,D_x) U_\kappa^{-1} = (a\circ\kappa^{-1})^{W}(y,D_y).
 \end{equation}
 
 By applying the following three symplectic transformations to $\msH^\theta(x,\xi)$:
\[
\begin{split}
  \kappa_1(x,\xi) &=  (x_1,\xi_2,\xi_1,-x_2), \ \kappa_2(x,\xi)= \Big(x_1 + \tfrac{x_2}{B}, x_2, \xi_1, \xi_2-\tfrac{x_1}{B}\Big),\\
  \kappa_3(x,\xi)&= \left(\sqrt{B}x_1, -\tfrac{x_2}{B}, \tfrac{\xi_1}{\sqrt{B}}, -B\xi_2\right),
\end{split}
\]
we find 
\begin{equation}
  \label{eq: symplectic_transformation_applied}
  \begin{cases}
    \msH^\theta_0\circ \kappa_1^{-1} \circ \kappa_2^{-1} \circ \kappa_3^{-1}(x,\xi)= \sqrt{B}(\Sigma_1^\theta \xi_1 + \Sigma_2^\theta x_1),\\
    \msV\circ \kappa_1^{-1}\circ\kappa_2^{-1}\circ\kappa_3^{-1}(x,\xi) = \msV(x_2+ h^{\frac{1}{2}}x_1,h\xi_2+h^{\frac{1}{2}}\xi_1).
  \end{cases}
\end{equation}
By \eqref{eq: Symplectic_transformation} and \eqref{eq: symplectic_transformation_applied}, the unitary operator $U_{\kappa} := U_{\kappa_3}\circ U_{\kappa_2}\circ U_{\kappa_1}$ has then the desired properties.
 \end{proof}
 
\medskip
\smallsection{Spectral Descriptions}
As mentioned in Remark \ref{rmk: perturbation}, we study the spectral properties of $\msG^\theta$ and $\msH^\theta$ by viewing them as perturbations of $\msG_0^\theta$ and $\msH_0^\theta$. Therefore, we start with $\msG_0^\theta$ and  $\msH_{0}^\theta$:

\begin{lemm}
  \label{lemma: eigenspace}
The spectral decompositions of $\mathscr G_0^\theta$ and $\msH_{0}^\theta $ are given by
\[
  \begin{split}
    &\Spec(\mathscr G_0^\theta) = \{\lambda_n := \sgn(n) \sqrt{2|n|}: n\in\ZZ\}\text{~with~eigenspace~}N_n^\theta,\\
    &\Spec(\msH_{0}^\theta) = \{\lambda_{n,B} := \sgn(n) \sqrt{2|n|B}: n\in\ZZ\} \text{~ with~eigenspace~} \msU N_n^\theta,
  \end{split}
\]
where 
$$ N_n^\theta = \operatorname{span}\left\{\begin{pmatrix}
 x\mapsto  u_n^\theta(x_1) s_1(x_2)\\ 0 
\end{pmatrix},\begin{pmatrix}
  0 \\x\mapsto u_n^{-\theta}(x_1)s_2(x_2)
\end{pmatrix}: \text{ For all } s_1,s_2 \in L^2(\RR_{x_2})\right\}.$$ Here $u_{n}^\theta = e^{-\frac{i\theta}{4}\sigma_3}u_n e^{\frac{i\theta}{4}\sigma_3}$, $u_{n} = C_n\begin{pmatrix}
    \sgn(n) r_{|n| - 1}\\ ir_{|n|}
    \end{pmatrix}$, $C_n = \begin{cases}
      \frac{1}{\sqrt{2}}, & n \in\ZZ\setminus\{0\} \\
      1, & n = 0
    \end{cases}$, as well as $r_{-1} = 0$, $r_m = C_m'(D_{x_1} + ix_1)^me^{-\frac{x_1^2}{2}}$ where $C_m'$ is constant s.t.  $\Vert r_m\Vert_{L^2(\RR_{x_1})} = 1$ for $m\in \NN$.
\end{lemm}
\begin{proof}
    The main observation here is for $G_D := \sigma_1 D_{x_1} +\sigma_2 x_1 = \begin{pmatrix}
      0 & a \\ a^* & 0
    \end{pmatrix}$ where $a = D_{x_1} - i x_1$, we have $[a,a^*] = 2$. Thus $a$ and $a^*$ form a pair of annihilator and creator. By the standard  argument for the ladder operators, there is a sequence of normalized $r_m(x_1) = C_m'\cdot (a^*)^m e^{-\frac{x_1^2}{2}} = C_m'(D_{x_1} + ix_1)^me^{-\frac{x_1^2}{2}}$, for $m \geq 0$ s.t. $ar_m = \sqrt{2m}r_{m-1}$ and $a^*r_m = \sqrt{2(m+1)}r_{m+1}$. Then one can check by computation and \eqref{eq:Unitary} that $u_n^\theta$, $N_n^\theta$ and $\msU N_n^\theta$ defined above are eigenvectors and eigenspaces of $\msG_0^\theta$, $\msG_0^\theta$ and $\msH_0^\theta$ with respect to eigenvalue $\lambda_n$, $\lambda_n$ and $\lambda_{n,B}$, for all $n\in\ZZ$.
\end{proof} 
\begin{rem}
\label{rmk: gap}
    Since $\msH^\theta = \msH_0^\theta + \msV$, thus $\operatorname{Spec}(\msH^\theta) \subset \overline{B_{\Vert \msV\Vert_\infty}\left(\Spec(\msH_0^\theta)\right)} = \bigcup\limits_n \overline{B_{\Vert \msV\Vert_\infty}(\lambda_{n,B})}.$ Fix $n$, since $\msV$ is bounded, when $B$ is large enough, $\left\{\overline{B_{\Vert \msV\Vert_\infty}(\lambda_{j,B})}\right\}_{|j - n|\leq 1}$ are disjoint. Since the DOS measure $\rho$ is supported on the spectrum, by \eqref{eq: def_of_DOS}, the regularized trace $\tilde{\Tr}(f(\msH^\theta))$ is not affected by modifying $f$ within the spectral gap $(\lambda_{k-1,B} + \Vert \msV\Vert_\infty, \lambda_{k,B} - \Vert \msV\Vert_\infty)$, i.e. $\tilde{\Tr}((\chi_{\Lambda_{k,B,\msV}}f)(\msH^\theta)) = \tilde{\Tr}((\chi_{B_{\Vert \msV\Vert_\infty}(\lambda_{k,B})}f)(\msH^\theta))$, for any $k\in\ZZ$. Thus we will start with $f$ supported on a fixed $\Lambda_{n,B,\msV}$ to avoid the influence of bands nearby and then consider the general case of $f$ supported on a fixed number of bands (see Theorem \ref{thm: chiral trace}, \ref{thm: anti-chiral trace} and their proofs in Subsection \ref{ss: proof_of_dos_expansion}).
\end{rem}
\begin{rem}
\label{rmk: shift_to_n}
 Both $\lambda_{n,B}$ and $\lambda_n$ are called Landau levels of $\msH_0^\theta$ and $\msG_0^\theta$ respectively. To study the corresponding operators near the Landau levels, we denote $\msH_n^\theta := \msH^\theta - \lambda_{n,B}$, $\msH_\on^\theta := \msH_0^\theta - \lambda_{n,B}$, $\msG_n^\theta := \msG^\theta - \lambda_n$ and  $\msG_\on^\theta := \msG_0^\theta - \lambda_n$. 
\end{rem}

\smallsection{Helffer-Sj\"ostrand formula and regularized traces} We proceed by recalling the Helffer-Sj\"ostrand formula. Let $K\in \NN$. Given $f\in C_c^{K+1}(\RR)$, we can always find  $\tilde f$, an order-$K$ quasi-analytic extension of $f$, by which we mean a function $\tilde f\in C^{K+1}_c(\CC)$, such that
\begin{equation}
  \label{eq: quasi-analytic_extension}
  \tilde f\vert_{\RR}=f, \ \text{and} \ \vert \partial_{\bar z}f\vert \leq C\Vert f \Vert_{C^{K+1}}\vert \Im z\vert^{K}, \ \text{for~some} \ C>0. 
\end{equation}
The concrete construction can be found in \cite[Sec. 4.1]{AJ06} or \cite[Theorem 8.1]{DS99}, where we can also choose $\tilde{f}$ s.t. $\supp(\tilde{f})\supset \supp(f)$ is arbitrarily close to $\supp(f)$. We omit the proof which can be found in the quoted references.
\begin{lemm}[Helffer-Sj\"ostrand formula]
\label{lemma: HS_formula}
  Let $H$ be a self-adjoint operator on a Hilbert space. Let $f\in C^{K+1}_c(\RR)$ and $\tilde f$ be its order-$K$ quasi-analytic extension, then 
  \begin{equation}
    \label{eq: HS_formula}
    f(H) = \frac{1}{2\pi i}\int_{\CC} \partial_{\bar{z}}\tilde{f}(z) (z - H)^{-1} dz \wedge d\bar{z}. 
  \end{equation}
\end{lemm}
In particular, for $f\in C_c^{K+1}(\Lambda_{n,B,\msV})$, define $f_0(x) = f(x +\lambda_{n,B})$ a function localized around zero. By Remark \ref{rmk: shift_to_n}, \eqref{eq:Unitary} and \eqref{eq: HS_formula}, we have
\begin{equation}
  \label{eq: HSresolvent}
  \begin{split}
    \msU f(\msH^{\theta}){\msU}^{-1} = \msU f_0(\msH_n^\theta) \msU &= -\frac{i}{2\pi}\int_{\mathbb C}\partial_{\bar z}\tilde f_0(z)(z-\msU\msH^{\theta}_{n}\msU)^{-1} ~dz\wedge d\bar{z}\\
    &= \frac{i\sqrt{h}}{2\pi}\int_{\mathbb C}\partial_{\bar z}\tilde f_0(z)(\msG_{n}^\theta - \sqrt{h}z)^{-1} ~dz\wedge d\bar{z}.
  \end{split}
\end{equation}
Thus to study $f(\msH^\theta)$, it is enough to study the resolvent $(\msG_{n}^\theta - \sqrt{h}z)^{-1}$.

\subsection{Second reduction: Grushin problem}
\label{subsec: Grushin problem} 
In this subsection, we apply the Schur complement formula twice for operators $\msG_\on^\theta$ and $\msG_n^\theta$ to characterize $(\msG_{n}^\theta-\sqrt{h}z)^{-1}$ using the effective Hamiltonian. In our context, the Schur complement formula is also called a Grushin problem and we shall use that terminology in the sequel. See \cite{SZ07} for more information on Grushin problem.

\smallsection{Unperturbed Grushin problem} To set up our Grushin problem, we introduce the space $B_{x_1}^k:=B^k(\RR_{x_1};\CC^4) := (1+D_{x_1}^2+x_1^2)^{-k/2}L^2(\RR_{x_1};\CC^4)$. Then \[\msG_{0,n}^\theta, \msG_n^\theta: \ B_{x_1}^{k+1}\otimes L^2(\RR_{x_2};\CC) \to B^k_{x_1} \otimes L^2(\RR_{x_2};\CC)\subset L^2(\RR_x^2;\CC^4)\] are bounded. Define
$R_n^+ = R_n^+(\theta):B^k_{x_1}\otimes  L^2(\RR_{x_2};\CC) \rightarrow L^2(\RR_{x_2};\CC^2)$ and $R_n^- = R_n^-(\theta):L^2(\RR_{x_2};\CC^2)\rightarrow B^k_{x_1}\otimes  L^2(\RR_{x_2};\CC)$ by
\begin{equation}
\label{eq:R_pm}
(R_n^+ t)(x_2)  = \int_{\RR} K_n^\theta(x_1)^* t(x_1,x_2) \ dx_1 \text{ and }R_{n}^- (s)(x) =K_n^\theta(x_1) s(x_2)
\end{equation}
with  
\begin{equation}
  \label{eq: Knt}
  K_n^\theta (x_1) = \begin{pmatrix} u_n^\theta(x_1) & 0 \\ 0 & u_n^{-\theta}(x_1) \end{pmatrix}_{4\times 2}.
\end{equation}
 Then $(R_n^+)^*=R_n^-$. 
 

First, we consider the Grushin problem for the unperturbed operator $\mathscr G_\on^\theta-\sqrt h z$:
\begin{lemm}[Unperturbed Grushin]
  \label{lemma:unperturbed Grushin}
  Fix $n\in \ZZ$. Let $R_n^{+}$ and $R_n^{-}$ be defined as \eqref{eq:R_pm}. Let
   \begin{equation}
   \mathcal P_\on = \mathcal P_\on(z;h,\theta) :=  \begin{pmatrix} \msG_\on^\theta-\sqrt{h}z &R_n^- \\R_n^+ & 0
   \end{pmatrix}.
  \end{equation}
Then $\mathcal P_\on$ is invertible iff $\sqrt{h}z\notin \{\lambda_m - \lambda_n : m\neq n\}$, and the inverse is
   \begin{equation}
   \label{eq: ME_on}
  \begin{split}
  \mathcal{E}_\on :=  (\mathcal P_\on)^{-1} =: \begin{pmatrix} E_{0,n} &E_{0,n,+} \\E_{0,n,-} & E_{0,n,{\pm}}
   \end{pmatrix}
  \end{split}
  \end{equation}
  where $E_{0,n,+} =R_n^-$, $E_{0,n,-}=R_n^{+}$, $E_{0,n,\pm}(z;h)= \sqrt{h}z \indic_{\CC^{2 \times 2}}$ and
  \begin{equation}
    \label{eq: ent}
    E_\on^\theta(z;h) = \sum_{m\neq n} \frac{K_m^\theta (K_m^\theta)^*}{\lambda_m - \lambda_n -\sqrt{h}z} = \sum_{m\neq n} \frac{\begin{pmatrix}
      u_m^\theta(u_m^\theta)^* & 0 \\ 0 & u_m^{-\theta}(u_m^{-\theta})^*
    \end{pmatrix}}{\lambda_m - \lambda_n -\sqrt{h}z}  =: \begin{pmatrix}
      \ent & 0 \\ 0 & \ennt
    \end{pmatrix}
  \end{equation}
  with $\lambda_n = \sgn(n)\sqrt{2|n|}, n\in\ZZ$. Furthermore, we have
  \[
  \begin{split}
      &E_{\on,-} (\msG_{0,n}^\theta - \sqrt{h}z) E_{\on, +} = -E_{\on,\pm} \quad \text{and}\\
      & (\msG_\on^\theta-\sqrt{h}z)^{-1} = E_\on - E_{\on,+} (E_{\on,\pm})^{-1} E_{\on,-}.
  \end{split}
     \]
\end{lemm}
\begin{rem}
One can verify that $E_{0,n}^\theta$ maps $N_n^\theta$ to $0$ and $N_m^\theta$ to $\frac{N_m^\theta}{\lambda_m - \lambda_n - \sqrt{h}z}$ if $m\neq n$.
\end{rem}
\smallsection{Perturbed Grushin problem}
Next, we consider the perturbed Grushin problem for $\mathscr G_{n}^\theta - \sqrt{h}z$.
\begin{lemm}[Perturbed Grushin]
  \label{lemm:perturbed Grushin}
 Let $R_n^{\pm}$, $\msW^W$ be defined as \eqref{eq:R_pm}, \eqref{eq: Unitary_transform_for_multiplication}. Let
  \begin{equation}
    \mathcal P_{n} = \mathcal P_{n}(z;h,\theta) :=  \begin{pmatrix}\msG_n^\theta -\sqrt{h}z &R_n^- \\R_n^+ & 0
  \end{pmatrix} = \mathcal{P}_\on + \sqrt{h} \mathcal{W}
  \end{equation}
  where $ \mathcal W:=  \operatorname{diag}(\msW^W_{4 \times 4}, 0_{2\times 2}).$ Fix $n\in \ZZ$, there exist $h_0 = \min\left\{ \frac{1}{2\Vert \msW\Vert_\infty}, \frac{\lambda_{|n|+1} - \lambda_{|n|}}{4\Vert \msW\Vert_\infty}\right\}$, s.t. for all $ h\in [0,h_0)$, $\mathcal{P}_{n}$ is invertible with inverse
   \begin{equation}
  \begin{split}
  \label{eq:Grushin}
   \mathcal E_{n}:= (\mathcal{P}_{n})^{-1} =: \begin{pmatrix} E_{n} & \Ep \\ \Em &  \Epm
   \end{pmatrix}
  \end{split}
  \end{equation}
which is analytic in $|z|\leq 2\Vert \msW\Vert_\infty$.  $\Epm(z):L^2(\RR_{x_2};\CC^2)\to L^2(\RR_{x_2};\CC^2)$ is called the effective Hamiltonian and satisfy
\begin{equation}
    \label{eq: epm}
        \Epm(z) = \sqrt{h} \left(z - R_n^+ \msW^W (\indic + \sqrt{h}E_{0,n}\msW^W)^{-1} R_n^-\right)=: \sqrt{h}(z - Z^W).
\end{equation}
In addition, we have 
  \begin{align}
    \label{eq: linear_in_z}
      &\Em (\mathscr G_{n}^\theta - \sqrt{h}z)\Ep  = -\Epm \Rightarrow \sqrt{h}\Em  \Ep  = \partial_z \Epm,\\
      \label{eq:Reso Id}
        &(\mathscr G_{n}^\theta - \sqrt{h}z)^{-1} = E_{n}- \Ep  E_{n,{\pm}}^{-1} \Em.
  \end{align}
  \end{lemm}
  \begin{proof}
    Let $h_0$ be defined as above. When $h\in [0,h_0)$, $|z|<2\Vert \msW\Vert_\infty$, we have 
    \[
    \begin{cases}
      \sqrt{h}z \notin \{\lambda_m - \lambda_n:m \neq n\} &\Rightarrow \mathcal P_\on \text{~ is ~ invertible ~ with~} \Vert \mathcal P_\on \Vert\geq 1.\\
      \sqrt{h}\Vert \mathcal W\Vert_\infty \leq \frac{1}{2} &\Rightarrow \mathcal P_n = \mathcal P_\on + \sqrt{h} \mathcal W \text{~ is ~ invertible~ with~inverse~} \mathcal E_n.\\
      |\sqrt{h}z|\leq \frac{\lambda_{|n|+1} - \lambda_{|n|}}{2} &\Rightarrow \text{~by~} \eqref{eq: ent} \text{~and~} \eqref{eq: ME_on}, E_\on(z) \text{~and~} \mathcal E_\on(z) \text{~are~ analytic}.
    \end{cases}
    \]
Furthermore, 
\[
   \mathcal E_n := \mathcal{P}_{n}^{-1}  = (I + \sqrt{h}\mathcal P_\on^{-1}\mathcal W)^{-1}\mathcal P_\on^{-1} =  \sum_{j=0}^{\infty} (-1)^j h^{j/2}\left(\mathcal{E}_\on \mathcal W \right)^j \mathcal{E}_\on.\]
In particular, we get from the $(2,2)$-block of $\mathcal P_n^{-1}$ that
 \begin{equation}
 \begin{split}
 \Epm(z) &= E_{0,n,\pm}(z) + \sum_{j=1}^{\infty} (-1)^j h^{j/2} E_{0,n,-} \msW^W (E_\on \msW^W)^{j-1} E_{0,n,+}\\
 &= \sqrt{h}z - \sqrt{h} R_n^+ \msW^W (\indic + \sqrt{h} E_\on\msW^W)^{-1} R_n^-.
 \end{split}
 \end{equation}
 In fact, by direct computation, one get that $E_\on$, $E_{n,+}$ and $E_{n,-}$ can all be represented by entries of $\mathcal E_\on$ which we proved are analytic, thus $\mathcal E_n(z)$ is also analytic. 
 
 In the end, \eqref{eq: linear_in_z} and \eqref{eq:Reso Id} follows from $\mathcal E_n \mathcal P_n \mathcal E_n = \mathcal P_n$ and the diagonalization on $\mathcal P_n$.
  \end{proof}
\subsection{Properties of effective Hamiltonian}
\label{ss: properties_of_effective_Hamiltonian}

In this subsection, we proceed with our study of $\Epm(z)$, $\Epm^{-1}(z)$ and $\partial_z \Epm \circ \Epm^{-1}$, with their symbols denoted by $\Epm(x_2,\xi_2;z,h)$, $\Epm^{-1}(x_2,\xi_2;z,h)$ and $r_n(x_2,\xi_2;z,h) := \partial_z \Epm \# \Epm^{-1}(x_2,\xi_2;z,h)$. Apart from analyzing boundedness and asymptotic expansions of symbols, we are especially interested in understanding the $z$-dependence and $z$ vs. $h$ competition of the symbols.

Before starting to analyze these properties, we introduce a key concept of this section: the operator-valued symbol and its quantization.

\smallsection{Operator-valued symbol}
Let $b^w(x_2,\xi_2;x_1,D_{x_1})\in S\left(\RR_{x_2,\xi_2}^2;\mathcal{L}(B_{x_1}^{k+1};B_{x_1}^k)\right)$, which we shall call an operator(-in-$(x_1,D_{x_1})$)-valued symbol (in $(x_2,\xi_2)$), then its $(x_2,hD_{x_2})$-Weyl quantization is defined as  $b^W(x_2,hD_{x_2};x_1,D_{x_1}): L^2(\RR_{x_2};B_{x_1}^{k+1}) \to L^2(\RR_{x_2};B_{x_1}^k)$ such that
\[
   \left( b^W(x_2,hD_{x_2};x_1,D_{x_1})u\right)(x) = \int e^{\frac{i(x_2 - y_2) \xi_2}{h}} \left(b^w\left(\frac{x_2+y_2}{2},\xi_2;x_1,D_{x_1}\right)u\right)(x_1;y_2) \frac{dy_2d\xi_2}{2\pi h}.
\]
In particular, if we have a symbol $a\in S(\RR_{x,\xi}^4)$, and we view $(x_2,\xi_2)$ as parameters and consider the $(x_1,D_{x_1})$-Weyl quantization of it,  we get $a^w(x,D_{x_1},\xi_2)$ which is an operator-valued symbol in $(x_2,\xi_2)$ (the superscript $w$ represent the $(x_1,D_{x_1})$-Weyl quantization). If we do a further $(x_2,hD_{x_2})$-Weyl quantization of $a^w(x,D_{x_1},\xi_2)$, then we get the $(1,h)$-Weyl quantization defined in \eqref{eq: quantization}. 

\begin{rem}
    \label{rmk: operator_valued_symbol_example}
    For the rest of this section, given an operator, e.g. $\msG_0^\theta$, $\Epm$ and $\msW^W$ in \eqref{eq:Unitary}, \eqref{eq:Grushin} and \eqref{eq: Unitary_transform_for_multiplication}, instead of viewing them as the $(1,h)$-Weyl quantization of the scalar-valued symbol in $S(\RR_{x,\xi}^4)$, we will view them as the $(x_2,hD_{x_2})$-Weyl quantization of the operator-valued symbol in $S\left(\RR_{x_2,\xi_2}^2;\mathcal{L}(B_{x_1}^{k_1};B_{x_1}^{k_2})\right)$, for appropriate $k_1,k_2\in\ZZ$. 
    
    In particular, since $\msG_0^\theta$ only depends on $(x_1,D_{x_1})$, $\Epm$ only depends on $(x_2,hD_{x_2})$, $\msW^W(x,D_{x})$ is the $(1,h)$-Weyl quantization of the symbol $\msV(x_2+\sqrt{h}x_1,h\xi_2-\sqrt{h}\xi_1)$, we see that the operator-valued symbol of $\msG_0^\theta$, $\Epm$ and $\msW^W$ are respectively 
    \[\Sigma_1^\theta x_1+ \Sigma_2^\theta D_{x_1}, \  \Epm(x_2,\xi_2;z,h), \text{~and~} \tilde{\msV}^w(x,D_{x_1},\xi_2) :=  \msV^w(x_2+\sqrt{h}x_1,\xi_2 - \sqrt{h}D_{x_1})\]
    where $\tilde{f}(x,\xi) = f(x_2 + \sqrt{h}x_1,\xi_2 - \sqrt{h} \xi_1)$. And since now
    \begin{equation}
        \label{eq: W_to_tilde_V}
        \msU \msV(x) \msU^{-1} =  \msW^W(x,D_x) = \tilde{\msV}^W(x,D_{x_1},hD_{x_2}),
    \end{equation}
    we will use $\tilde\msV^W$ to replace $\msW^W$ in Lemma 
  \ref{lemma: Unitary} and \ref{lemm:perturbed Grushin}. Finally, we mention that the proof of Lemma 
  \ref{lemma: Unitary} implies in general 
  \begin{equation}
      \label{eq: f_unitary_to_tilde_f}
      \msU f(x) \msU^{-1} = \tilde{f}^W(x,D_{x_1},hD_{x_2}).
  \end{equation}
\end{rem}

\smallsection{Boundedness with $z$ dependence} We now study the boundedness of the operator-valued symbol $\Epm$, $\Epm^{-1}$ and $r_n$ as well as the $z$ dependence of them. 

Notice that since $\Epm$ only depends on $(x_2,hD_{x_2})$, when viewed as a $(x_2,hD_{x_2})$-Weyl quantization, its operator-in-$(x_1,D_{x_1})$-valued symbol coincides with its $\CC_{2\times 2}$-valued symbol. For convenience, we write $S_\delta^k(\RR_{x_2,\xi_2}^2;\CC_{2\times 2})$ as $S_\delta^k$ and omit the ``$0$'' in $\delta$ and $k$.

\begin{lemm}[Boundedness]
  \label{lemma: Properties of r} 
  Let $h_0$, $\Epm$ be as in Lemma \ref{lemm:perturbed Grushin}. Then for all $ h\in [0,h_0)$, we have the symbol of $\Epm$,  $\Epm(x_2,\xi_2;z,h)$, belongs to $S^{-\frac{1}{2}}$ uniformly in $|z|\leq 2\Vert \msV\Vert_\infty$, i.e. for any $\alpha, \beta>0$, there is $C_{\alpha, \beta,n} = C_{\alpha,\beta, n}(\Vert \msV\Vert_\infty)$,  s.t.
  \[
  \sup\limits_{(x_2,\xi_2)\in\RR^2}\Vert\partial_{x_2}^\alpha\partial_{\xi_2}^\beta \Epm(x_2,\xi_2;z,h)\Vert_{\CC_{2\times 2}}\leq C_{\alpha,\beta,n}\sqrt{h}, \quad \text{ for all } |z|\leq 2\Vert \msV\Vert_\infty.
  \]
  Furthermore, if $|\Im z|\neq 0$, then we  also have that for all $ h\in [0,h_0)$, $|z|\leq 2\Vert \msV\Vert_\infty$, $\alpha, \beta>0$,
    \begin{align}
    \label{eq: bound on Epm inv}
     &\Vert\partial_{x_2}^\alpha\partial_{\xi_2}^\beta \Epm^{-1}(x_2,\xi_2;z,h)\Vert_{\CC_{2\times 2}}\leq C_{\alpha,\beta,n}\max\left(1, \frac{h^{3/2}}{|\Im z|^3}\right) h^{-\frac{1}{2}} |\Im z|^{-(|\alpha|+|\beta|)-1},\\
    \label{eq: bound on r}
    &\Vert\partial_{x_2}^\alpha\partial_{\xi_2}^\beta r_n(x_2,\xi_2;z,h)\Vert_{\CC_{2\times 2}}\leq C_{\alpha,\beta,n}\max\left(1, \frac{h^{3/2}}{|\Im z|^3}\right) |\Im z|^{-(|\alpha|+|\beta|)-1}.
   \end{align}
In particular, if $0<\delta <\frac{1}{2}$ and $|\Im z|\geq h^\delta$, then $\Epm^{-1} \in S_\delta^{\frac{1}{2}+\delta}$ and $r_n\in S_\delta^\delta$.
\end{lemm}
  \begin{proof}
  When $h\in [0,h_0)$, $|z|\leq 2\Vert \msV\Vert_\infty$, $\Epm$ is a $\Psi$DO because $\mathcal P_n$ is. In fact, by checking term by term, we have the operator-valued symbol $\mathcal P_n(x,D_{x_1},\xi_2)\in S(\RR_{x_2,\xi_2}^2;\mathcal{L}(B_{x_1}^{k+1} \times \CC^2;B_{x_1}^k\times \CC^2))$. By invertibility and Beal's lemma,  $\mathcal E_n(x_2,\xi_2;z,h)\in S(\RR_{x_2,\xi_2}^2;\mathcal{L}(B_{x_1}^{k+1} \times \CC^2;B_{x_1}^k\times \CC^2)).$ In particular, we have 
  \begin{equation}
      \label{eq: everything_in_S}
          \begin{cases}
        R_n^+\in S(\RR_{x_2,\xi_2}^2;\mathcal{L}(B_{x_1}^k;\CC^2)), \ R_n^- \in S(\RR_{x_2,\xi_2}^2;\mathcal{L}(\CC^2;B_{x_1}^{k+1})),\\
        E_\on \in S(\RR_{x_2,\xi_2}^2;\mathcal{L}(B_{x_1}^{k} ;B_{x_1}^{k+1})), \ \tilde{\msV}^w\in S(\RR_{x_2,\xi_2}^2;\mathcal{L}(B_{x_1}^{k+1} ;B_{x_1}^k)).
    \end{cases}
  \end{equation}
Furthermore, by \eqref{eq: ent}, when $|\sqrt{h}z|\leq \frac{\lambda_{|n|+1} - \lambda_{|n|}}{2}$, $E_\on$ is uniformly bounded. Thus $\Epm, \partial_z \Epm \in S^{-\frac{1}{2}}$ uniformly.
  
  Then we consider $\Epm^{-1}$ and $r_n$. Let $l_1$, $l_2$, $\cdots$ be linear forms on $\RR_{x_2,\xi_2}^{2}$. Let $L_j = l_j(x_2,hD_{x_2})$. Since $\Epm\circ \Epm^{-1} = I$, we get
    \[
      \operatorname{ad}_{L_j}\Epm^{-1} = -\Epm^{-1} \circ \operatorname{ad}_{L_j}\Epm\circ \Epm^{-1}, 
    \]
    where $\operatorname{ad}_{L_j}A = [L_j, A]$. Since $\operatorname{ad}_{L_j}(A \circ B) = (\operatorname{ad}_{L_j}A)\circ B + A \circ \operatorname{ad}_{L_j}B$, thus 
    \[
      \operatorname{ad}_{L_j}(\partial_z \Epm\circ\Epm^{-1}) = -\partial_z\Epm \circ  \Epm^{-1}\circ \operatorname{ad}_{L_j}\Epm \circ \Epm^{-1} +  \operatorname{ad}_{L_j}\partial_z\Epm \circ \Epm^{-1}.
    \]
By \eqref{eq:Reso Id}, $\Vert \sqrt{h} \Epm^{-1}\Vert_{\CC_{2\times 2}}  = \mathcal O(|\Im z|^{-1})$. Recall that $\Epm$, $\partial_z \Epm\in S^{-\frac{1}{2}}$, thus 
\[
      \begin{split}
        &\Vert  \operatorname{ad}_{L_j}(\sqrt{h}\Epm^{-1})\Vert_{\CC_{2\times 2}} = \mathcal O\left(\frac{h}{|\Im z|^2}\right)\text{ and }\Vert  \operatorname{ad}_{L_j}(\partial_z\Epm \circ \Epm^{-1})\Vert_{\CC_{2\times 2}} = \mathcal O\left(\frac{h}{|\Im z|^2}\right).
      \end{split}
    \]
    By induction, 
    \[
      \begin{split}
      &\Vert  \operatorname{ad}_{L_1} \circ \cdots \circ  \operatorname{ad}_{L_N}(\sqrt{h} \Epm^{-1})\Vert_{\CC_{2\times 2}} = \mathcal O\left(\tfrac{h^N}{|\Im z|^{N+1}}\right)\\
      &\Vert  \operatorname{ad}_{L_1} \circ \cdots \circ  \operatorname{ad}_{L_N}(\partial_z\Epm \circ \Epm^{-1})\Vert_{\CC_{2\times 2}} = \mathcal O\left(\tfrac{h^N}{|\Im z|^{N+1}}\right).
      \end{split}
    \]
    By a parametrized version of Beal's lemma, \cite[Prop. 8.4]{DS99}, we get \eqref{eq: bound on Epm inv} and \eqref{eq: bound on r}. 
  \end{proof}
  
   \smallsection{Asymptotic Expansion with $z$ dependence} We proceed by discussing the asymptotic expansion of $\Epm$, $\Epm^{-1}$ and $r_n$. Again, we are concerned with $z$-dependence of each term in the asymptotic expansions. In order to focus on the main points, we outsource further details concerning the asymptotic expansion of $\Epm$ and $\Epm^{-1}$, c.f. Prop. \ref{prop: z_dependence_each_term}, and its proof in the Appendix \ref{sec:appendix1}, and present a shorter version here that only summarizes the results that we eventually need in the sequel.

\begin{lemm}[Asymptotic expansion]
  \label{lemma: explicit_asymptotic_expansion_needed}
  Let $h_0$, $\Epm$ be as in Lemma \ref{lemm:perturbed Grushin},   $0<\delta<1/2$. If $h\in [0,h_0)$, $|z|\leq 2\Vert \msV\Vert_\infty$, $|\Im z|\geq h^\delta$, then $r_n(x_2,\xi_2;z,h) = \partial_z \Epm \# \Epm^{-1}$ has an asymptotic expansion in $S_\delta^\delta$: 
  \begin{equation}
      \label{eq: r_n_on_z}
      \begin{split}
        &\quad r_n(x_2,\xi_2;z,h)\sim \sum\limits_{j =0}^\infty h^{\frac{j}{2}} r_{n,j}(x_2,\xi_2;z), \text{~with~} h^{\frac{j}{2}}r_{n,j} \in S_\delta^{(j+1)\delta - \frac{j}{2}}.
        \end{split}
  \end{equation}
  More specifically, there are $d_{n,j,k,l}(x_2,\xi_2;z)$, $e_{n,j,k,\alpha}(x_2,\xi_2)\in S$ s.t.
  \begin{equation}
      \label{eq: r_n_z_dependence}
        r_{n,j} = \sum\limits_{k = 0}^j(z - z_{n,0})^{-1} \prod\limits_{l = 0}^k \left[d_{n,j,k,l}(x_2,\xi_2;z)(z - z_{n,0})^{-1}\right],
  \end{equation}
  with $\prod\limits_{l = 0}^k d_{n,j,k,l}(x_2,\xi_2;z) = \sum\limits_{\alpha = 0}^{j+k -2} z^\alpha e_{n,j,k,\alpha}(x_2,\xi_2)$ and $z_{n,0}$ given in Prop.  \ref{prop: expofQk}. Let $R_{n,J} := r_n - \sum\limits_{j = 0}^{J - 1} h^{\frac{j}{2}}r_{n,j}$, then $R_{n,J} \in S_\delta^{(J+1)\delta - \frac{J}{2}}$, i.e. for all  $ \alpha, \beta>0$, there is $C_{\alpha, \beta,n}'$ s.t. 
  \begin{equation}
      \label{eq: R_J_symbol_class}
     \sup\limits_{(x_2,\xi_2)\in\RR^2} |\partial_{x_2}^\alpha\partial_{\xi_2}^\beta R_{n,J} |\leq C_{\alpha, \beta,n}' h^{\frac{J}{2} - (J+1)\delta -  \delta(|\alpha| + |\beta|)}.
  \end{equation}Furthermore, for the expansion of $\Tr_{\CC^2}(r_n)$, we have for $\eta = x_2 + i\xi_2$,
  \begin{equation}
    \label{eq: explicite_expansion_of_trace}
    \begin{split}
      \operatorname{Chiral} \ \msH_{\ch, n} (J = 3) : & \Tr_{\CC^2}(r_{\ch, n,0} + h^{\frac{1}{2}}r_{\ch, n,1} + hr_{\ch,n,2}) = \frac{2}{z} + 0 + \frac{\lambda_n^2}{z^3}\mathfrak U(\eta)h, \\
      \operatorname{Anti-Chiral} \ \msH_{\ach,n}^\theta (J = 2): & \Tr_{\CC^2}(r_{\ach, n,0} + h^{\frac{1}{2}}r_{\ach, n,1}) = \frac{2z}{z^2 - c_n^2} + \frac{2s_n^2(z^2+c_n^2)}{(z^2 - c_n^2)^2}\sqrt{h} ,
    \end{split}  
  \end{equation}
  where $\mathfrak U(\eta) = \frac{\alpha_1^2}{8}\left[\alpha_1^2(|U_-(\eta)|^2 - |U(\eta)|^2)^2 + 4|\partial_{\bar{\eta}}\overline{U_-(\eta)}-\partial_\eta U(\eta)|^2\right]$, $\partial_\eta = \frac{1}{2}(\partial_{x_2} - i\partial_{\xi_2})$, $s_n(\eta) = \begin{cases}
    \alpha_0  \sin(\thot) |V(\eta)| & n\neq 0 \\
    \alpha_0 |V(\eta)| & n = 0,
  \end{cases}$ and $c_n(\eta) = \begin{cases}
    \alpha_0 \cos(\thot)|V(\eta)| & n \neq 0\\
    \alpha_0 |V(\eta)| & n = 0.
  \end{cases}$
\end{lemm}
\begin{rem}
\label{rmk: higher_order_terms_on_z}
Notice by Prop. \ref{prop: expofQk}, $z_{n,0} = 0$ for the chiral model. Thus we have $r_{\ch,n,j} = \sum\limits_{k = 0}^{2(j - 1)} z^{k - j - 1} f_{n,j,k}(x_2,\xi_2)$ for appropriate $f_{n,j,k}\in S$ when $j\geq 1$. 
\end{rem}

\subsection{Trace formula}
\label{subsec: pfofdos}
Now we are ready to characterize $\tilde{\Tr}f(\msH^\theta))$ using $\Epm$ and still use the operator-valued symbol and $(x_2,hD_{x_2})$-quantization in this subsection.
\begin{lemm}
  \label{lemma: d.o.s.computable}
  Let $\Epm$ be as in Lemma \ref{lemm:perturbed Grushin}. Let $f\in C_c^{K+1}(\Lambda_{n,B,\msV})$ and $f_0(x) = f(x +\lambda_{n,B})$ be as in   \eqref{eq: HSresolvent}. Then the regularized trace $\tilde{\Tr}(f(\msH^\theta))$ satisfies
  \begin{equation}
      \label{eq: Trace_formula}
       \tilde{\Tr}(f(\msH^{\theta})) = -\frac{i}{4\pi^2 h|E|} \int_{\CC} \int_{E}\partial_{\bar{z}} \tilde{f_0} \Tr_{\CC^2}(r_n(x_2,\xi_2;z,h)) \ dx_2 \ d\xi_2 ~dz\wedge d\bar{z},
  \end{equation}
\end{lemm}
Lemmas needed for the following proof are outsourced to  Appendix \ref{appendix2}.
\begin{proof}
   By \eqref{eq: HSresolvent}, \eqref{eq:Reso Id}, and the analyticity of $E_{n}(z)$ when $h\in [0,h_0)$, $|z|\leq 2\Vert \msV\Vert_\infty$,
  \[
    \begin{split}
      \msU f(\msH^{\theta}){\msU}^{-1} = -\frac{i\sqrt{h}}{2\pi}\int_\CC \partial_{\bar{z}} \tilde{f}_0 (\Ep  \Epm^{-1}\Em )(z) \ dz\wedge d\bar{z}.
    \end{split}
  \]
  Thus we have 
  \[
    \begin{split}
      \tilde{\Tr}f(\msH^\theta) &= \lim\limits_{R\to \infty} \frac{1}{4R^2} \Tr_1\left(\indic_Rf(\msH^\theta)\indic_R\right) = \lim\limits_{R\to\infty} \frac{1}{4R^2} \Tr_1(\tcRW \msU f(\msH^\theta)\msU^{-1}\tcRW)\\
      &\stackrel{\text{2}}{=} \lim\limits_{R\to\infty} -\frac{i\sqrt{h}}{8\pi R^2} \Tr_1 \left(\int_\CC \partial_{\bar{z}}\tilde{f}_0(\tcRW\Ep \Epm^{-1}\Em \tcRW) ~dz\wedge d\bar{z}\right)   \\
      &\stackrel{\text{3}}{=} \lim\limits_{R\to\infty} -\frac{i\sqrt{h}}{8\pi R^2} \int_\CC \partial_{\bar{z}}\tilde{f}_0\Tr_1 \left(\tcRW\Ep \Epm^{-1}\Em \tcRW\right) ~dz\wedge d\bar{z}\\
      &\stackrel{\text{4}}{=} \lim\limits_{R\to\infty} -\frac{i\sqrt{h}}{8\pi R^2} \int_\CC \partial_{\bar{z}}\tilde{f}_0\Tr_2\left(\bcRW \Em \Ep \Epm^{-1}\bcRW \right) ~dz\wedge d\bar{z}\\
      &\stackrel{\text{5}}{=} \lim\limits_{R\to\infty} -\frac{i}{8\pi R^2} \int_\CC \partial_{\bar{z}}\tilde{f}_0\Tr_2\left(\bcRW \partial_z\Epm\Epm^{-1}\bcRW\right) ~dz\wedge d\bar{z}\\
      &\stackrel{\text{6}}{=}\lim\limits_{R\to\infty} -\frac{i}{16\pi^2 h R^2} \int_\CC \partial_{\bar{z}}\tilde{f}_0\int_{\RR^2}\Tr_{\CC^2}\left(\bar \indic_R \#\partial_z \Epm\#\Epm^{-1}\#\bar \indic_R \right)~d{x_2}~d{\xi_2} ~dz\wedge d\bar{z}\\
      &\stackrel{\text{7}}{=} -\frac{i}{4\pi^2 h|E|} \int_\CC \int_{E} \partial_{\bar{z}}\tilde{f}_0  \Tr_{\CC^2}\left(\partial_z \Epm\#\Epm^{-1}\right) ~d{x_2} ~d{\xi_2}  ~dz\wedge d\bar{z}\\
    \end{split}
  \]
where $\msU \indic_R \msU^{-1} =: \tcRW$ follows from \eqref{eq: f_unitary_to_tilde_f}. Here, $\bcRW = \bcRW(x_2,hD_{x_2})$ where $\bar{\indic}_R(x_2,\xi_2)$ coincides with $\indic_R(x_1,x_2)$ but is viewed as a function of phase space variables $(x_2,\xi_2)$ rather than $x$. In addition, $\Tr_1 := \Tr_{L^2(\RR^2_{x};\CC^4)}$, $\Tr_2 := \Tr_{L^2(\RR_{x_2};\CC^2)}$.

The second line follows from the Helffer-Sj\"ostrand formula in Lemma \ref{lemma: HS_formula}. The third line follows from Lemma \ref{lemma: trace class}, where we proved $\cRW \Ep  \Epm^{-1} \Em \cRW$ is trace class. The fourth line follows directly from Lemma \ref{lemma: diffoftrace}. The fifth line follows from \eqref{eq:Reso Id}. The sixth line follows from 
\begin{equation}
  \label{eq: trace_QuantitoSymbol}
  \Tr_{\mathcal{L}(L^2(\RR_{x_2};H_1);L^2(\RR_{x_2};H_2))}(a^W(x_2,hD_{x_2})) = \frac{1}{2\pi h} \int_{\RR^2_{x_2,\xi_2}} \Tr_{\mathcal{L}(H_1,H_2)}(a(x_2,\xi_2)) d{x_2}d{\xi_2}.
\end{equation}
The seventh line follows from periodicity of $\msV$ and thus periodicity of $\partial_z \Epm \# \Epm^{-1}$, which follows immediately from its asymptotic expansion. 
\end{proof}

\subsection{Proof of main results}
\label{ss: proof_of_dos_expansion}
Now we can prove our main Theorems \ref{thm: chiral trace} and \ref{thm: anti-chiral trace}:

\begin{proof}[Proof of Theo. \ref{thm: chiral trace}, \ref{thm: anti-chiral trace}]
Let $0<\delta <1/2$. Assume $f\in C_c^{N+1}(\Lambda_{n,B,\msV})$. Let $f_0(x) := f(x+\lambda_{n,B})$ which is supported on a neighbourhood of $0$.  Recall by Lemma \ref{lemma: d.o.s.computable}, we need to compute
  \begin{equation}
      \label{eq: regularized trace}
       \begin{split}
      \tilde{\Tr}(f(\msH_n^\theta)) = -\frac{i}{4\pi^2 h|E|} \int_{E} \int_{\CC} \partial_{\bar{z}} \tilde{f}_0\Tr_{\CC^2}(r_n(x_2,\xi_2;z,h)) ~dz\wedge d\bar{z} \ dx_2 \ d\xi_2.
    \end{split}
  \end{equation}
We can rewrite the integral with expansion $r_n = \sum_{j=0}^{J-1} h^{j/2} r_{n,j} + R_{n,J}$ as in Lemma \ref{lemma: explicit_asymptotic_expansion_needed} 
\[
  \begin{split}
    \left[\int_{\CC} \partial_{\bar{z}}\tilde{f}_0 \Tr_{\CC^2} (r_n) dz \wedge d\bar{z}\right](x_2,\xi_2;h) =& \int_{\CC} \partial_{\bar{z}}\tilde{f}_0 \sum\limits_{j = 0}^{J-1} h^{\frac{j}{2}}\Tr_{\CC^2}(r_{n,j})dz\wedge d\bar{z}\\
    +&\int_{|\Im z|\geq h^\delta} \partial_{\bar{z}}\tilde{f}_0 \Tr_{\CC^2}(R_{n,J}) dz\wedge d\bar{z}\\
    +& \int_{|\Im z|\leq h^\delta} \partial_{\bar{z}}\tilde{f}_0 \Tr_{\CC^2}(R_{n,J}) dz\wedge d\bar{z}
    := A_{1} + A_{2} + A_{3}.
  \end{split}
\]
Notice that by Remark \ref{rmk: gap}, we only need to consider $f_0$ supported at $|z|\leq \Vert \msV\Vert$, for which we can pick $\tilde{f}_0$ s.t. $\tilde{f}_0$ is supported inside $|z|\leq 2\Vert \msV\Vert_\infty$ for the integral. As in Lemma \ref{lemma: explicit_asymptotic_expansion_needed}, we take $J = 3$ in the chiral case and $J = 2$ in the anti-chiral case.

First of all, we compute $A_1$ by \eqref{eq: explicite_expansion_of_trace} and the general version of Cauchy's integral formula, see \cite[(3.1.11)]{Ho03}: Let $X $ be an open subset of $\CC$. Let $g\in C_c^m(X)$, with $m \ge n$ then 
\begin{equation}
    \label{eq: Cauchy_integral}
    2\pi ig^{(n)}(\zeta) = \int_X \partial_{\bar{z}}g(z) \frac{n!}{(z - \zeta)^{n+1}} dz \wedge d\bar{z}.
\end{equation}
In particular, take $X$ to be a small open neighborhood of $\supp(\tilde{f}_0)$. By \eqref{eq: Cauchy_integral} and the definition of $f_0$, we have 
\[
  \begin{split}
    A_{1,\ch} &= \int_{\CC} \partial_{\bar{z}} \tilde{f}_0 \left[\frac{2}{z} +  \frac{\lambda_n^2}{z^3}\mathfrak U(\eta) h \right] dz \wedge d\bar{z}= 2\pi i \left[2f(\lambda_{n,B}) +  \frac{\lambda_n^2}{2}\mathfrak U(\eta) f''(\lambda_{n,B})h \right],\\
    A_{1,\ach} &= \int_{\CC} \partial_{\bar{z}} \tilde{f}_0 \Tr_{\CC^2}\left[ \frac{1}{z - c_n} + \frac{1}{z+c_n} + \frac{s_n^2\sqrt{h}}{(z -c_n)^2} + \frac{s_n^2\sqrt{h}}{(z+c_n)^2}\right] dz\wedge d\bar{z}  \\
    &= 2\pi i\left[f(\lambda_{n,B} + c_n) + f(\lambda_{n,B}-c_n) + f'(\lambda_{n,B} +c_n)s_n^2\sqrt{h}+ f'(\lambda_{n,B}-c_n)s_n^2\sqrt{h}\right].
  \end{split}
\]
For $A_2$, by \eqref{eq: R_J_symbol_class} and $|z|\leq 2\Vert \msV\Vert_\infty$, when $|\Im z |\geq h^\delta$, there are $C_{n},C_n'>0$ such that 
\[
  \begin{split}
  &|A_{2}| \leq \int_{|\Im z |\geq h^\delta} |\partial_{\bar{z}} \tilde{f}_0 | C_{n}  h^{\frac{J}{2} - (J+1)\delta} 2 L(dz)\leq C_n' \Vert f \Vert_{C^{K+1}} \Vert \mathscr V \Vert_{\infty} h^{\frac{J}{2} - (J+1)\delta}.
  \end{split}
\]
Finally, by \eqref{eq: quasi-analytic_extension}, \eqref{eq: bound on r}, \eqref{eq: r_n_z_dependence}, $0<\delta <1/2$ and $|z|\leq 2\Vert \msV\Vert_\infty$, we have for some $C_{n,j},C_n''$
\begin{equation}
    \label{eq: estimates_A_1}
    \begin{split}
    |A_{3}| &\leq \int_{|\Im z|\leq h^\delta} |\partial_{\bar{z}} \tilde{f}_0| \left[|\Tr_{\CC^2}(r_n)| + \sum\limits_{j = 0}^{J-1} \left|\Tr_{\CC^2}(h^{\frac{j}{2}} r_{n,j})\right|\right] dz \wedge d\bar{z}\\
    &\leq \int_{|\Im z|\leq h^\delta} \Vert f \Vert_{C^{K+1}}|\Im z|^K \left[\max\left(\frac{1}{|\Im z|}, \frac{h^{\frac{3}{2}}}{|\Im z|^4}\right) + \sum\limits_{j = 0}^{J-1} \frac{C_{n,j}h^{\frac{j}{2}}}{|\Im z|^{j+1}}\right] dz\wedge d\bar{z}\\
    &\leq 2 C_{n}''\Vert f \Vert_{C^{K+1}}  \Vert \mathscr V \Vert_{\infty}\left[ \max\left(h^{(K-1)\delta}, h^{(K-4)\delta + \frac{3}{2}}\right) + \sum\limits_{j = 0}^{J-1} h^{\frac{j}{2} + (K - j - 1)\delta} \right] \\
    &\leq C_{n}''\Vert f \Vert_{C^{K+1}} \Vert \mathscr V \Vert_{\infty} h^{(K-1)\delta}.
  \end{split}
\end{equation}
Define $C_{n,K,f,V} = \max(C_n, C_n',C_n'') \Vert \mathscr V \Vert_{\infty} \Vert f \Vert_{C^{K+1}}$. We see 
\[
    |A_{2,\ch}|\leq C_{n,K,f,\msV} h^{\frac{3}{2}-4\delta}, \  \ |A_{2,\ach}| \leq C_{n,K,f,\msV} h^{1-3\delta}, \ \ |A_3|\leq C_{n,K,f,\msV} h^{(K-1) \delta}.
\]
Combine the estimates of $A_1$, $A_2$, $A_3$, and plug them into \eqref{eq: regularized trace}, we have
\begin{equation}
    \label{eq: expansion}
    \begin{split}
    &\tilde{\Tr}f(\msH_\ch) = \frac{1}{\pi h} f(\lambda_{n,B}) + \frac{|n|}{2\pi}f''(\lambda_{n,B}) \mathfrak U(\eta) + \mathcal O_{n,K,f,\msV}(h^{\frac{1}{2} - 4\delta} + h^{(K -1)\delta - 1}),\\
    &\tilde{\Tr}f(\msH^\theta_{\ach}) = \frac{1}{2\pi h}t_{n,0}(f) + \frac{1}{2\pi \sqrt{h}} t_{n,1}(f) + \mathcal O(h^{-3\delta} + h^{(K - 1)\delta - 1})
\end{split}
\end{equation}
where $t_{n,0}(f) = \Ave[f(\lambda_{n,B} - c_n) + f(\lambda_{n,B} + c_n)]$, $t_{n,1}(f) = \Ave[s_n^2f(\lambda_{n,B} - c_n) + s_n^2f(\lambda_{n,B} + c_n)]$, and $\Ave(g) = \frac{1}{|E|}\int_E g(\eta) d\eta$. Thus we proved \eqref{eq: Asymp_chiral} and \eqref{eq: Anti-chiral_d.o.s.}.

In general, fix $N\in \NN^+$ and we consider $2N+1$ Landau levels centered at $0$. Let $B$ be large enough such that $\left\{\overline{B_{\Vert \msV\Vert_\infty}(\lambda_{n,B})}\right\}_{n = -N}^N$ do not intersect. For any $f\in C^{K+1}_c(\lambda_{-(N+1),B} +\Vert \msV\Vert_\infty, \lambda_{N+1,B} - \Vert \msV\Vert_\infty)$, by Remark \ref{rmk: gap}, values of $f$ on the gap do not contribute to $\tilde\Tr(f(\msH^\theta))$, thus we can apply the partition of unity of $f$ on $\left\{\Lambda_{n,B,\msV}\right\}_{n= -N}^N$, i.e. find $f_n$ such that $f = \sum\limits_{n = -N}^N f_n$ and $\supp f_n \subset \Lambda_{n,B,\msV}$. Then we can apply \eqref{eq: Asymp_chiral} and \eqref{eq: Anti-chiral_d.o.s.} to each $f_n$ and take the sum. That gives us the rest of the Theorem \ref{thm: chiral trace} and \ref{thm: anti-chiral trace}.

Furthermore, as mentioned in Remark \ref{rmk: higher_order_terms_on_z}, $z_{n,0} = 0$ in the chiral case, thus each term in the expansion of $r_{n,\ch}$ is of the form $r_{n,j,\ch} = \sum\limits_{k = 0}^{k - j - 1} z^{k - j - 1}f_{n,j,k}(x_2,\xi_2)$. Now assume $f$ is smooth enough, then for any $J\in\NN$, by \eqref{eq: Cauchy_integral}, we can see that 
\begin{equation}
\label{eq:chiralmore}
A_{1,\ch} = \sum\limits_{j = 0}^{J-1} h^{j/2} \sum\limits_{k=0}^{k-j-1} F_{n,j,k}(\eta)f^{(j-k)}(\lambda_{n,B}), \ \text{~for~some~} F_{n,j,k}(\eta)\in S.
\end{equation}
Thus for the chiral case, every term in the asymptotic expansion of $\tilde{\Tr}(f(\msH^\theta))$ only depends on derivatives $f^{(k)}$ at $\lambda_{n,B}$. 

\end{proof}
\subsection{Differentiability}
Finally, we comment on the differentiability of the regularized trace with respect to the magnetic field. That $h \mapsto \tilde \Tr (f(\mathscr H^{\theta}))$ is a differentiable function follows already from Lemma \ref{lemma: trace on E}. However, what does not follow from Lemma \ref{lemma: trace on E} is that the asymptotic expansion itself in Theorems \ref{thm: chiral trace} and \ref{thm: anti-chiral trace} is differentiable. The following Proposition, which uses the same notation as Theorems \ref{thm: chiral trace} and \ref{thm: anti-chiral trace} shows that term-wise differentiation yields the right asymptotic expansion:

\begin{prop}[Differentiability]
\label{prop: differentiability}
  Under the same assumption of $\lambda_{n,B}$, $\epsilon$,  as in Theorem \ref{thm: chiral trace}, we have that $B \mapsto \tilde{\Tr}(f(\msH^\theta))$ is differentiable. For all $\epsilon$, $f\in C^K(\Lambda_{n,B,\msV})$, that $K>\frac{6}{\epsilon} - 2$, then for  $\mathcal O_{n,K,f,\msV} =\mathcal O_n( \Vert \mathscr V \Vert_{\infty} \Vert f \Vert_{C^{K}})$, we have:   For the chiral model $\msH^\theta = \msH_\ch$,
   \begin{equation}
       \label{eq: dh_chiral}
       \begin{split}
        \partial_B \tilde{\Tr}(f(\msH_\ch)) &= \frac{\sqrt{2|n|B}}{2\pi} f'(\lambda_{n,B}) +\frac{ f(\lambda_{n,B})}{\pi} + \frac{(2|n|)^{\frac{3}{2}}}{8\pi \sqrt{B}}\Ave(\mathfrak U) f'''(\lambda_{n,B})+ \mathcal O_{n,K,f,\msV}(B^{-1+\epsilon})
    \end{split}
   \end{equation}
 For the anti-chiral model $\msH^\theta = \msH_\ach^\theta$, 
\begin{equation}
\begin{split}
    \label{eq: dh_antichiral}
     \partial_B \tilde{\Tr}(f(\msH_\ach^\theta)) &=\frac{\sqrt{2|n|B}}{4\pi } t_{n,0}(f') + \frac{1}{4 \pi}\left(2t_{n,0}(f) + \sqrt{2|n|} t_{n,1} (f')\right )+  \mathcal O_{n,K,f,\msV}(B^{-\frac{1}{2} + 3\delta})
     \end{split}
\end{equation}
In particular, when $n = 0$, we get a better estimate for the chiral and anti-chiral case respectively:
\begin{equation}
    \label{eq: n_is_0}
    \begin{split}
        &\partial_B\tilde{\Tr}(f(\msH_{\ch})) = \frac{1}{\pi}f(0) + \mathcal O_{0,K,f,\msV}(B^{-\frac{3}{2} +4\delta})\\
        &\partial_B\tilde{\Tr}(f(\msH_{\ach}^\theta)) = \frac{1}{2\pi}t_{0,0}(f) + \frac{3}{4\pi\sqrt{B}}t_{0,1}(f) + \mathcal O_{0,K,f,\msV} (B^{-1+3\delta})
    \end{split}
\end{equation}
where $t_{n,0}(f)$, $t_{n,1}(f)$, $\mathfrak U$, $s_n$ and $c_n$ are the same as in Theorem \ref{thm: chiral trace}, \ref{thm: anti-chiral trace}.
\end{prop}

To prove this proposition, we will need to prove two auxiliary Lemmas \ref{lemm: boundedness_number_2} and \ref{lemm:bound} discussing properties of $\partial_h \Epm$, $\partial_h \Epm^{-1}$ and $\partial_h r_n$, which are similar to the two properties needed for $\Epm$, $\Epm^{-1}$ and $r_n$ previously in \ref{lemma: Properties of r} and \ref{lemma: explicit_asymptotic_expansion_needed}. The rest of the proof is similar to Sec. \ref{ss: proof_of_dos_expansion}. We start with some preparations:
To discuss the differentiability of asymptotic expansions, we define $\#_h^M$ for $a(x,\xi;h), b(x,\xi;h)\in S(\RR_{x,\xi}^2)$ by 
\begin{equation}
  \label{eq: new_sharp}
  \begin{split}
    a \#^M_h b &= \left[e^{\frac{ih}{2}\sigma(D_x,D_{\xi}; D_y, D_\eta)} \left(\frac{i}{2}\sigma(D_x,D_\xi;D_y,D_\eta)\right)^M\right]\left(a(x,\xi,h) b(y,\eta,h)\right)|_{\substack{x = y \\ \xi = \eta}}\\
    &= \sum\limits_{|\alpha| = |\beta| = M} C_{\alpha, \beta} (\partial_{x,\xi}^\alpha a ) \# (\partial_{y,\eta}^\beta b),
  \end{split}
\end{equation}
where $\sigma(x,\xi;y,\eta) = \langle \xi, y \rangle - \langle x,\eta\rangle$. Then we see that, 
\begin{equation}
  \label{eq: diff_sharp}
  \partial_h^M(a \# b) = a \#^M_h b \ + \sum\limits_{\substack{i + j + k = M \\  j \neq M}} C_{i,j,k} \ (\partial_h^i a) \ \#_h^j \ (\partial_h^k b).
\end{equation}
The following result is derived for general $M \in \mathbb N$ but we will, for simplicity, only consider the  $M =1$ case later:
\begin{lemm}[Boundedness]
\label{lemm: boundedness_number_2}
Let $h_0$, $\Epm$ be as in Lemma \ref{lemm:perturbed Grushin}. The symbol $\Epm(x_2,\xi_2;z,h)$ is smooth in $h$ when $h<h_0$ and for any $M\in \NN$, $\partial_h^M \Epm \in S^{M - \frac{1}{2} }$ uniformly in $|z|\leq 2\Vert \msV\Vert_\infty$, i.e. for any multi-index $\alpha$, $\beta$, there is $C_{\alpha, \beta, n} = C_{\alpha, \beta, n}(\Vert \msV\Vert_\infty)$ such that 
\[
 \Vert\partial_{x_2}^\alpha\partial_{\xi_2}^\beta \partial_h^M \Epm(x_2,\xi_2;z,h)\Vert_{\CC_{2\times 2}}\leq C_{\alpha,\beta,n}\sqrt{h}, \quad \text{ for all } |z|\leq 2\Vert \msV\Vert_\infty.
\]
If $|\Im z|\neq 0$, $M>0$, then $\partial_h^M \Epm^{-1}$ and $\partial_h^M r_n$ satisfy
\begin{align}
    \label{eq: partial_Epm}
    &\Vert \partial_{x_2}^\alpha\partial_{\xi_2}^\beta\partial_h^M \Epm^{-1}(x_2,\xi_2;z,h)\Vert_{\CC_{2\times 2}} \leq C_{\alpha,\beta,n}  \max\left(1, \frac{h^{\frac{3}{2}}}{|\Im z|^3}\right) h^{-\frac{1+2M}{2}} |\Im z|^{-2M - |\alpha| - |\beta|},\\
    \label{eq: partial_r}
    &\Vert \partial_{x_2}^\alpha\partial_{\xi_2}^\beta\partial_h^M r_n(x_2,\xi_2;z,h)\Vert_{\CC_{2\times 2}} \leq C_{\alpha,\beta,n}  \max\left(1, \frac{h^{\frac{3}{2}}}{|\Im z|^3}\right) h^{-M} |\Im z|^{-2M - |\alpha| - |\beta|}.
\end{align}
In particular, when $0<\delta<1/2$ and $|\Im z|\geq h^\delta$, we have $\partial_h^M \Epm^{-1} \in S_\delta^{M(2\delta +1) + \frac{1}{2}}$ and  $\partial_h^M r_n\in S_\delta^{M(2\delta+1)}$.
\end{lemm}

\begin{proof}
 Let $\mathcal P_n$ be as in Lemma \ref{lemm:perturbed Grushin}, by \eqref{eq: everything_in_S}, $\msG^\theta - \sqrt{h} z$, $R_n^\pm\in S(\RR_{x_2,\xi_2}^2)$. Furthermore, since $\msG^\theta = \msG_0^\theta + \sqrt{h}\tilde{\msV}^w$, by direct computation, we see $\partial_h^M (\msG^\theta - \sqrt{h}z)\in S^{M - \frac{1}{2}}$ while $\partial_h^M R_n^\pm = 0$, for $M >0$.

 Then consider $\mathcal E_n = \mathcal P_n^{-1}$. First of all, by the proof of Lemma \ref{lemma: Properties of r}, we have $\mathcal E_n(x,D_{x_1},\xi_2)\in S(\RR_{x_2,\xi_2}^2;\mathcal{L}(B_{x_1}^k\times \CC^2;B_{x_1}^{k+1} \times \CC^2))$. By differentiating  $\mathcal E_n = \mathcal E_n \# \mathcal P_n \# \mathcal E_n$ with respect to $h$ and using \eqref{eq: new_sharp} and \eqref{eq: diff_sharp}, we have
 \begin{equation}
   \label{eq: partial_h_mathcal_E}
   \partial_h \mathcal E_n = -\mathcal E_n \# \partial_h \mathcal P_n \# \mathcal E_n + \sum\limits_{|\alpha| = |\beta| = 1} C_{\alpha, \beta} \left( \partial_{x_2,\xi_2}^\alpha \mathcal E_n \# \partial_{x_2,\xi_2}^\beta \mathcal P_n \# \mathcal E_n\right).
 \end{equation}
Since $\partial_h \mathcal P_n \in S^{\frac{1}{2}}$, thus $\partial_h \mathcal E_n \in S^{\frac{1}{2}}$ above. By differentiating \eqref{eq: partial_h_mathcal_E} with respect to $h$ and using  \eqref{eq: new_sharp} and \eqref{eq: diff_sharp}, we see that $\partial_h^2 \mathcal E_n \in S^{\frac{3}{2}}$. An iterative argument shows that $\partial_h^M \mathcal E_n \in S^{M - \frac{1}{2}}$. In particular, $\partial_h^M \Epm \in S^{M - \frac{1}{2}}$. Furthermore, by differentiating $\Epm^{-1} = \Epm^{-1}\# \Epm \#\Epm^{-1}$ with respect to $h$ and using \eqref{eq: diff_sharp} and \eqref{eq: new_sharp}, we have 
  \begin{equation}
  \begin{split}
    \label{eq: partial_h_c}
    \partial_h \Epm^{-1} = -\Epm^{-1} \# \partial_h \Epm \# \Epm^{-1} - \sum\limits_{|\alpha| = |\beta| = 1} C_{\alpha, \beta} \partial_{x_2,\xi_2}^\alpha \Epm^{-1} \# \partial_{x_2,\xi_2}^\beta \Epm \# \Epm^{-1}. 
    \end{split}
  \end{equation}
  When $|\Im z| \geq h^\delta$, by \eqref{eq: bound on r} and \cite[Theorem 4.23(ii)]{Zw12}, we see that 
  \[
  \Vert \partial_h \Epm^{-1}\Vert = \mathcal O(h^{-\frac{3}{2}}|\Im z|^{-2})+ \mathcal O(h^{-\frac{1}{2}}|\Im z |^{-3}) =  \mathcal O(h^{-\frac{3}{2}}|\Im z|^{-2}).\] Furthermore, since $[D_{x_j}, A^W] = (D_{x_j}A)^W$ and $-[x_j, A^W] = (hD_{\xi_j}A)^W$, we see that 
  \[
    \begin{split}
     \Vert \operatorname{ad}_{L_{j_1}} \circ \cdots \operatorname{ad}_{L_{j_N}} (\partial_h \Epm^{-1})^W\Vert = \mathcal O\left(\frac{h^{-\frac{3}{2}}}{|\Im z|^{2}} \ \frac{h^N}{|\Im z|^N}\right).
    \end{split}
  \]
  By \cite[Prop. 8.4]{DS99}, we get 
  \begin{equation}
    \label{eq: first_partial_h}
      \Vert \partial_{x_2}^\alpha\partial_{\xi_2}^\beta\partial_h \Epm^{-1}(x_2,\xi_2;z,h)\Vert_{\CC_{2\times 2}} \leq C_{\alpha,\beta} \max\left(1, \frac{h^{\frac{3}{2}}}{|\Im z|^3}\right) h^{-\frac{3}{2}} |\Im z|^{-2 - |\alpha| - |\beta|}.
  \end{equation}
  Iterating this process by taking $\partial_h$ of \eqref{eq: partial_h_c}, expanding it and using  \eqref{eq: new_sharp},  \eqref{eq: diff_sharp} and \eqref{eq: first_partial_h}, we see that every time we differentiate, we derive an extra order of $1/(h|\Im z|^2)$. Thus we obtain \eqref{eq: partial_Epm} for $M>0$. Then
  \[
    \partial_h r_n = \partial_h \partial_z \Epm \# \Epm^{-1} + \partial_z \Epm \# \partial_h \Epm^{-1} + \sum\limits_{|\alpha| = |\beta| = 1} C_{\alpha, \beta} (\partial_{x_2,\xi_2}^\alpha \partial_z \Epm) \# (\partial_{x_2,\xi_2}^\beta\Epm^{-1}).
  \]
  By \eqref{eq: bound on r}, \eqref{eq: partial_Epm} and \cite[Theorem 4.23(ii)]{Zw12}, we see that $\Vert \partial_h r_n^W\Vert = \mathcal O(h^{-1}|\Im z|^{-2})$. By the same argument as for $\Epm^{-1}$, we get \eqref{eq: partial_r}.
\end{proof}

We shall now focus on $M = 1$, for simplicity, and study the asymptotic expansion of $\partial_h r_n$.
 \begin{lemm}[Asymptotic expansion]
 \label{lemm:bound}
Let $0<\delta<1/2$ and $|\Im z|\geq h^\delta$, then $\partial_h r_n$ has an asymptotic expansion in $S_\delta^{1+2\delta}$: 
\begin{equation}
    \label{eq: asymptotic_of_dh_r_n}
    \partial_h r_n \sim \sum\limits_{j = 1}^\infty  \ \frac{j}{2} \  h^{\frac{j}{2} - 1}r_{n,j} =: \sum\limits_{j = 1}^\infty h^{\frac{j}{2} - 1}q_{n,j}, \text{~where~} r_{n,j} \text{~are~given~in~Lemma~\ref{lemma: explicit_asymptotic_expansion_needed}.}
\end{equation}
Then $h^{\frac{j}{2} - 1} q_{n,j}\in S_\delta^{(j+1)\delta +1 -\frac{j}{2}}$. Let $Q_{n,J} := \partial_h r_n - \sum\limits_{j = 1} ^{J-1} h^{\frac{j}{2}-1} q_{n,j} \in S_\delta^{(J+1)\delta +1 - \frac{J}{2}}$, i.e., $\text{ for all } \alpha, \beta>0$, there is $C_{\alpha, \beta,n}''$ such that 
\begin{equation}
    \label{eq: Q_n,J}
     \sup\limits_{(x_2,\xi_2)\in\RR^2}|\partial_{x_2,}^\alpha\partial_{\xi_2}^\beta Q_{n,J}| \leq C_{\alpha, \beta, n}''h^{\frac{J}{2} - 1 - (J+1)\delta - \delta(|\alpha| + |\beta|)}.
\end{equation}
Furthermore, for the expansion of $\Tr_{\CC^2}(\partial_h r_n)$, we have for $\eta = x_2+i\xi_2$, 
 \begin{equation}
 \label{eq: explicit_exp_of_q_n}
    \begin{split}
      \operatorname{Chiral} \ \msH_{\ch, n} (J = 3): & \Tr_{\CC^2}( h^{-\frac{1}{2}}q_{n,1} + q_{n,2}) =\frac{\lambda_n^2}{z^3}\mathfrak U(\eta), \\
      \operatorname{Anti-Chiral} \ \msH_{\ach,n}^\theta (J = 2): & \Tr_{\CC^2}(h^{-\frac{1}{2}}q_{n,1}) = \frac{s_n^2(z^2+c_n^2)}{(z^2 - c_n^2)^2\sqrt{h} }.
    \end{split}  
  \end{equation}
 \end{lemm}
 We will prove that the termwise differentiation of the asymptotic expansion of $r_n$ in \eqref{eq: r_n_on_z} is indeed an asymptotic expansion of $\partial_h r_n$ in $S_\delta^{2\delta+1}$.
 \begin{proof}
 Let $g = \sqrt{h}$ and consider $r_n \sim \sum\limits_{j = 0}^\infty g^j r_{n,j}$. By Borel's theorem, see for instance \cite[Theorem 4.15]{Zw12} or \cite[Theorem 1.2.6]{Ho03}, we see that for such $r_{n,j}\in C^\infty(\RR_{x_2,\xi_2}^2)$, there is $\tilde{r}_n\in C^\infty(\RR_{g}^+ \times \RR_{x_2,\xi_2})$ such that $\tilde{r}_n = \sum\limits_{j = 0}^\infty g^j r_{n,j}$. Thus
 \begin{equation}
     \label{eq: termwise_diff}
    \partial_g \tilde{r}_n = \sum\limits_{j = 1}^\infty j g^{j-1} r_{n,j}.
 \end{equation}
  On the other hand, uniqueness in Borel's theorem implies that $\tilde{r}_n - r_n = \mathcal O(h^\infty)$. Thus $\partial_g \tilde{r}_n - \partial_g r_n = \mathcal O(g^\infty)$. Thus \eqref{eq: termwise_diff} is also an asymptotic expansion of $\partial_g r_n$. Furthermore, since $\partial_h r_n = \frac{1}{2\sqrt{h}}\partial_g r_n$, thus we proved $\partial_h r_n$ has the following asymptotic expansion in $S_\delta^{1 + 2\delta}$:
  \[
  \partial_h r_n \sim \frac{1}{2\sqrt{h}}\sum\limits_{j =1}^\infty j h^{\frac{j - 1}{2}}r_{n,j} = \sum\limits_{j =1}^\infty \ \frac{j}{2}h^{\frac{j}{2} -1 }r_{n,j}.
  \]
  The rest of the Lemma follows from Lemma \ref{lemma: explicit_asymptotic_expansion_needed}.
\end{proof}
 \begin{proof}[Proof of Prop. \ref{prop: differentiability}
]
Recall that $f_0(z) = f(z+\sqrt{2|n|/h})$ also depends on $h$. By differentiating \eqref{eq: Trace_formula} with respect to $h$, we get
\begin{equation}
      \begin{split}
          \partial_h &\tilde{\Tr}(f(\msH^{\theta}_{c,n}))  = \frac{i}{4\pi^2 h^2|E|} \int_{\CC} \int_{E}\partial_{\bar{z}} \tilde{f}_0(z) \Tr_{\CC^2}(r_n) \ dx_2 \ d\xi_2 ~dz\wedge d\bar{z}\\
          &+\frac{i\sqrt{2\vert n\vert/h}}{8\pi^2 h^{2}|E|} \int_{\CC} \int_{E}\partial_{\bar{z}} \tilde{f'}_0(z+\sqrt{2n/h}) \Tr_{\CC^2}(r_n) \ dx_2 \ d\xi_2 ~dz\wedge d\bar{z} \\
          &- \frac{i}{4\pi^2 h|E|} \int_{\CC} \int_{E}\partial_{\bar{z}} \tilde{f}_0(z) \Tr_{\CC^2}(\partial_h r_n) \ dx_2 \ d\xi_2 ~dz\wedge d\bar{z}:= -B_1 - B_2 - B_3.
       \end{split}
  \end{equation}
where the asymptotic expansion of $B_{1} = \frac{1}{h}\tilde{\Tr} (f(\msH^\theta))$ and $B_{2} = \sqrt{\frac{|n|}{2h^3}}\tilde{\Tr}(f'(\msH^\theta))$ are known by \eqref{eq: expansion}. While $B_3$ can be computed by splitting the integral as in Subsection \ref{ss: proof_of_dos_expansion}:
\[
  \begin{split}
    \left[\int_{\CC} \partial_{\bar{z}}\tilde{f}_0 \Tr_{\CC^2} (\partial_h r_n) dz \wedge d\bar{z}\right](x_2,\xi_2;h) =& \int_{\CC} \partial_{\bar{z}}\tilde{f}_0 \sum\limits_{j = 1}^{J-1} h^{\frac{j}{2}-1}\Tr_{\CC^2}(q_{n,j})dz\wedge d\bar{z}\\
    +&\int_{|\Im z|\geq h^\delta} \partial_{\bar{z}}\tilde{f}_0 \Tr_{\CC^2}(Q_{n,J}) dz\wedge d\bar{z}\\
    +& \int_{|\Im z|\leq h^\delta} \partial_{\bar{z}}\tilde{f}_0 \Tr_{\CC^2}(Q_{n,J}) dz\wedge d\bar{z}
    := A_{1}' + A_{2}' + A_{3}',
  \end{split}
\]
and we imitate the estimates of $A_1$, $A_2$, $A_3$ in the Subsection \ref{ss: proof_of_dos_expansion} with $\partial_h r_n$ instead of $r_n$, and we use Lemma \ref{lemm: boundedness_number_2} and \ref{lemm:bound} instead of Lemma \ref{lemma: Properties of r} and \ref{lemma: explicit_asymptotic_expansion_needed}. In short, we need \eqref{eq: Cauchy_integral} and \eqref{eq: explicit_exp_of_q_n} for $A_1'$, \eqref{eq: Q_n,J} for $A_2'$, \eqref{eq: partial_r} and Lemma \ref{lemm:bound} for $A_3'$ and we derive that
\[
\begin{cases}
    A_{1,\ch}' = \pi i f''(\lambda_{n,B})\lambda_n^2\mathfrak U(\eta), \quad A_{1,\ach}' = \frac{\pi i}{\sqrt{h}} \left(s_n^2f'(\lambda_{n,B} - c_n) + s_n^2f'(\lambda_{n,B} + c_n)\right), \\
    |A_2'|\leq C_{n,K,f,\msV} h^{\frac{J}{2} - 1 - (J+1)\delta}, |A_3'|\leq C_{n,K,f,\msV} h^{(K - 2)\delta - 1},
\end{cases}
\]
from which we can find $B_3$. We summarize $B_1$, $B_2$, $B_3$ below:

For the chiral model where $J = 3$, we have
\[
\begin{split}
    &B_{1,\ch} = \frac{1}{\pi h^2} f(\lambda_{n,B}) + \frac{|n|}{2\pi h}\Ave(\mathfrak U) f''(\lambda_{n,B}) + \mathcal O_{n,K,f,\msV}(h^{-\frac{1}{2} - 4\delta} + h^{(K - 1)\delta - 2}),\\
    &B_{2,\ch}  = \frac{\sqrt{2|n|}}{2\pi h^{\frac{5}{2}}} f'(\lambda_{n,B}) + \frac{(2|n|)^{\frac{3}{2}}}{8\pi h^{\frac{3}{2}}}\Ave(\mathfrak U) f'''(\lambda_{n,B}) + \mathcal O_{n,K,f,\msV}(h^{-1 - 4\delta} + h^{(K -1) \delta - \frac{5}{2}})\\
    &B_{3,\ch} = \frac{|n|}{2\pi h} f''(\lambda_{n,B})\Ave(\mathfrak U) + \mathcal O_{n,K,f,\msV}(h^{-\frac{1}{2} - 4\delta} + h^{(K-2)\delta - 2}).
\end{split}
\]
When $n \neq 0$ and $K>\frac{3}{2\delta} - 3$, we have 
\[
    \partial_h\tilde{\Tr}(f(\msH_{\ch,n})) = -\frac{\sqrt{2|n|}}{2\pi h^{\frac{5}{2}}} f'(\lambda_{n,B}) - \frac{1}{\pi h^2} f(\lambda_{n,B}) - \frac{(2|n|)^{\frac{3}{2}}}{8\pi h^{\frac{3}{2}}}\Ave(\mathfrak U) f'''(\lambda_{n,B}) + \mathcal O_{n,K,f,\msV}h^{-1 - 4\delta}.
\]
When $n  = 0$ and $K>\frac{3}{2\delta} - 3$, since $B_2 = 0$, we get a better estimate:
\[
    \partial_h\tilde{\Tr}(f(\msH_{\ch,0})) = -\frac{1}{\pi h^2} f(0)  - O_{0,K,f,\msV}h^{-\frac{1}{2} - 4\delta}.
\]
For the anti-chiral model where $J = 2$, we have
\[
\begin{split}
     &B_{1,\ach}   = \frac{1}{2\pi h^2}t_{n,0}(f) + \frac{1}{2\pi h^{\frac{3}{2}}}t_{n,1}(f) + \mathcal O_{n,K,f,\msV}(h^{-1-3\delta} + h^{(K-1)\delta - 2}),\\
    &B_{2,\ach}  = \frac{\sqrt{2|n|}}{4\pi h^{\frac{5}{2}}} t_{n,0}(f') +\frac{\sqrt{2|n|}}{4\pi h^{2}}t_{n,1}(f') + \mathcal O_{n,K,f,\msV}(h^{-\frac{3}{2} - 3\delta} + h^{(K - 1)\delta - \frac{5}{2}}),\\
    &B_{3,\ach} =  \frac{1}{4\pi h^{\frac{3}{2}}}t_{n,1}(f) +  \mathcal O_{n,K,f,\msV}(h^{-1 - 3\delta} + h^{(K - 2) \delta - 2}).
\end{split}
\] 
Thus when $n \neq 0$ and $K>\frac{1}{\delta} - 2$, we have
\[
    \partial_h \tilde{\Tr}(f(\msH^\theta_{\ach})) = -\frac{\sqrt{2|n|}}{4\pi h^{\frac{5}{2}}} t_{n,0}(f') - \frac{1}{4\pi h^2}\left(2t_{n,0}(f) +\sqrt{2|n|}t_{n,1}(f')\right) - \mathcal O_{n,K,f,\msV}h^{-\frac{3}{2} - 3\delta}.
\]
If $n = 0$ and $K>\frac{1}{\delta} - 2$, since $B_2 = 0$, we get a better estimate:
\[
\partial_h \tilde{\Tr}(f(\msH^\theta_{\ach})) =- \frac{1}{2\pi h^2}t_{0,0}(f)  - \frac{3}{4\pi h^{\frac{3}{2}}}t_{0,1}(f)+  O_{0,K,f,\msV}h^{-1 - 3\delta}.
\]
Recall $h =\frac{1}{B}$. By $\partial_B = -\frac{1}{B^2} \partial_h$, we get the results \eqref{eq: dh_chiral}, \eqref{eq: dh_antichiral} and \eqref{eq: n_is_0}.
\end{proof}

\section{Magnetic response quantities}
\label{sec:applications}

This section discusses applications of the regularized trace expansions derived in the previous section, cf. Theorems \ref{thm: chiral trace} and \ref{thm: anti-chiral trace} as well as Proposition \ref{prop: differentiability}. They form the rigorous foundation of our analysis in this section and we shall focus on qualitative features rather here, instead. 
\begin{figure}
    \centering
    \includegraphics[width=7.5cm]{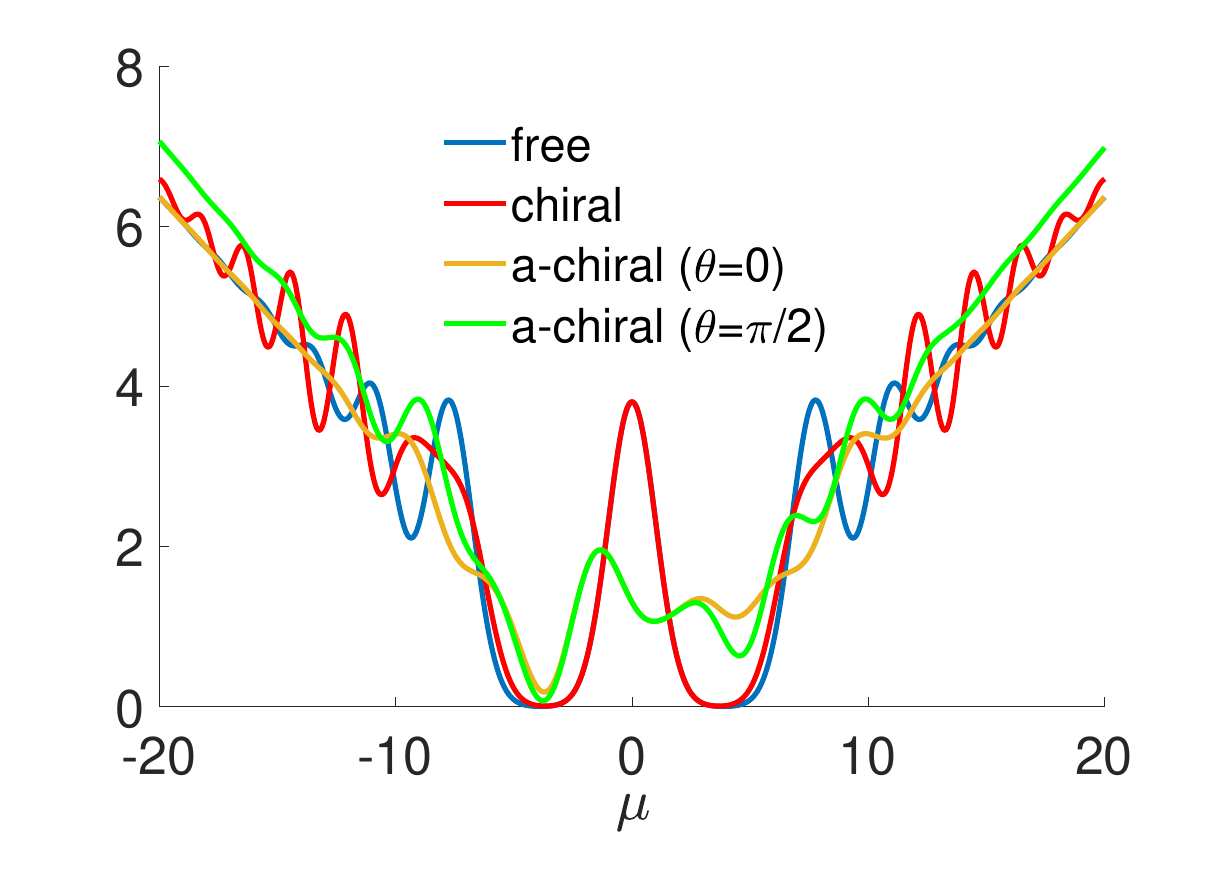} 
    \includegraphics[width=7.5cm]{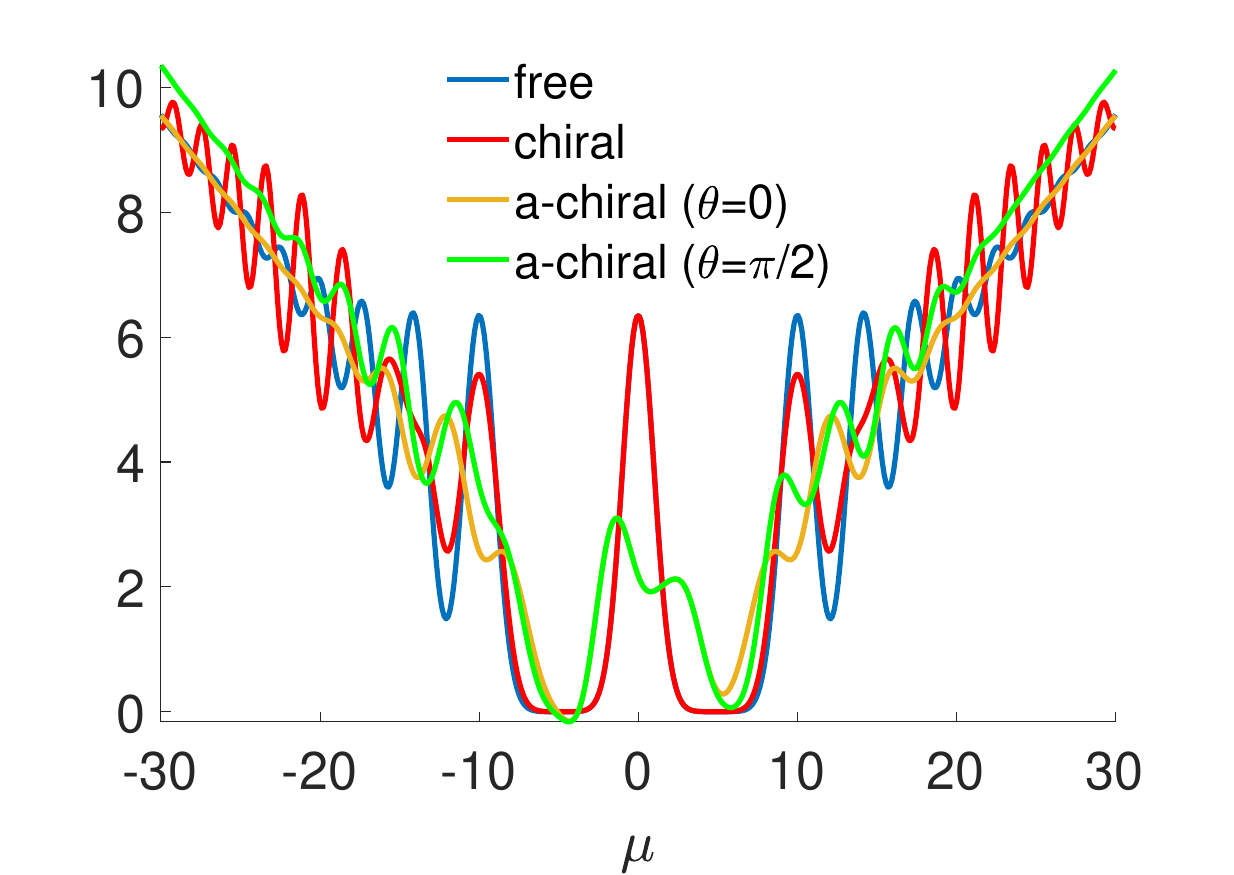} 
    \caption{SdH oscillations: Smoothed out DOS $\rho(f_{\mu})$ with $f_{\mu}(x) = e^{-\frac{(x-\mu)^2}{2\sigma^2}}/\sqrt{2\pi}\sigma$ illustrating the oscillatory features. On the left, $B=30$ and on the right $B=50$ for $\sigma=1.$}
    \label{fig:gaussfaltung}
\end{figure}

\begin{figure}
    \includegraphics[width=7.5cm]{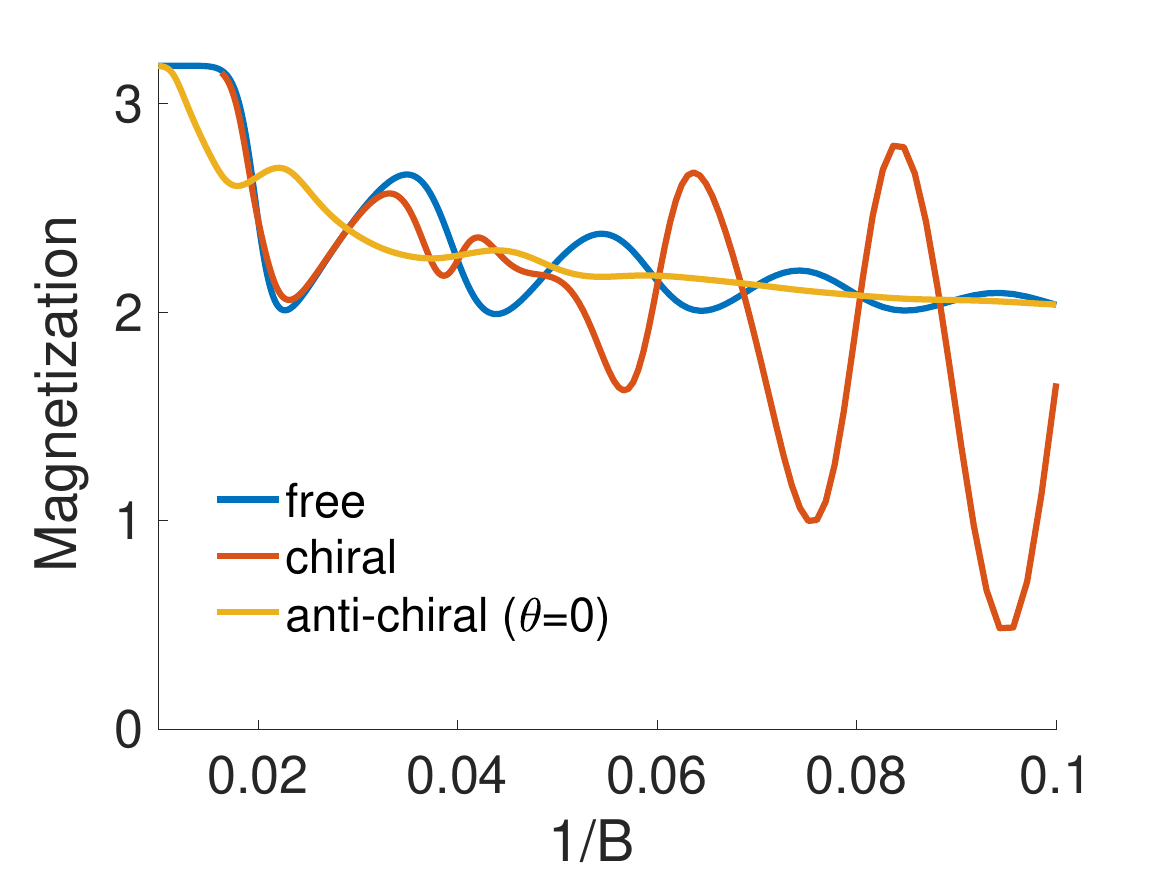}
    \includegraphics[width=7.5cm]{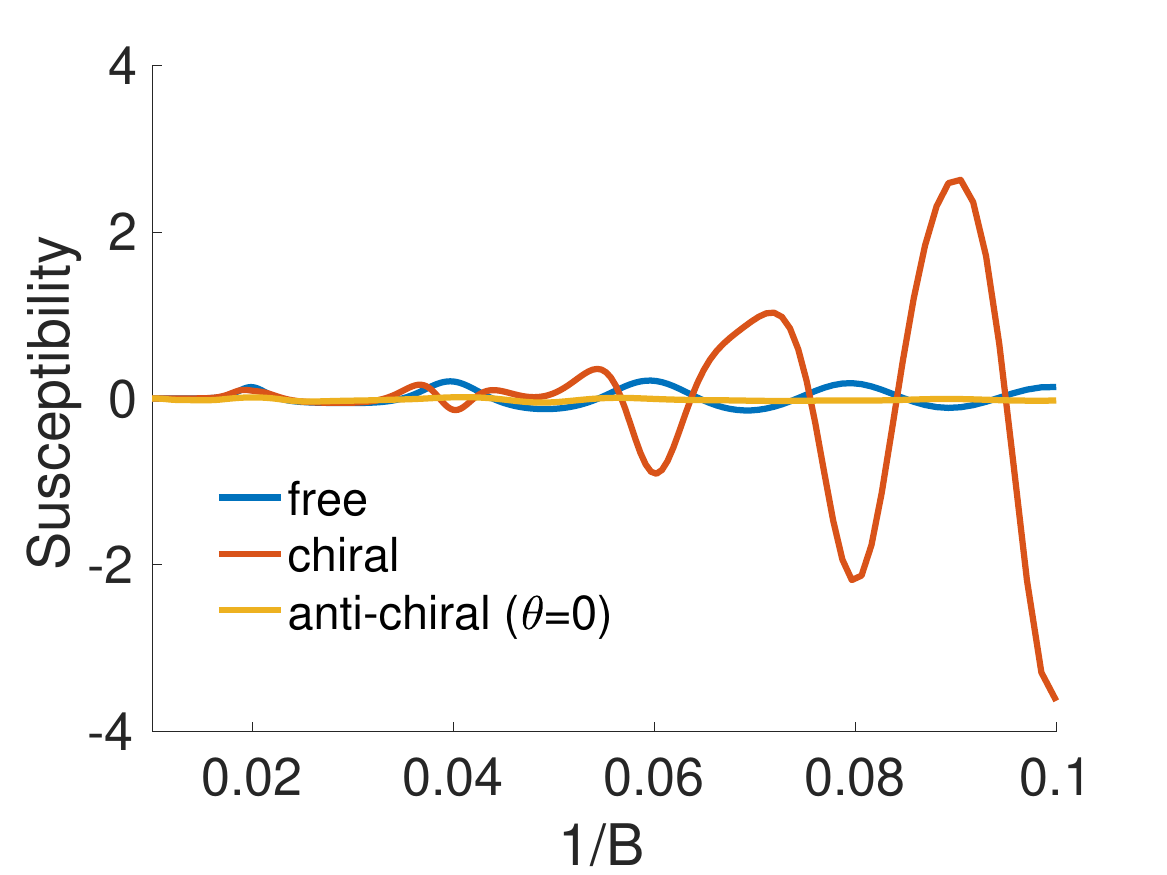}
     \caption{Magnetization and susceptibility for $\beta=4$, $\alpha_i=3/5,$ and chemical potentials  $\mu=5$ (left) and $\mu=10$ (right).}
    \label{fig:dHvA}
\end{figure}
Our main contribution on magnetic response properties of TBG is a careful analysis of the oscillatory behaviour of the DOS. While this effect can be easily explained using the Poisson summation formula, we shall illustrate this phenomenon, by considering a Gaussian density $f_{\mu}(x) = e^{-\frac{(x-\mu)^2}{2\sigma^2}}/\sqrt{2\pi}\sigma$ and analyze the Shubnikov--de Haas (SdH) oscillations in a smoothed-out version of the DOS $\mu \mapsto \rho(f_{\mu})$ in Figure \ref{fig:gaussfaltung} for $\sigma=1$ using the asymptotic formulae of Theorems \ref{thm: chiral trace} and \ref{thm: anti-chiral trace}. 
As a general rule from our study, we find that the AB/BA interaction leads to an enhancement of this oscillatory behaviour compared to the non-interacting case, while the $AA^{\prime}$/$BB^{\prime}$ interaction damps oscillations. The smoothing effect of the $AA^{\prime}$/$BB^{\prime}$ interaction is due to a splitting and broadening of the highly degenerate Landau levels. This splitting has also consequences for the Quantum Hall effect, see Fig. \ref{fig:QHE}. We also study the de Haas--van Alphen (dHvA) effect in TBG, see Fig. \ref{fig:dHvA} and \ref{fig:my_label} for which we find a similar phenomenon.

We study magnetic response quantities by thoroughly analyzing the following cases: 
\begin{itemize}
  \item The free or non-interacting case, corresponds to two non-interacting sheets of graphene modeled by the direct sum of two magnetic Dirac operators, see also \cite{BZ19,BHJZ21} for similar results in a quantum graph model and \cite{SGB04} for a thorough analysis of the magnetic Dirac operator, directly.
    \item The chiral case, which corresponds to pure AB/BA interaction.
     \item The anti-chiral case, which corresponds to pure $AA^{\prime}$/$BB^{\prime}$ interaction.
\end{itemize}
For our analysis of the de Haas-van Alphen effect, we shall employ a cut-off function $\eta_N \in C_c^{\infty}(\RR)$ that is one on the interval $[0,\sqrt{2BN}]$ and smoothly decays to zero outside of that interval, enclosing precisely $N+1$ Landau levels and $\eta_N^{\operatorname{sym}}$ which is equal to one on $[-\sqrt{2BN},\sqrt{2BN}].$ The choice of cut-off function mainly plays the role of a reference frame. In particular, for the study of magnetic oscillations it seems more natural to consider $\eta_N$ instead of $\eta_N^{\operatorname{sym}}$ as the former cut-off function singles out the effect of individual Landau levels moving past a fixed chemical potential $\mu.$ 
We shall employ the leading order terms for the regularized trace in this section, as specified in  Theorems \ref{thm: chiral trace} and \ref{thm: anti-chiral trace} and Proposition \ref{prop: differentiability}. For this reason, we write functionals $\rho(f),$ where $f \in C^{\infty}(\RR),$ as $\rho(f) \sim g,$ to indicate that $g$ are the first terms in the asymptotic expansion of $\rho(f)$ and analogously for derivatives of $\rho(f)$ with respect to the magnetic field. 
\begin{figure}[!ht]
    \centering
    \includegraphics[width=7.5cm]{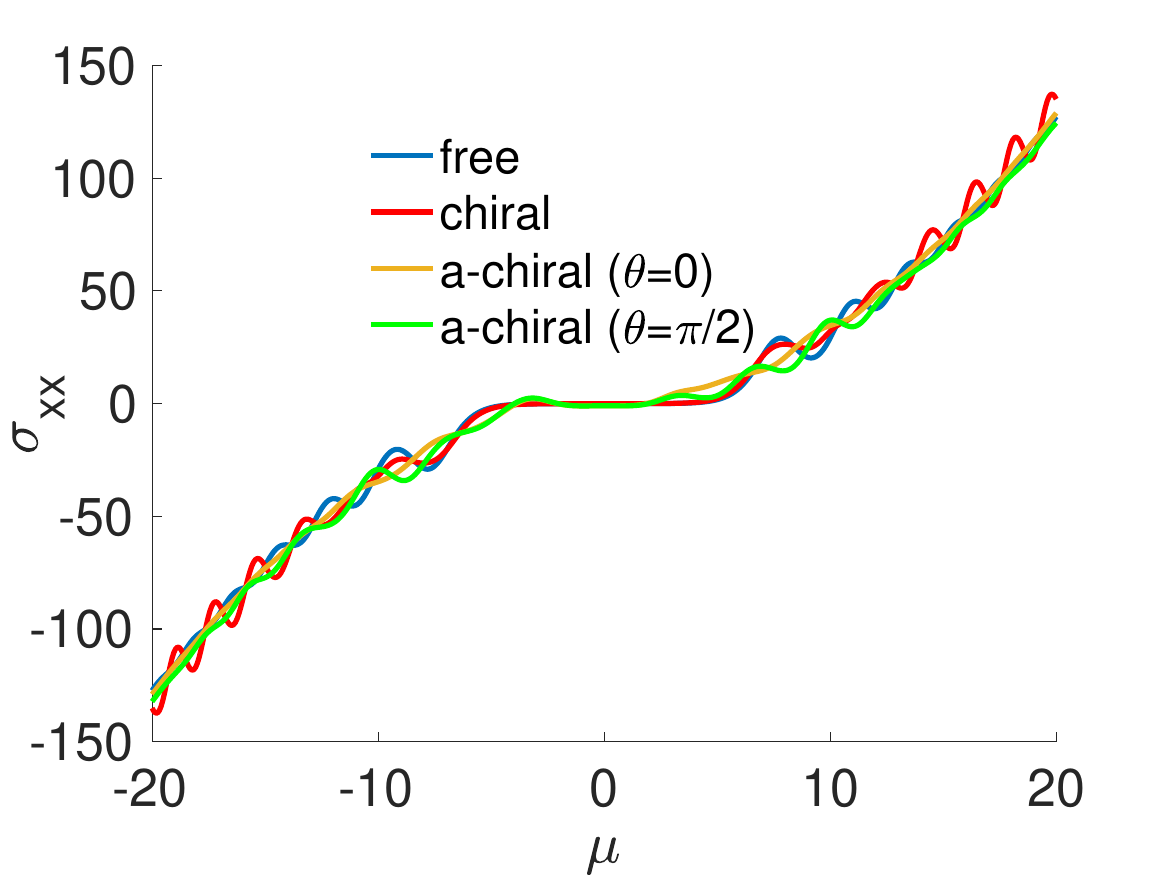} 
    \includegraphics[width=7.5cm]{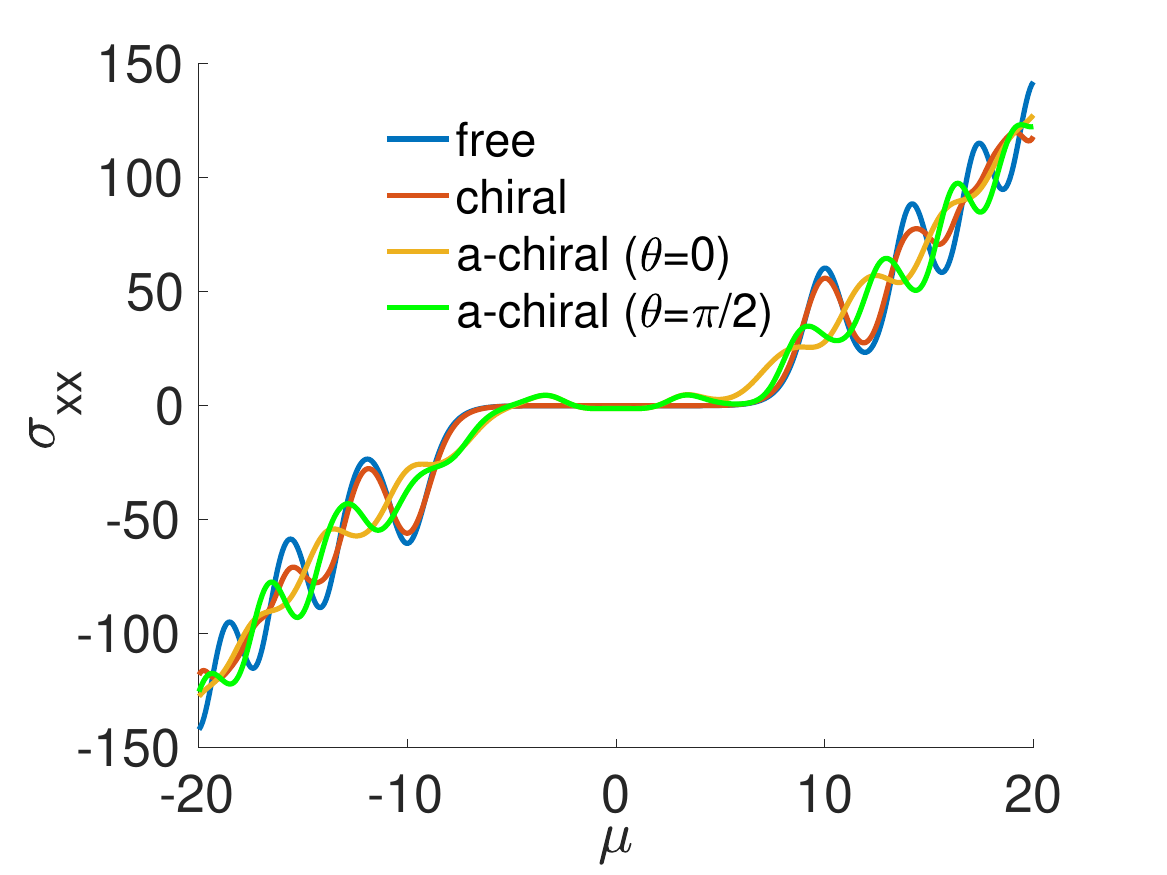}  \\
    \caption{Smoothed out longitudinal conductivity $\sigma_{xx} \propto -\rho(\lambda n'_{\beta}(\lambda-\mu))$ with $n_{\beta}$, the Fermi-Dirac distribution, showing Shubnikov-de Haas oscillations. On the left, $B=30$ and on the right $B=50$ for $\beta=1.5.$ with $\alpha_i = \frac{3}{5}.$}
    \label{fig:sigmaxx_chem}
\end{figure}
\subsection{Shubnikov-de Haas oscillations}
We shall start by discussing \emph{Shubnikov - de Haas (SdH)} oscillations in the density of states. A common method of measuring SdH oscillations is by measuring longitudinal conductivity and resistivity, see also \cite{W11,T11}. In the following, let $\sigma \in \mathbb R^{2 \times 2}$ be the conductivity matrix, such that the current density $j=\sigma E$, where $E$ is an external electric field, then the resistivity matrix is just $\rho = \sigma^{-1}.$ Hence, we shall focus on conductivities in the sequel.

The SdH oscillations are most strongly pronounced at low temperatures in the regime of strong magnetic fields and describe oscillations in the longitudinal conductivity $\sigma_{xx}$ of the material. 

The expression for the longitudinal conductivity goes back to Ando et al \cite{A70,A82} who derived the following relation, see also \cite{FS14},  
\[\sigma_{xx}(\beta,\mu,B) = -\int_{0}^{\infty} n_{\beta}'(\lambda-\mu) \lambda \eta^{\text{sym}}_N(\lambda) \ d\rho(\lambda), \]
where $n_{\beta}(x) = \frac{1}{e^{\beta x}+1}$ is the Fermi-Dirac statistics.
In the free case, i.e. without any tunnelling potential, the oscillations happen precisely at the relativistic Landau levels. For the chiral model, oscillations caused by higher Landau levels are enhanced compared to the free case, whereas oscillations in the anti-chiral case are much more smoothed out.

The oscillatory behaviour of the longitudinal conductivity is visible both as a function of chemical potential, for a fixed magnetic field strength, as shown in Fig. \ref{fig:sigmaxx_chem} as well as function of inverse magnetic field in Fig. \ref{fig:sigmaxx_B} for fixed chemical potential.

\begin{figure}[!ht]
    \centering
    \includegraphics[width=7.5cm]{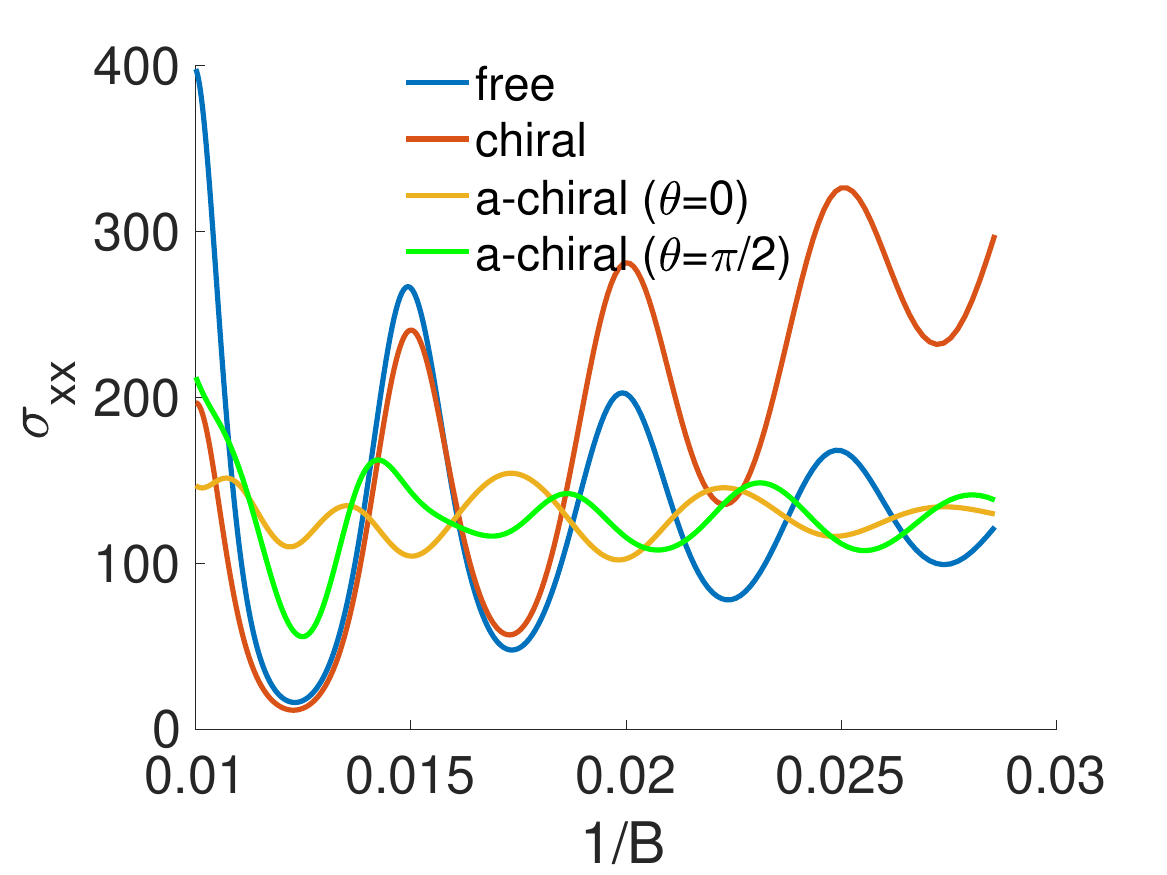} 
    \includegraphics[width=7.5cm]{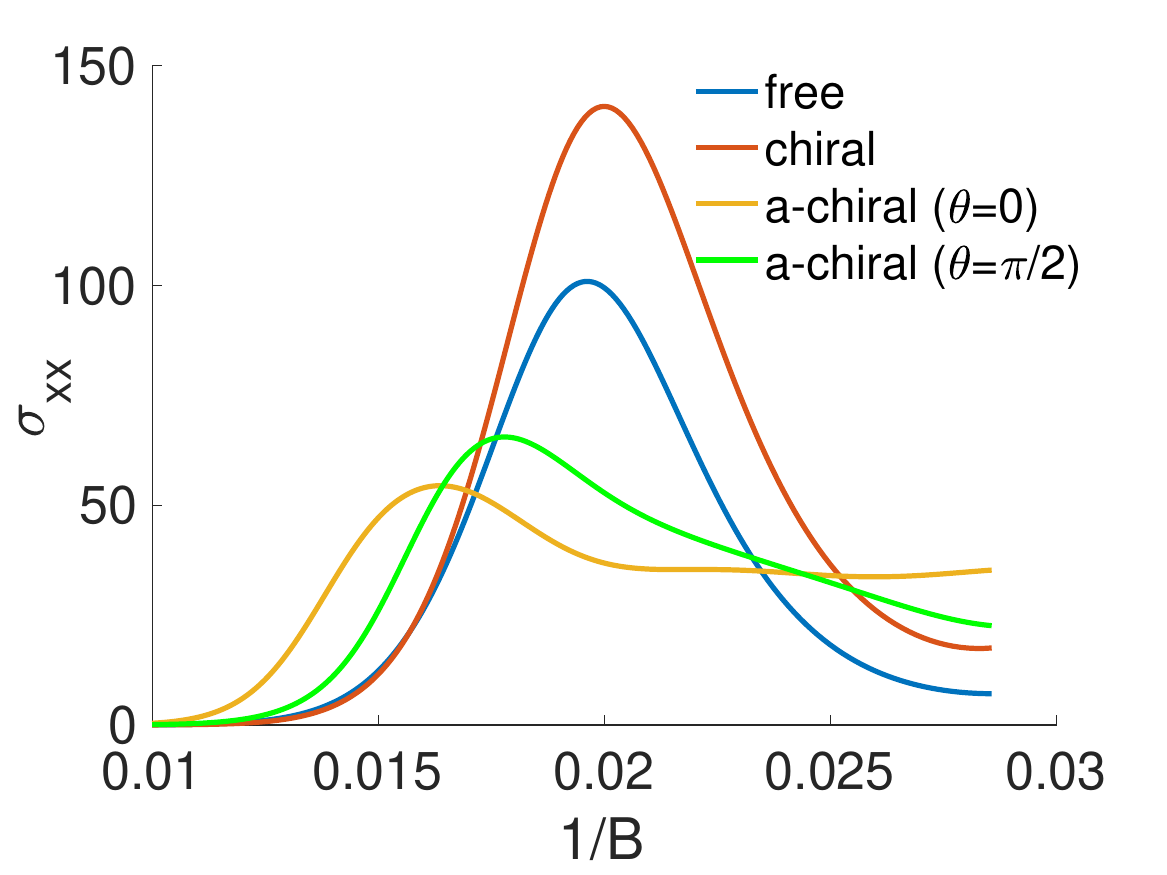}  \\
    \caption{Smoothed out longitudinal conductivity $\sigma_{xx} \propto -\rho(\lambda n'_{\beta}(\lambda-\mu))$ with $n_{\beta}$, the Fermi-Dirac distribution, showing Shubnikov-de Haas oscillations. On the left, $B=30$ and on the right $B=50$, both for $\beta=2.5.$ with $\alpha_i = 0.35.$}
    \label{fig:sigmaxx_B}
\end{figure}

\subsection{De Haas-van Alphen oscillations}

In 1930, de Haas and van Alphen who discovered that both the magnetization and the magnetic susceptibility of metals show an oscillatory profile as a function of $B^{-1}.$ This effect is called the de Haas-van Alphen (dHvA) effect. Even in the simpler case of graphene, both the experimental as well as theoretical foundations of that effect are not yet well-understood \cite{L11,KH14,SGB04}. One problem in
understanding the dHvA effect \cite{SGB04}, lies in the dependence of the chemical potential on the external magnetic field. To simplify mathematical analysis, it is more convenient to work in the grand-canonical ensemble, which is also discussed in \cite{CM01,SGB04,KF17}. The comparison with the canonical ensemble is made in this subsection as well.

 The grand thermodynamic potential for a DOS measure $\rho$, at inverse temperature $\beta,$ and field-independent chemical potential $\mu$ is defined as 
\[ \Omega_\beta(\mu,B):= (f_{\beta} * (\eta_N \rho))(\mu),\]
where $f_{\beta}(x):=-\beta^{-1} \log(e^{\beta x}+1).$ 
The magnetization $M$ and susceptibility $\chi$ are then in the grand-canonical ensemble defined as 
$$ M(\beta,\mu,B)  = - \frac{\partial \Omega_{\beta}(\mu,B) }{\partial B} \text{ and } \chi(\beta,\mu,B) = \frac{\partial M_{\beta}(\mu,B)}{\partial B}.$$
The susceptibility describes the response of a material to an external magnetic field. When $\chi>0$ the material is paramagnetic, when $\chi<0$ diamagnetic, and strongly enhanced $\chi \gg 1$ for ferromagnets.

While the approximation of computing the magnetization in the grand canonical ensemble is common, one should strictly speaking compute it in the canonical ensemble, instead.

\begin{figure}
    \centering
    \includegraphics[width=7.5cm]{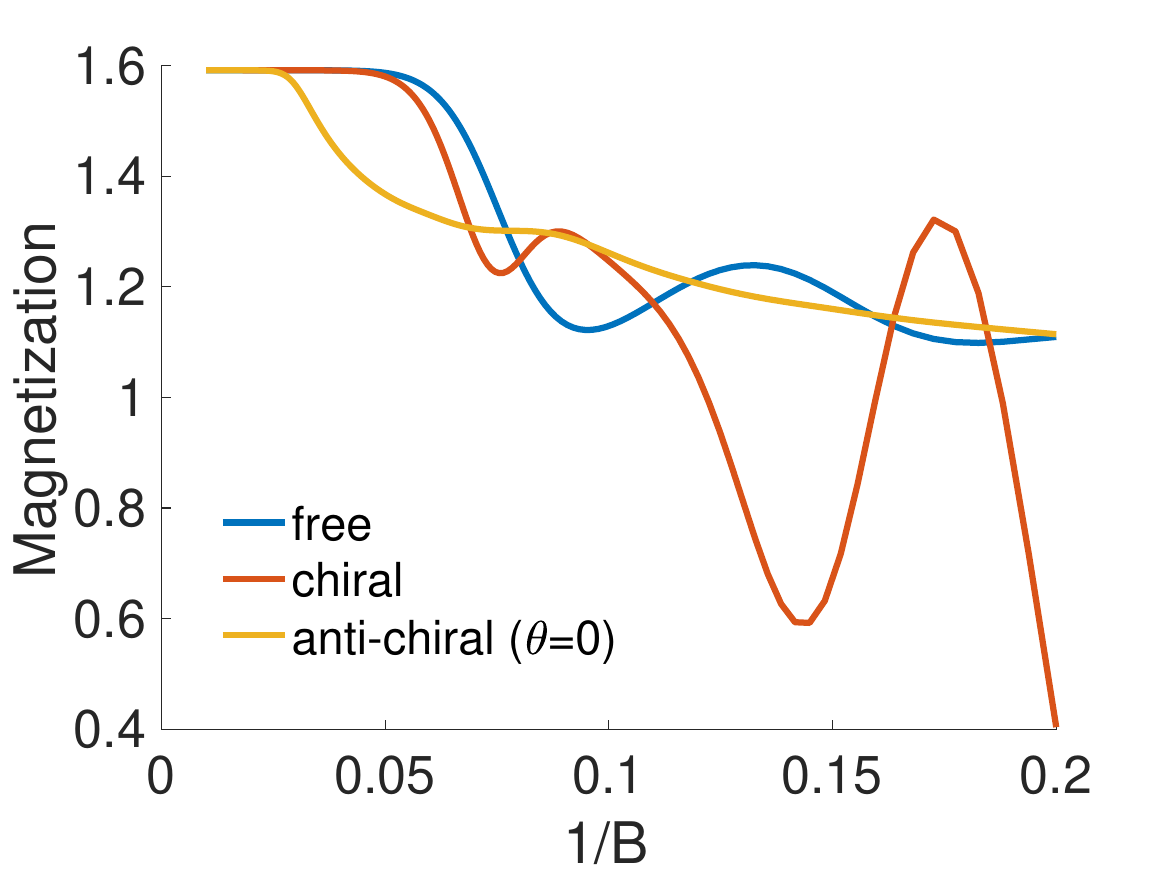}    \includegraphics[width=7.5cm]{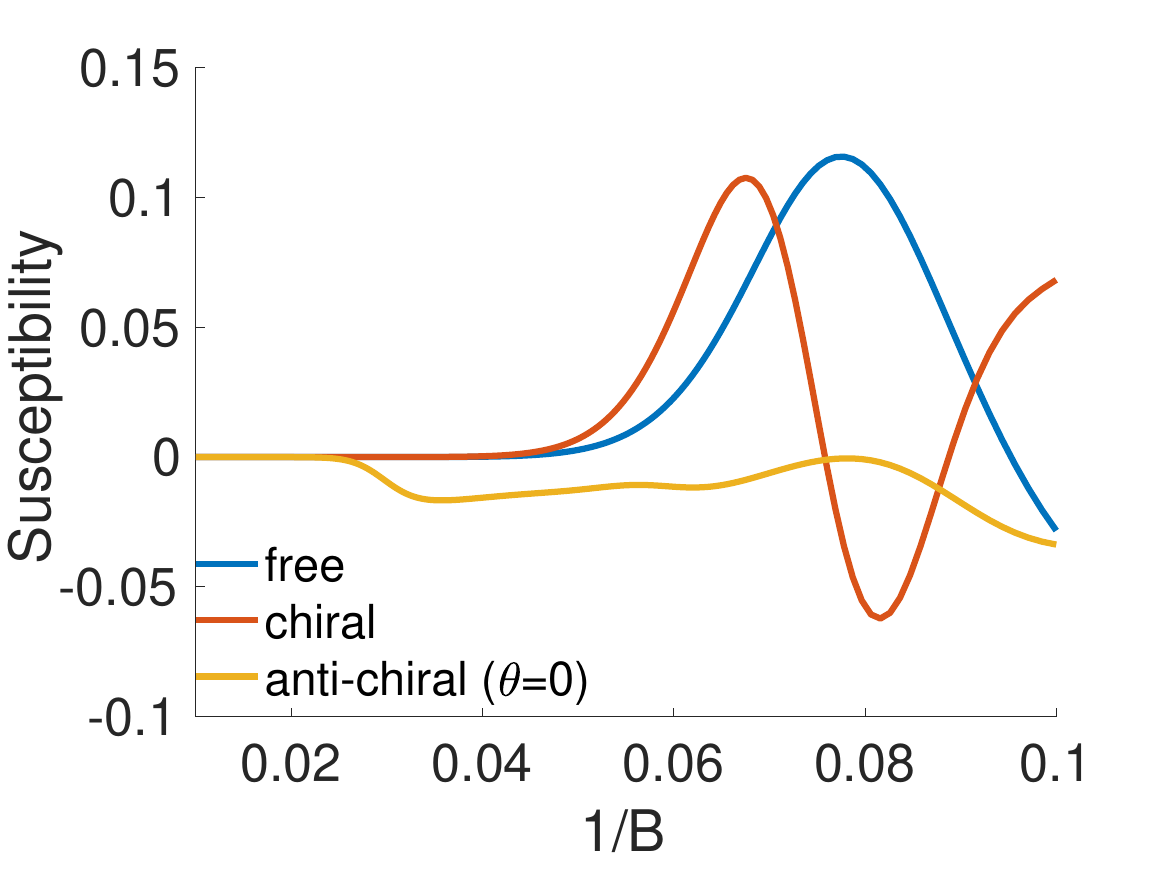}
     \caption{Magnetization and susceptibility for $\beta=4$, $\alpha_i=3/5,$ and chemical potential  $\mu=5$.}
    \label{fig:my_label}
\end{figure}

In this case, the charge density $\varrho$ given by the Fermi-Dirac statistics, with $n_{\beta}(x) := \frac{1}{e^{\beta x}+1}$, according to
\[ \varrho  =-\frac{\Omega_{\beta}(\mu,B)}{\partial \mu} = \rho(n_{\beta}(\cdot -\mu))\]
is fixed and the chemical potential becomes a function of $\rho$ and $B$. 

To see that this uniquely defines $\mu$ as a function of $\varrho$ and $B$ large enough, it is sufficient to observe that 
$$ \mu \mapsto \sum_{n \in \mathbb Z}(\eta_N n_{\mu})(\lambda_n \sqrt{B})$$
is a monotonically increasing function.
\begin{figure}[!ht]
    \centering
    \includegraphics[width=7.0cm]{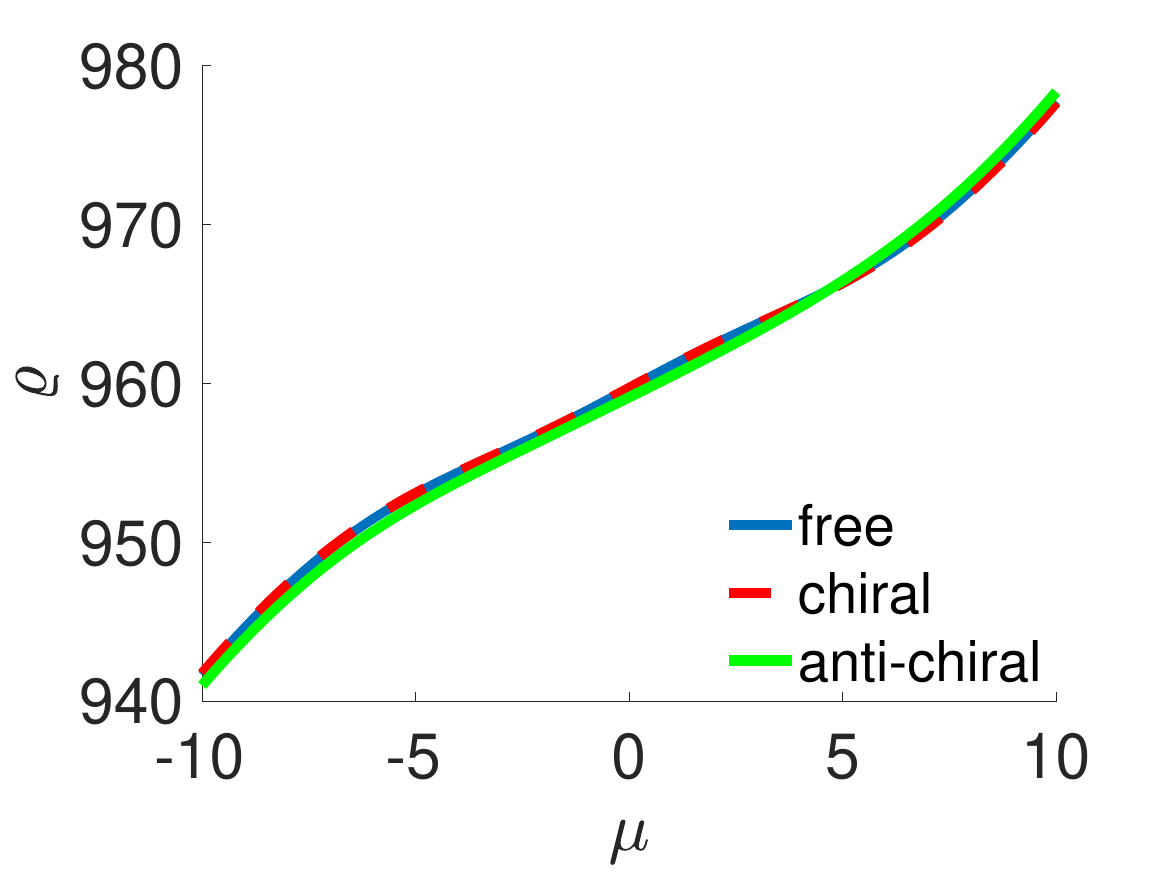}
     \includegraphics[width=7.0cm]{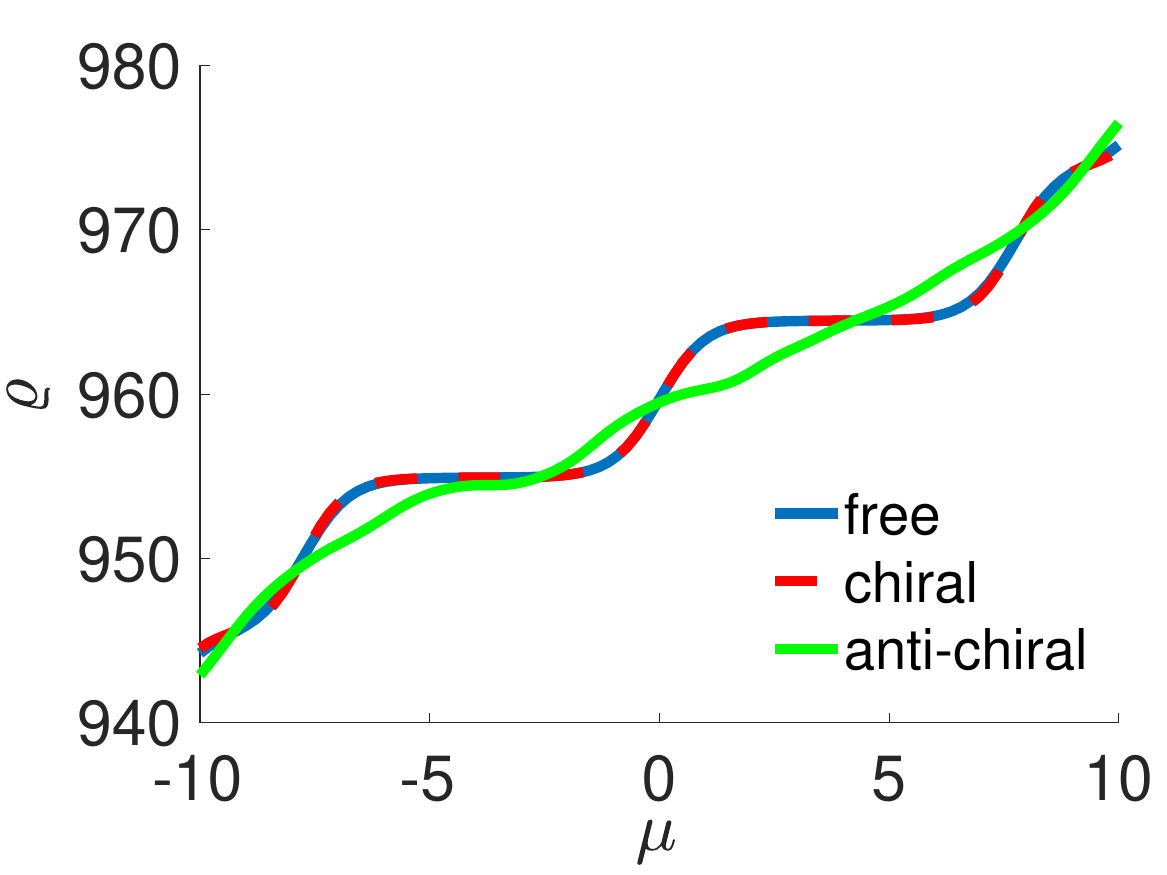}
    \caption{ Charge density with respect to chemical potential. Magnetic field $B=30$ for $\beta=1/2$ and $\beta=2$. We consider $100$ Landau levels around zero and an anti-chiral model with $\theta=0.$}
    \label{fig:my_label1}
\end{figure}
The Helmholtz free energy is then given as
$$F_{\beta}(\varrho,B) =\Omega_{\beta}(\mu(\varrho,B),B) + \mu(\rho,B)\varrho$$
with the magnetization given as the derivative $ M(\beta,\varrho,B) = -\frac{\partial F_{\beta}(\varrho,B)}{\partial B}.$
Hence, the magnetization in the canonical ensemble is also given by 
\[ M(\beta,\varrho,B) = - \frac{\partial \Omega_{\beta}(\mu,B)}{\partial B} \Big\vert_{\mu=\mu(\varrho,B)},\]
where the difference to the grand-canonical ensemble lies in the $B$-dependent chemical potential. The dHvA oscillations are shown in Figures \ref{fig:dHvA} and \ref{fig:my_label}, with the $AB^{\prime}$/$BA^{\prime}$ interaction leading to enhanced oscillations and the $AA^{\prime}$/$BB^{\prime}$ interaction damping the oscillations, compared to the non-interacting case.
\subsection{Quantum Hall effect}

\begin{figure}[!t]
\includegraphics[width=7.5cm, ]{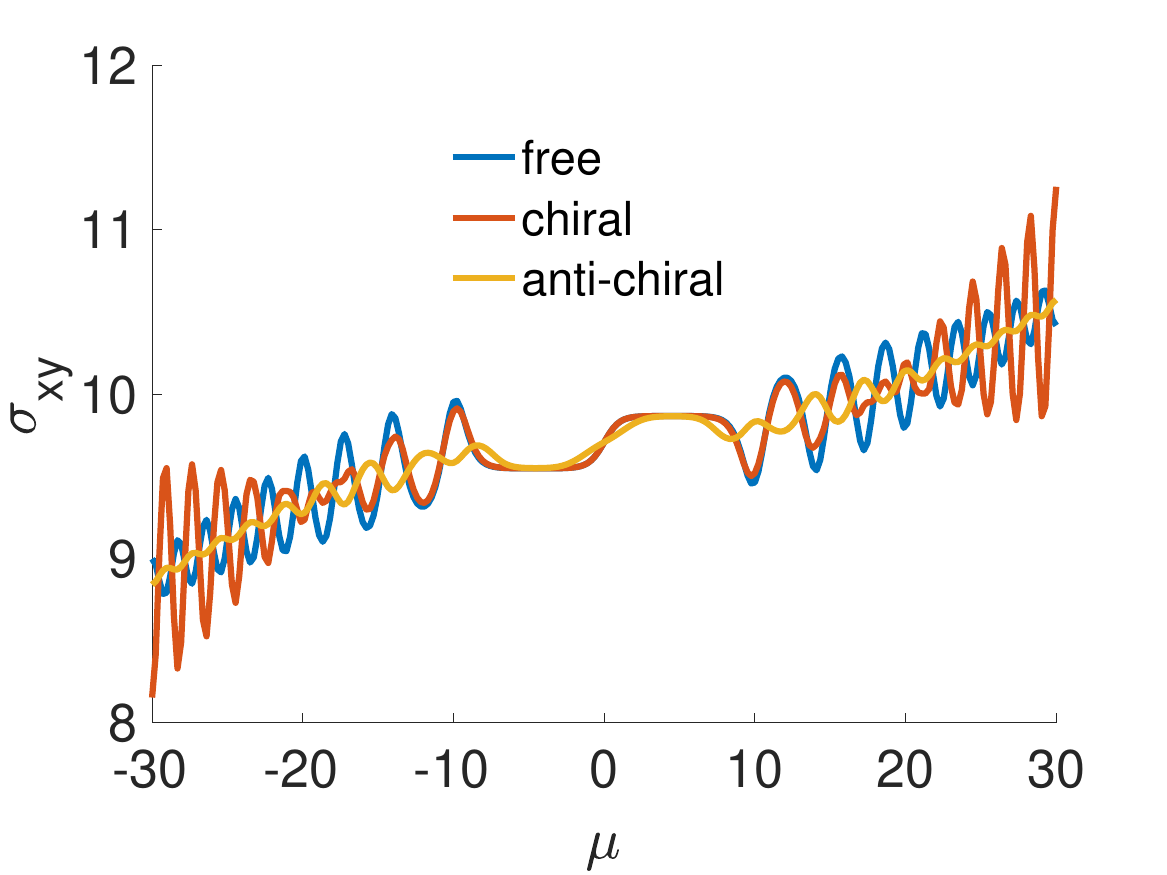}
\includegraphics[width=7.5cm, ]{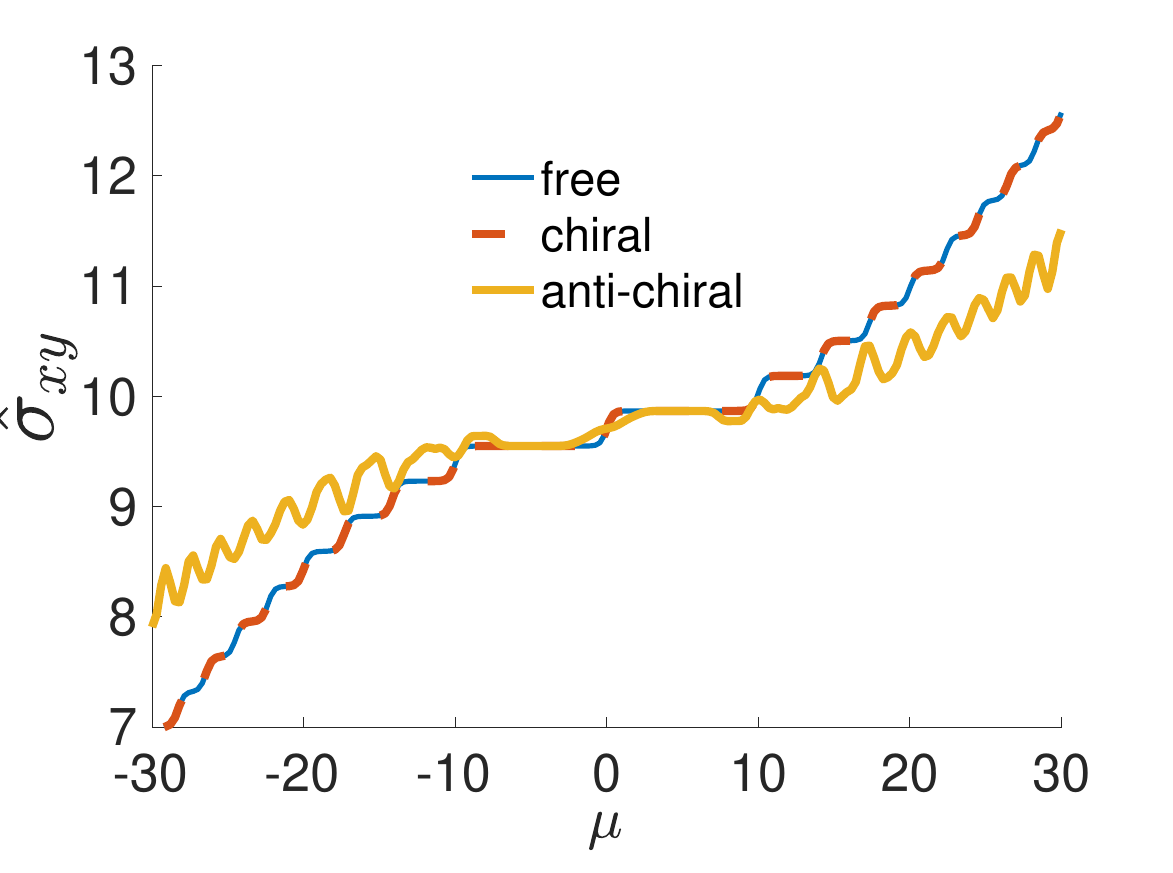}
\caption{Full quantum Hall conductivity \eqref{eq:Hall_cond} on the left with $\beta=2$, $B=40$ and on the right the high temperature conductivity \eqref{eq:high_Hall_cond} with $\beta=5$,$B=50$.}
\label{fig:QHE}
\end{figure}
The (transversal) quantum Hall conductivity $\sigma_{xy}$ is, by the Streda formula \cite[(16)]{MK12}, for a Fermi energy $\mu$ given by
\begin{equation}
\label{eq:Hall_cond}
     \sigma_{xy}(\beta,\mu,B) = \sum_{n=-N}^N \frac{\partial \rho(\eta_N n_{\beta}(\bullet-\mu))}{\partial B}.
\end{equation}

In case of the chiral Hamiltonian, the Gibbs factor $\gamma_{\beta, n}(\mu) = e^{\beta(\lambda_n \sqrt{B}-\mu)}$ allows us to write
\begin{equation}
    \begin{split}
&\sigma_{xy,c}(\beta,\mu,B) =\Bigg(\sum_{n=-N}^N \frac{n_{\beta}(\lambda_n\sqrt{B}-\mu)}{\pi}\left(1-\frac{\beta \lambda _n \sqrt{B}}{2}\gamma_{\beta, n}(\mu) n_{\beta}(\lambda_n \sqrt{B}-\mu)\right)\\
&+\sum_{n=-N}^N-\frac{\lambda_n\vert \lambda_n \vert^2 \beta^3\Ave(\mathfrak U)}{4\pi\sqrt{B}} n^4_{\beta}(\lambda_n\sqrt{B}-\mu)\left(\gamma_{\beta, n}(\mu)-4\gamma_{\beta, n}(\mu)^2+\gamma_{\beta, n}(\mu)^3\right)\Bigg) \left(1+o(1)\right)
\end{split}
\end{equation}

At very low temperatures, and $\mu$ well between two Landau levels, the contribution of the derivative of the Landau levels with respect to $B$ can be discarded. 

We then obtain the high-temperature limiting expression
\begin{equation} 
\label{eq:high_Hall_cond}
\hat{\sigma}_{xy,c}(\beta, \mu,B) := \sum_{n=-N}^{N} \frac{\eta_{\beta}(\lambda_{n,B}-\mu)}{\pi} \xrightarrow[\beta \to \infty]{} \vert \{ n ; \vert \lambda_{n,B}\vert \le \mu\}\vert
\end{equation}
as $n_{\beta}(\lambda_{n,B}-\mu) \to 1-H(\lambda_n \sqrt{B}-\mu)$ for $\beta  \uparrow \infty,$ where $H$ is the Heaviside function. 

This expression reveals the well-known staircase profile of the Hall conductivity which can already be concluded in this model in the $\beta \to \infty$ limit from Proposition \ref{prop:Sjostrand}.

For the $AA^{\prime}$/$BB^{\prime}$ interaction, the situation is rather different. Due to the broadening and splitting of the Landau levels, the staircase profile is less pronounced at non-zero temperature. Setting $\hat{\sigma}_{xy,ac}(\beta,\mu,B) :=t_{n,0}(n_{\beta}(\bullet-\mu)) -\frac{t_{n,1}(n_{\beta}(\bullet-\mu))}{2\sqrt{B}},$
where in the limit $\beta \to \infty,$ the second term vanishes, for $\mu$ away from the spectrum as $n_{\beta}'$ is a $\delta_0$ approximating sequence such that also in case of the $AA^{\prime}$/$BB^{\prime}$ interaction $\lim_{\beta \to \infty} \hat{\sigma}_{xy}(\beta,\mu,B) = \vert \{ n ; \vert \lambda_{n,B}\vert \le \mu\}\vert.$

\appendix
\section{Asymptotic expansion}
\label{sec:appendix1}
In this appendix, we shall prove Prop. \ref{prop: z_dependence_each_term} which, in particular, includes the proof of Lemma \ref{lemma: explicit_asymptotic_expansion_needed}. The quantization is as in Subsection \ref{ss: properties_of_effective_Hamiltonian}.

\begin{prop}
\label{prop: z_dependence_each_term}
Let $h_0$,  $\Epm$ be as in Lemma \ref{lemm:perturbed Grushin}. For $h\in [0,h_0)$, $|z|\leq 2\Vert \msW\Vert_\infty$, we have
  \begin{enumerate} 
      \item \label{part_one} The symbol $\frac{1}{\sqrt{h}}\Epm$ has an asymptotic expansion in $S$: There are $a_{n,j,k}\in S$ such that 
  \begin{equation}
  \label{eq: Epm_exp}
     \frac{1}{\sqrt{h}} \Epm(x_2,\xi_2;z,h)\sim \sum\limits_{j = 0}^\infty h^{\frac{j}{2}}E_{n,j}(x_2,\xi_2;z) \ \text{with} \  E_{n,j} = \sum\limits_{k = 0}^{j - 1} a_{n,j,k}(x_2,\xi_2)z^k, j\geq 1.
  \end{equation} 
  In particular, $E_{n,0} = z - z_{n,0}, \quad E_{n,1} = -z_{n,1},\quad E_{n,2} = - z_{n,2},$
  where $z_{n,j}$ are given in Lemma \ref{prop: expofQk}.
  \item \label{part_two} Let $0<\delta<1/2$, if $|\Im z|\geq h^{\delta}$, then $\sqrt{h}\Epm^{-1}$ has an asymptotic expansions in $S_\delta^\delta$: There are $b_{n,j,k,l}, c_{n,j,k}\in S$ such that in terms of $\prod\limits_{l = 0}^k b_{n,j,k,l}(x_2,\xi_2;z) = \sum\limits_{\alpha = 0}^{j +k - 2} z^\alpha c_{n,j,k}(x_2,\xi_2)$ the expansion of
  \begin{equation}
      \label{eq: Epm_inv}
      \begin{split}
     & \sqrt{h} \Epm^{-1}\sim \sum\limits_{j = 0}^\infty h^{\frac{j}{2}}F_{n,j} (x_2,\xi_2;z), \text{ with } F_{n,j} = \sum_{k = 0}^{j} (z - z_{n,0})^{-1} \prod\limits_{l = 0}^{k}\left(b_{n,j,k,l}(x_2,\xi_2;z)(z - z_{n,0})^{-1}\right).
  \end{split}
  \end{equation}
    Thus $h^{\frac{j}{2}} F_{n,j}\in S^{j(\delta - \frac{1}{2})+ \delta}$. In particular, we have
\begin{equation}
    \label{eq: Epm_inv_explicit}
    \begin{split}
        &F_{n,0} =(z - z_{n,0})^{-1}, \ F_{n,1} = F_{n,0}z_{n,1}F_{n,0},\ F_{n,2} = F_{n,0}\left(z_{n,1}F_{n,1} + z_{n,2}F_{n,0} - \frac{\{F_{n,0}, z - z_{n,0}\}}{2i} \right), 
    \end{split}
\end{equation}
    where $\{\cdot, \cdot\}$ is the Poisson bracket.
    \item \label{part_three} Let $0<\delta<1/2$, if $|\Im z|\geq h^{\delta}$, then $r_n$ has an  asymptotic expansions in $S_\delta^\delta$: There are $d_{n,j,k,l}(x_2,\xi_2;z), e_{n,j,k,\alpha}(x_2,\xi_2)\in S$, such that in terms of $\prod\limits_{l = 0}^k d_{n,j,k,l}(x_2,\xi_2;z) = \sum\limits_{\alpha = 0}^{j+k -2} z^\alpha e_{n,j,k,\alpha}(x_2,\xi_2)$
  \begin{equation}
      \label{eq: r_n_exp}
      \begin{split}
        &\quad r_n(x_2,\xi_2;z,h) \sim \sum\limits_{j =0}^\infty h^{\frac{j}{2}} r_{n,j}(x_2,\xi_2;z,h),
        \text{with~} r_{n,j} = \sum\limits_{k = 0}^j(z - z_{n,0})^{-1} \prod\limits_{l = 0}^k \left(d_{n,j,k,l}(x_2,\xi_2;z)(z - z_{n,0})^{-1}\right).
        \end{split}
  \end{equation}
 Thus $h^{\frac{j}{2}}r_{n,j}\in S_0^{(J+1)\delta - \frac{J}{2}}$. In particular, 
\begin{equation}
    \label{eq: r_n_exp_explicit}
    r_{n,0} = F_{n,0},\quad r_{n,1} = F_{n,1},\quad r_{n,2} = F_{n,2} - (\partial_z z_{n,2})F_{n,0}.
\end{equation}
\item \label{part_four} Finally, let $\eta = x_2 + i\xi_2$, then the leading terms of $\Tr_{\CC^2}(r_n)$ are:
 \begin{equation}
    \begin{split}
      \operatorname{Chiral} \ \msH_{\ch, n}: & \Tr_{\CC^2}(r_{\ch, n,0} + h^{\frac{1}{2}}r_{\ch, n,1} + hr_{\ch,n,2}) = \frac{2}{z} + 0 + \frac{\lambda_n^2}{z^3}\mathfrak U(\eta)h, \\
      \operatorname{Anti-Chiral} \ \msH_{\ach,n}^\theta: & \Tr_{\CC^2}(r_{\ach, n,0} + h^{\frac{1}{2}}r_{\ach, n,1}) = \frac{2z}{z^2 - c_n^2} + \frac{2s_n^2(z^2+c_n^2)}{(z^2 - c_n^2)^2}\sqrt{h} ,
    \end{split}  
  \end{equation}
   where $\mathfrak U(\eta) = \frac{\alpha_1^2}{8}\left[\alpha_1^2(|U_-(\eta)|^2 - |U(\eta)|^2)^2 + 4|\partial_{\bar{\eta}}\overline{U_-(\eta)}-\partial_\eta U(\eta)|^2\right]$, $\partial_\eta = \frac{1}{2}(\partial_{x_2} - i\partial_{\xi_2})$, $s_n(\eta) = \begin{cases}
    \alpha_0  \sin(\thot) |V(\eta)| & n\neq 0 \\
    \alpha_0 |V(\eta)| & n = 0,
  \end{cases}$ and $c_n(\eta) = \begin{cases}
    \alpha_0 \cos(\thot)|V(\eta)| & n \neq 0\\
    \alpha_0 |V(\eta)| & n = 0.
  \end{cases}$
    \end{enumerate}
\end{prop}

We will prove Proposition \ref{prop: z_dependence_each_term} in the rest of this appendix in two steps: First, we compute explicitly the leading terms (three terms for the chiral model, two for anti-chiral model) in the expansion of $Z_n(x_2,\xi_2;z,h)$, the symbol of $Z_n^W$, where $\Epm = \sqrt{h}(z - Z_n^W)$ by \eqref{eq: epm}. Then, we exhibit the $z$ dependence for each term in the expansion of $\Epm$, from which we build up both the legitimacy of the existence of asymptotic expansions of $\Epm^{-1}$ and $r_n$, and the $z$ dependence of each term in the expansions.

\smallsection{Explicit leading terms}
Recall that by \eqref{eq: epm} and \eqref{eq: W_to_tilde_V}, $\Epm = \sqrt{h}(z - Z_n^W)$ with
\begin{equation}
    \label{eq: Z_n}
      \begin{split}
    Z_n^W &(x_2,hD_{x_2};h) = R_n^+\tilde{\msV}^W(I + \sqrt{h}E_{0,n}^\theta\tilde{\msV}^W)^{-1}R_n^-
   \\ &= \sum\limits_{k = 0}^\infty h^{\frac{k}{2}}(-1)^k R_n^+\tilde{\msV}^W(E_{0,n}^\theta\tilde{\msV}^W)^k R_n^-
    =: \sum\limits_{k = 0}^\infty h^{\frac{k}{2}}Q_{n,k}^W(x_2,hD_{x_2};h),
  \end{split}
\end{equation}
where $R_n^\pm$, $E_\on^\theta$, $\tilde{\msV}^W$ are given in \eqref{eq:R_pm}, \eqref{eq: ent} and \eqref{eq: W_to_tilde_V}. Then we can express the asymptotic expansion of $Z_n(x_2,\xi_2)$ in terms of $Q_{n,k}(x_2,\xi_2)$:

\begin{prop}
  \label{prop: expofQk}
Let $Q_{n,k}^W(x_2,hD_{x_2};h) = (-1)^k R_n^+\tilde{\msV}^W(E_{0,n}^\theta\tilde{\msV}^W)^k R_n^-$. Then symbols $Q_{n,0}$, $Q_{n,1}$, $Q_{n,2}$ have the following asymptotic expansions 
\[
  \begin{split}
    &Q_{n,0}(x_2,\xi_2;h) = Q_{n,0}^\upzero(x_2,\xi_2) + \sqrt{h}Q_{n,0}^\upone(x_2,\xi_2) + hQ_{n,0}^\uptwo(x_2,\xi_2) +\mathcal O_{S}(h^{\frac{3}{2}}),\\
    &Q_{n,1}(x_2,\xi_2;h) = Q_{n,1}^\upzero(x_2,\xi_2) +\sqrt{h}Q_{n,1}^\upone(x_2,\xi_2) +\mathcal O_{S}(h),\\
    &Q_{n,2}(x_2,\xi_2;h) = Q_{n,2}^\upzero(x_2,\xi_2) +\mathcal O_{S}(\sqrt{h}).
  \end{split}
\]
For the chiral Hamiltonian, with $\eta = x_2\ + i\xi_2$, $D_\eta = \frac{1}{2}(D_{x_2} -i D_{\xi_2})$,
\[
\begin{split}
    &Q_{\ch,n,0}^{\upzero} = Q_{\ch,n,0}^{\uptwo} = Q_{\ch,n,2}^{\upzero} = 0, \quad Q_{\ch,n,1}^{\upzero} = -\frac{\alpha_1^2\lambda_n}{4}\left[|U|^2 - |U_-|^2\right] \sigma_3,\\
    &Q_{\ch,n,0}^{\upone} = \frac{\lambda_n\alpha_1}{2} \begin{pmatrix}
      0 & D_\eta U - D_{\bar{\eta}}\overline{U_-}\\
      D_\eta U_- - D_{\bar{\eta}} \overline{U} & 0 
    \end{pmatrix},\\
    &Q_{\ch,n,1}^{\upone} = \begin{cases}
      -\frac{\alpha_1^2z}{4}\left[2|n|(|U|^2+|U_-|^2)\indic_{2\times 2} + (|U|^2 - |U_-|^2)\sigma_3\right] & n \neq 0, \\
      -\frac{\alpha_1^2z}{2}\begin{pmatrix}
        |U|^2 & 0 \\ 0 & |U_-|^2
      \end{pmatrix} & n = 0.
    \end{cases}
\end{split}
\]
Analogously, for the anti-chiral Hamiltonian, when $\msH^\theta = \msH_\ach$, we have, when $n = 0$, 
\[
\begin{split}
    &Q_{\ach, 0, 0}^\upone = Q_{\ach, 0, 1}^\upzero = Q_{\ach, 0, 1}^\upone = Q_{\ach, 0, 2}^\upzero = 0, \\
    &Q_{\ach, 0, 0}^\upzero = \alpha_0 \begin{pmatrix}
        0 & e^{-\frac{\theta}{2}i}V \\ e^{\frac{\theta}{2}} V^* & 0 
      \end{pmatrix}, Q_{\ach, 0, 0}^\uptwo =\frac{\alpha_0}{4} \begin{pmatrix}
        0 & e^{-\frac{\theta}{2}i}\Delta_{x_2,\xi_2} V\\
        e^{\frac{\theta}{2}i}\Delta_{x_2,\xi_2} \bar{V}
      \end{pmatrix},
\end{split}
\]
when $n \neq 0$, 
\[
\begin{split}
    &Q_{\ach, n,0}^\upone = 0, \ Q_{\ach, n,1}^\upzero =\frac{\alpha_0^2|V|^2\sin^2(\tfrac{\theta}{2})}{2\lambda_n}\indic_{2\times 2}, \ Q_{\ach, n,1}^\upone = -\frac{z\alpha_0^2|V|^2\sin^2(\tfrac{\theta}{2})}{4\lambda_n^2}\indic_{2\times 2},\\
    &Q_{\ach, n,0}^\upzero = \alpha_0 \cos(\tfrac{\theta}{2})\begin{pmatrix}
        0 & V^* \\ V & 0 
      \end{pmatrix}, Q_{\ach, n,2}^\upzero = -\frac{\alpha_0^3|V|^2\sin^2(\tfrac{\theta}{2})\cos(\tfrac{\theta}{2})}{4\lambda_n^2}\begin{pmatrix}
      0 & V \\ V^* & 0
    \end{pmatrix},\\
    & Q_{\ach, n,0}^\uptwo = \frac{\alpha_0}{4}\left(2|n|\cos(\tfrac{\theta}{2}) - i\sigma_3\sin(\tfrac{\theta}{2})\right)\begin{pmatrix}
      0 & \Delta_{x_2,\xi_2} V\\ \Delta_{x_2,\xi_2} \bar{V} & 0 
    \end{pmatrix}.
\end{split}
\]
In particular, $Z_n$ has an asymptotic expansion $Z_n\sim \sum\limits_{k = 0}^\infty h^{\frac{k}{2}} z_{n,k}$ in $S$ with  
\[z_{n,0} = Q_{n,0}^{\upzero}, \ z_{n,1} = Q_{n,1}^{\upzero} + Q_{n,0}^{\upone}, \ 
    z_{n,2} = Q_{n,2}^{\upzero} + Q_{n,1}^{\upone} + Q_{n,0}^{\uptwo}.\]
\end{prop}

\begin{proof}
$Q_{n,k}$ has the symbol $
    Q_{n,k}(x_2,\xi_2) = (-1)^k \int_{\RR_{x_1}}(K_n^\theta(x_1))^* \tilde{\msV}^w \#(E_\on^\theta \tilde{\msV}^w)^{\#k} K_n^\theta(x_1) dx_1.$
Recall that by\eqref{eq:contmodel},  \eqref{eq: Knt}, and \eqref{eq: ent}, we have 
\[
  K_n^\theta = \begin{pmatrix}
  u_n^\theta & 0 \\ 0 & u_n^{-\theta}
\end{pmatrix},\quad \msV = \begin{pmatrix}
  0 & T \\ T^* & 0 
\end{pmatrix}, \quad T = \begin{pmatrix}
  \alpha_0 V & \alpha_1 \overline{U_-}\\ \alpha_1 U & \alpha_0 V
\end{pmatrix}, \quad \Ent = \begin{pmatrix}
    \ent & 0 \\ 0 & \ennt
  \end{pmatrix}.
\] 
Thus, inserting the above expressions into the definition of $Q_{n,k}$, we find for its symbol
\[
  Q_{n,k} =\int \begin{pmatrix}
    \unt^* & 0 \\ 0 & (\unnt)^*
  \end{pmatrix} \begin{pmatrix}
    0 & \tilde{T}^w \\ (\tilde{T}^w)^* & 0 
  \end{pmatrix} \left( \begin{pmatrix}
    \ent & 0 \\ 0 & \ennt
  \end{pmatrix}\begin{pmatrix}
    0 & \tilde{T}^w \\ (\tilde{T}^w)^* & 0 
  \end{pmatrix}\right)^k \begin{pmatrix}
    \unt^* & 0 \\ 0 & \unnt
  \end{pmatrix} \frac{dx_1}{(-1)^k},
\]
where $\tilde{T}^w = T^w(\w)$. In particular,
\begin{equation}
  \label{eq: Q_i}
  \begin{split}
    &Q_{n,0}= \begin{pmatrix}
      0 & \int (\unt)^* \tilde{T}^w \unnt dx_1\\ \int (\unnt)^* (\tilde{T}^w)^* \unt dx_1
    \end{pmatrix},\\
    &Q_{n,1} = \begin{pmatrix}
      -\int (\unt)^* \tilde{T}^w \ennt (\tilde{T}^w)^* \unt dx_1 & 0 \\ 0 & -\int (\unnt)^* (\tilde{T}^w)^* \ent  \tilde{T}^w  \unnt dx_1
    \end{pmatrix}, \text{ and }\\
    &Q_{n,2} = \begin{pmatrix}
      0 & \int (\unt)^* \tilde{T}^w \ennt (\tilde{T}^w)^*\ent \tilde{T}^w \unnt dx_1\\
      \int (\unnt)^* (\tilde{T}^w)^*\ent \tilde{T}^w \ennt (\tilde{T}^w)^*\unt dx_1 & 0
    \end{pmatrix}.
  \end{split}
\end{equation}
Notice that since both $\tilde{T}^w$ and $\ent$ depend on $h$, we need to further expand them in order to obtain asymptotic expansions of $Q_{n,k}$. Thus the proof of Proposition \ref{prop: expofQk} rests now on the following two lemmas.
\begin{lemm} [Expansion of $\tilde{T}^w$ and $\ent$]\leavevmode
  \label{lemma: Twexpansion}
  \begin{enumerate}
    \item Let $T\in C^\infty_b(\RR_x^2)$. Recall the definition $\tilde T(x,\xi) := T(x_2+h^{\frac{1}{2}}x_1, \xi_2 - h^{\frac{1}{2}} \xi_1)\in S(\RR_{x,\xi}^4)$. Then
    \begin{equation}
      \label{eq: Twexpansion}
      \begin{split}
        \tilde{T}^w(x,D_{x_1},\xi_2) = &T(x_2,\xi_2) + \sqrt{h}\langle \nabla_{x_2,\xi_2}T(x_2,\xi_2), (x_1,-D_{x_1}) \rangle \\
         + & \frac{h}{2}\langle (x_1,-D_{x_1}), \operatorname{Hess}T (x_2,\xi_2)(x_1,-D_{x_1})^T \rangle +\mathcal O_{S(\RR_{x_2,\xi_2}^2;\mathcal L(B^3_{x_1};B^0_{x_1}))}(h^{\frac{3}{2}})
      \end{split}
    \end{equation}

    \item   Let $\ent$ be as in \eqref{eq: ent}. Then $\ent(x,D_{x_1},\xi_2)$ has an asymptotic expansion 
    $\ent \sim \sum\limits_{k = 0}^\infty h^{\frac{k}{2}}\sigma_k(e_\on^\theta)$ where 
        $
            \sigma_k(\ent) = \sum\limits_{m\neq n} \frac{z^k u_m^\theta(u_m^\theta)^*}{(\lambda_m - \lambda_n)^{k+1}}.$
  \end{enumerate}
\end{lemm}

\begin{lemm}[Projections]
  \label{lemma: projections}
  Let $S_n^\theta = \operatorname{span}\{u_n^\theta,u_{-n}^\theta\}$ with $S_n:=S_n^{0}$ and $u_n:=u_n^{0}.$ The following properties hold:
  \begin{enumerate}
    \item Reflection invariance with respect to $\theta$ such that $S_n^\theta = S_n^{-\theta}$, in particular $
              \unt = \cos\left(\thot\right) \unnt + i\sin\left(\thot\right) \unnnt .$
    \item  Let $M = \begin{pmatrix}
      0 & \alpha \\ \beta & 0
    \end{pmatrix} \in \mathbb C^{2\times 2}$ then $Mu_n\in S_{n-1}\cup S_{n+1}$, for any $n\geq 0$. More specifically for $\theta=0$
          \begin{equation}
          \begin{split}
            \label{eq: Tu_n}
              Mu_{\pm n} &= \frac{\alpha i}{2}(u_{n+1} - u_{-(n+1)}) {\mp} \frac{\beta i}{2}(u_{n-1} + u_{-(n-1)}), \text{ for }n \ge 2\\
              Mu_{\pm 1} &= \frac{\alpha i}{2}(u_2 - u_{-2}) \pm \frac{\beta}{\sqrt{2}} u_0,\text{ and }
              Mu_0 = \frac{\alpha}{\sqrt{2}}(u_1 - u_{-1}).
        \end{split}
        \end{equation}
    \item We have $x_1\unt\in S_{n-1}^\theta\cup S_{n+1}^\theta$, $D_{x_1}\unt\in S_{n-1}^\theta\cup S_{n+1}^\theta$. More specifically 
        \begin{equation}
              \begin{split}
               x_1u_{\pm n}^{\theta} &= \frac{\sqrt{2}}{4}[u_{n-1}^\theta(\sqrt{n}\pm +\sqrt{n-1})+u_{-(n-1)}^\theta(\sqrt{n} \mp \sqrt{n-1})\\
                  &+u_{n+1}^\theta(\sqrt{n+1} +\sqrt{n})\pm u_{-(n+1)}^\theta(\sqrt{n+1}\mp \sqrt{n})],\text{ for } \vert n \vert\ge 2\\
          x_1 u_{\pm 1}^\theta &= \frac{i}{2} u_0^\theta + \frac{\sqrt{2}}{4} [u_2^\theta(\sqrt{2}\pm\sqrt{1}) + u_{-2}^\theta(\sqrt{2} \mp \sqrt{1})]\text{ and }
          x_1 u_0^\theta = \frac{\sqrt{2}i}{4}(u_1^\theta + u_{-1}^\theta).
      \end{split}
      \end{equation}
  \end{enumerate}
\end{lemm}
\begin{proof}We omit the proof of this Lemma here as it follows from straightforward but lengthy basis expansions and the simple observation that $\langle u_m^{-\theta}, u_n^\theta\rangle = \cos\left(\thot\right) \delta_{m,n} + i \sin\left(\thot\right) \delta_{m,-n}.$
  \end{proof}
From the preceding Lemmas \ref{lemma: Twexpansion} and 
  \ref{lemma: projections}, we can compute the asymptotic expansion of each term of $Q_{n,k}$ in \eqref{eq: Q_i} and therefore prove  Prop.\ref{prop: expofQk}.

For the $(1,2)$-entry of $Q_{n,0}$, by Lemma \ref{lemma: Twexpansion}, we have
\[
  \begin{split} \int (\unt)^*\tilde{T}^w\unnt dx_1 =& \int (\unt)^* T \unnt dx_1 + \sqrt{h}\int (\unt)^* \langle \nabla_{x_2,\xi_2}T, (x_1,-D_{x_1}) \rangle\unnt dx_1 \\
    &+ \frac{h}{2}\int (\unt)^* \langle (x_1,-D_{x_1}), \operatorname{Hess} T(x_2,\xi_2)(x_1,-D_{x_1})^T \rangle\unnt dx_1\\
    =: &t_{n,0}^\upzero +\sqrt{h} t_{n,0}^\upone +h t_{n,0}^\uptwo + \mathcal O_{S(\RR_{x_2,\xi_2}^2;\CC_{2\times 2})}(h^{\frac{3}{2}}).
  \end{split}
\]
Specializing now to the chiral case, in which case the $\theta$-dependence can be gauged away, we choose $$T(x_2,\xi_2) = \begin{pmatrix} 0 & \alpha_1\overline{U(x_2,\xi_2)} \\ \alpha_1 U_-(x_2,\xi_2) & 0 \end{pmatrix}$$
where in the chiral case, by Lemmas \ref{lemma: Twexpansion} and \ref{lemma: projections}, we see that 
\[
  \begin{split}
  &t_{\ch, n,0}^{\upzero} =  0, \ t_{\ch, n,0}^{\upone} =\frac{\lambda_n\alpha_1 i}{2}(\partial_{\bar{w}}\overline{U_-} - \partial_w U),\text{ and  } t_{\ch, n,0}^{\uptwo}  = 0,
  \end{split}
\]
while in the anti-chiral case, choosing $T(x_2,\xi_2) =\alpha_0 V(x_2,\xi_2) \operatorname{id}_{\CC_{2\times 2}}$
\[
  \begin{split}
    &t_{\ach, n,0}^{\upzero} = \begin{cases}
      \alpha_0\cos(\thot)V & n\neq 0,\\
      \alpha_0 e^{-\frac{\theta}{2}i} V & n = 0,
    \end{cases}, \quad t_{\ach, n,0}^{\upone}  = 0, \text{ and }t_{\ach, n,0}^{\uptwo} = \begin{cases}
      \frac{\alpha_0}{4}(2|n|\cos(\thot) - i\sigma_3\sin(\thot))\Delta_{x_2,\xi_2}V, &n\neq 0\\
      \frac{\alpha_0}{4}e^{-i\frac{\theta}{2}}\Delta_{x_2,\xi_2}V & n = 0.
    \end{cases}
    \end{split}
\]
Due to the conjugacy relation $\int (\unt)^* (\tilde{T}^w)^*\unnt dx_1 = (\int (\unnt)^* \tilde{T}^w \unt dx_1)^*$, the expansion of $Q_{n,0}^\theta$ follows by \eqref{eq: Q_i}.

Similarly for the $(1,1)$-entry $Q_{n,1}^\theta$, denote
\[
  -\int (\unt)^* \tilde{T}^w \ennt (\tilde{T}^w)^* \unt dx_1 =: t_{n,1}^{\upzero} + t_{n,1}^\upone \sqrt{h}+\mathcal O_{S(\RR_{x_2,\xi_2}^2;\CC_{2\times 2})}(h)
\]
where, using Lemma \ref{eq: Twexpansion}, in the chiral case,
\[
  \begin{split}
    &t_{\ch, n,1}^{(0)} = -\frac{\alpha_1^2\lambda_n}{4}(|U|^2 - |U_-|^2)\text{ and }t_{\ch, n,1}^{\upone} = \begin{cases}
      -\frac{\alpha_1^2z}{4}[2|n|(|U|^2+|U_-|^2)+(|U|^2 - |U_-|^2)], & n \neq 0\\
      -\frac{\alpha_1^2z}{2}|U|^2, & n = 0
    \end{cases}
  \end{split}
\]
and in the anti-chiral case
\[
  \begin{split}
    &t_{\ach, n,1}^{\upzero} = \begin{cases}
      \frac{\alpha_0^2|V|^2\sin^2(\thot)}{2\lambda_n}, & n \neq 0\\
      0 , & n = 0
    \end{cases}\text{ and }t_{\ach, n,1}^{\upone} = \begin{cases}
      -\frac{\alpha_0^2|V|^2\sin^2(\thot)z}{4\lambda_n^2}, & n \neq 0\\
      0, & n = 0.
    \end{cases}
  \end{split}
\]

In a similar fashion, the $(2,2)$-entry of $Q_{n,1}$, defined in \eqref{eq: Q_i}, can be obtained by precisely the same computations after only replacing $\theta$ by $-\theta$ and $T^*$ by $T$, i.e. $U$ switching with $U_-$ and using $V^*$ instead of $V$. Thus the asymptotic expansion of $Q_{n,1}^\theta$ follows.

Similarly for $Q_{n,2}^\theta$ we restrict us to the $(1,2)$ entry in \eqref{eq: Q_i}. Then, we denote 
\[
  \int (\unt)^* \tilde{T}^w \ennt (\tilde{T}^w)^*\ent \tilde{T}^w \unnt dx_1 =: t_{n,2}^{\upzero} + \mathcal O_{S(\RR_{x_2,\xi_2}^2;\CC_{2\times 2})}(\sqrt{h}).
\]
It follows then by Lemma \ref{eq: Twexpansion}, that in the chiral model, $t_{\ch, n,2}^{\upzero} = 0$
while in the anti-chiral model, $t_{n,2}^{\upzero} = -\frac{\alpha_0^3|V|^2\sin^2(\thot)\cos(\thot)}{4\lambda_n^2}V. $
By the conjugacy relation 
\[\int (\unnt)^* (\tilde{T}^w)^*\ent \tilde{T}^w \ennt (\tilde{T}^w)^*\unt dx_1  = \Bigg(\int (\unt)^* \tilde{T}^w \ennt (\tilde{T}^w)^*\ent \tilde{T}^w \unnt dx_1\Bigg)^*,\] this also yields directly the expansion of $Q_{\ach, n,2}^\theta$.
\end{proof}

\smallsection{Existence, derivation and $z$-dependence}
Now we prove the rest of Prop. \ref{prop: z_dependence_each_term}, which includes the existence and derivation of asymptotic expansion of $\Epm^{-1}$ and $r_n$ and the $z$ dependence of each terms in the expansions of $\Epm$, $\Epm^{-1}$ and $r_n$.
\begin{proof}[Proof of Prop. \ref{prop: z_dependence_each_term}]
 By \eqref{eq: Z_n} and Prop. \ref{prop: expofQk}, $\Epm = \sqrt{h}(z - Z_n)$, and $Z_n$ has an asymptotic expansion in $S$. Thus, $\frac{1}{\sqrt{h}}\Epm$ also has an asymptotic expansion in $S$: $\frac{1}{\sqrt{h}}\Epm\sim \sum\limits_j h^{\frac{j}{2}}E_{n,j}$ with $E_{n,j}\in S$. To exhibit the $z$-dependence, we notice that only $E_\on$ depends on $z$ in \eqref{eq: Z_n}. Thus, by \eqref{eq: everything_in_S}, we have
 \[
 \begin{split}
     Z_n^W  &= R_n^+\tilde{\msV}^W(\indic + \sqrt{h}E_\on\tilde{\msV}^W)^{-1} R_n^-= R_n^+ \tilde{\msV}^W R_n^- + \sum\limits_{\alpha = 1}^\infty R_n^+ \tilde{\msV}^W (\sqrt{h}E_\on \tilde{\msV}^W)^\alpha R_n^-\\
     &= R_n^+ \tilde{\msV}^W R_n^- + \sum\limits_{\alpha = 1}^\infty h^{\frac{\alpha}{2}} R_n^+ \tilde{\msV}^W \left[\sum\limits_{m \neq n} \frac{K_m^\theta (K_m^\theta)^*}{\lambda_m - \lambda_n} \sum\limits_{\beta = 0}^\infty \left(\frac{\sqrt{h}z}{\lambda_m - \lambda_n}\right)^\beta \right]^\alpha R_n^-\\
     &= R_n^+ \tilde{\msV}^W R_n^- + \sum\limits_{\alpha = 1}^\infty \sum\limits_{\gamma = 0}^\infty  h^{\frac{\alpha + \gamma}{2}} z^{\gamma} A_{n,\alpha,\gamma}^W(x_2,hD_{x_2})= R_n^+\tilde\msV^WR_n^- - \sum\limits_{j = 1}^\infty h^{\frac{j}{2}} \left(\sum\limits_{k = 0}^{j - 1} z^k a_{n,j,k}^W(x_2,hD_{x_2})\right)
 \end{split}
 \]
 for some appropriate $A_{n,\alpha,\gamma}(x_2,\xi_2)\in S$ and $a_{n,j,k}(x_2,\xi_2)\in S$. Thus we proved part \eqref{part_one}.
 
 We can formally derive \eqref{eq: Epm_inv} and \eqref{eq: Epm_inv_explicit} for $\sqrt{h}\Epm^{-1}$, using a formal parametrix construction by using 
 \begin{equation}
      \label{eq: fake_sharp}
      a\widetilde\# b \sim \sum_k\frac{1}{k!}\left.\left(\left(\frac{ih}{2}\sigma(D_{x_2},D_{\xi_2};D_y,D_\eta)\right)^k\left(a(x_2,\xi_2)b(y,\eta)\right)\right)\right|_{x_2 = y,\ \xi_2 = \eta}.
\end{equation}
More specifically, there is a formal expansion of $\sqrt{h}\Epm^{-1}$, which is denoted by $\sqrt{h}F_n \sim \sum\limits_j h^{\frac{j}{2}}F_{n,j}$, such that $\frac{1}{\sqrt{h}}\Epm \widetilde{\#} \sqrt{h}F_n= \indic_{2\times 2}$. Denote $\sigma(D_{x_2},D_{\xi_2};D_y,D_\eta)$ in \eqref{eq: fake_sharp} by $\sigma$, we can solve for $F_{n,j}$ by considering
\[
\begin{split}
    \indic_{2\times 2} &= \Epm \widetilde{\#} F_n^{-1} \sim \sum\limits_{\alpha = 0}^\infty \sum\limits_{\beta = 0}^\infty h^{\frac{\alpha+\beta}{2}} E_{n,\alpha} \widetilde{\#} F_{n,\beta}\\
    &= \sum\limits_{\alpha = 0}^\infty \sum\limits_{\beta = 0}^\infty h^{\frac{\alpha+\beta}{2}} \sum\limits_{\gamma = 0}^\infty h^\gamma \left.\left(\left( \frac{i\sigma}{2}\right)^\gamma (E_{n,\alpha}(x_2,\xi_2)F_{n,\beta}(y,\eta))\right)\right|_{x_2 = y, \xi_2 = \eta}\\
    &= \sum\limits_{j = 0}^\infty \sum\limits_{\beta = 0}^j \sum\limits_{\alpha = 0}^{j -\beta} h^{\frac{j}{2}} \left.\left(\left( \frac{i\sigma}{2}\right)^{\frac{j - \alpha - \beta}{2}} (E_{n,\alpha}(x_2,\xi_2)F_{n,\beta}(y,\eta))\right)\right|_{x_2 = y, \xi_2 = \eta}.\\
\end{split}
\]
Then we compare the parameter of the term of $h^{\frac{j}{2}}$ on both sides and get
\[
    -E_{n,0}F_{n,j} = \sum\limits_{\beta = 0}^{j - 1} \sum\limits_{\alpha = 0}^{j - \beta} \left.\left(\left(\frac{i\sigma}{2}\right)^{\frac{j - \alpha - \beta}{2}}(E_{n,\alpha}(x_2,\xi_2)F_{n,\beta}(y,\eta))\right)\right|_{x_2 = y, \xi_2 = \eta},
\]
from which we can solve for $F_{n,j}$. Furthermore, by \eqref{eq: Epm_exp} and $E_{n,0} = z - z_{n,0}$, we can check inductively that for $j\geq 0$, there are $b_{n,j,k,l}$, $c_{n,j,k}$ such that 
\[
\begin{split}
     &F_{n,j} = \sum_{k = 0}^{j} (z - z_{n,0})^{-1} \prod\limits_{l = 0}^{k}\left(b_{n,j,k,l}(x_2,\xi_2;z)(z - z_{n,0})^{-1}\right),\\ &\text{with~} \prod\limits_{l = 0}^k b_{n,j,k,l}(x_2,\xi_2;z) = \sum\limits_{\alpha = 0}^{j +k - 2} z^\alpha c_{n,j,k}(x_2,\xi_2), \text{~for~appropriate~} c_{n,j,k}\in S.
\end{split}
\]

Notice that $\widetilde{\#}$ differs from the actual sharp product $\#$:
\begin{equation}
  \label{eq: true_sharp}
  a\# b = e^{\frac{ih}{2}\sigma(D_{x_2},D_{\xi_2};D_y,D_\eta)} \left(a(x_2,\xi_2)b(y,\eta)\right)|_{x_2 = y,\ \xi_2 = \eta}.
\end{equation}

Now we claim that this formal expansion for $\sqrt{h}F_n$ is legitimate as an asymptotic expansion in $S_\delta^\delta$ and in fact, it is exactly the asymptotic expansion of $\sqrt{h}\Epm$ when $|z|\leq 2\Vert \msV\Vert_\infty$ and $|\Im z|\geq h^\delta$. In fact, $\sqrt{h}(\Epm^{-1} - F_n)\in S^{-\infty}$.

In fact, since $|z|$ is bounded and $|\Im z|\geq h^\delta$ and $F_{n,j}$ is a rational function in $z$, thus $h^{\frac{j}{2}}F_{n,j} \in S_{\delta}^{j(\delta - \frac{1}{2}) +\delta}$. Since $j(\delta - \frac{1}{2}) +\delta\to -\infty$, \eqref{eq: Epm_inv} is not only a formal expansion but is indeed an asymptotic expansion of $F_n$ in the symbol class $S_\delta^\delta$. 

Furthermore, comparing \eqref{eq: fake_sharp} with \eqref{eq: true_sharp}, we see that $F_n\# \Epm = 1 - R_n$ with $R_n \in S^{-\infty}$. By Beal's lemma, there is $\tilde{R}_n\in S^{-\infty}$ such that $(1 - R_n^W)^{-1} = 1-\tilde{R}_n^W$. Thus $\sqrt{h}\Epm^{-1} = F \# (1 - \tilde{R}_n^W)\in S_\delta^\delta$ and have exactly the same asymptotic expansion as $F_n$ in \eqref{eq: Epm_inv} since $\tilde{R}_n\in S_\delta^{-\infty}$. Thus part \eqref{part_two} is proved. 

It follows that $r_n:= \partial_z \Epm \# \Epm^{-1}$ is also well-defined with an asymptotic expansion in $S_\delta^\delta$. Since
\[
\begin{split}
    r_n &\sim \sum\limits_{\alpha = 0}^\infty h^{\frac{\alpha}{2}}\partial_z E_{n,\alpha} \# \sum\limits_{\beta = 0}^\infty h^{\frac{\beta}{2}} F_{n,j}= \sum\limits_{\alpha = 0}^\infty \sum\limits_{\beta = 0}^\infty h^{\frac{\alpha + \beta}{2}} \sum \limits_{\gamma= 0}^\infty h^\gamma \left.\left(\left(\frac{i\sigma}{2}\right)^\gamma(E_{n,\alpha}(x_2,\xi_2;z)F_{n,\beta}(y,\eta;z))\right)\right|_{x_2 = y,\xi_2 = \eta}\\
    &= \sum_{j = 0}^\infty \sum\limits_{\alpha =0}^j \sum\limits_{\beta = 0}^{j - \alpha} h^{\frac{j}{2}} r_{n,j,\alpha, \beta} \left.\left(\left(\frac{i\sigma}{2}\right)^{\frac{j - \alpha - \beta}{2}}(E_{n,\alpha}(x_2,\xi_2;z)F_{n,\beta}(y,\eta;z))\right)\right|_{x_2 = y,\xi_2 = \eta}.
\end{split}
\]
Combining it with part \eqref{part_one} and \eqref{part_two} and the fact that $\sigma$ is linear in $D_{x_2}$, $D_{\xi_2}$, we get part \eqref{part_three}. Part \eqref{part_four} follows directly from parts \eqref{part_one}, \eqref{part_two}, \eqref{part_three} with Prop. \ref{prop: expofQk}.
\end{proof}

\section{For the proof of Lemma \ref{lemma: d.o.s.computable}}
\label{appendix2}
In this subsection, we provide several lemmas that together complete the proof of Lemma \ref{lemma: d.o.s.computable}. We start with a proposition that expresses the Hilbert-Schmidt norm of the quantization in terms of its operator-valued symbol.
\begin{prop}
  \label{prop: HS norm}
  Let $\msH_1$, $\msH_2$ be two Hilbert spaces. Let $P: \RR^2 \to \LH$ be an operator-valued symbol in the symbol class $S(\RR^2_{y,\eta};\LH)$. Furthermore, let $\Vert\cdot \Vert_{\HS}$ denote the Hilbert-Schmidt norm of maps $\msH_1$ to $\msH_2$ or $L^2(\RR_y;\msH_1)$ to $L^2(\RR_y;\msH_2)$. Then 
  \[  \Vert P^W(y,hD_y)\Vert_{\HS}^2=\frac{1}{2\pi h}\int_{\RR^2}\Vert P(y,\eta)\Vert_{\HS}^2 \ dy \ d\eta.\]
    In particular, if $\msH_1 = \msH_2 = \RR$, for the scalar-valued symbol $P$, we have 
  \begin{equation}
    \label{eq: scalarHSnorm}
    \Vert P^W(y,hD_y)\Vert_{\HS}^2 = \frac{\Vert P(y,\eta)\Vert_{L^2(\RR^{2};\RR)}^2}{2\pi h}.
  \end{equation}
    
  \end{prop}
  The next Lemma allows us to interchange the order of trace and integration.

  \begin{lemm}
  \label{lemma:E_-chi_R}
  Let $\Em$, $\Ep$ be as in \eqref{eq:Grushin}. Let $\tcRW$, $\bcRW$ be as in the proof of Lemma \ref{lemma: explicit_asymptotic_expansion_needed}. Then, there exists a constant $C>0$ such that
  \[
    \begin{split}
      &\Vert \bcRW \Em\Vert_{\operatorname{HS}(L^2(\RR^2_x), L^2(\RR_{x_2}))}\leq Ch^{-1/2} R\text{ and }\Vert \Em \tcRW \Vert_{\operatorname{HS}(L^2(\RR^2_x), L^2(\RR_{x_2}))}\leq Ch^{-1/2} R.
    \end{split}
  \] 
  \end{lemm}
  \begin{proof}
    The first equation follows from \eqref{eq: scalarHSnorm}. For the second equation, we first recall that
    \begin{claim}
      \label{claim: compofHS}
      If $a\in S(\RR^{2n};\ML(X,Y);m_1)$, $b\in S(\RR^{2n};\HS(Y,Z);m_2)$ and $m_1m_2\in L^2(\RR^{2n}_{x,\xi}) $, where $m_1,m_2$ are order functions, then
      \[
        b \# a \in S(\RR^{2n};\HS(X,Z);m_1m_2) \text{ and }
        (b  \#  a)^W = b^Wa^W \in \HS(L^2(\RR^n_x;X);L^2(\RR^n_x;Y)).   
      \]
    \end{claim}
  Similar to Lemma 1 in \cite{W95}, we can show that
    \begin{claim}
      \label{claim: tcRw} For any $k'$ such that $1<k'$, we have 
      \begin{enumerate}
        \item $\Em(x_2,\xi_2)\in S(\RR^2_{x_2,\xi_2};\ML(\Bkp;\CC^2))$,
        \item $\tcRw(x,D_{x_1},\xi_2)\in S(\RR^2_{x_2,\xi_2};\HS(\Lt;\Bkp);m)$,
        where $m(x_2,\xi_2) = (1+ (|(x_2,\xi_2)| - R)_+)^{-k'}$ is the order function.
      \end{enumerate}
    \end{claim}
  Then it follows that, by Claim \ref{claim: compofHS}, we have $\Em \# \tcRw \in S(\RR^2_{x_2,\xi_2};\HS(\Lt);m)$, i.e.
  \[
    \Vert \Em \# \tcRw (x_2,\xi_2)\Vert_{\HS(\Lt)} \leq m(x_2,\xi_2) = (1+(|(x_2,\xi_2)| - R)_+)^{-k'}.
  \]
  Thus by Prop. \ref{prop: HS norm}, since for all $k>0$,
  \[\int_{\RR^2}[1+(|(x_2,\xi_2)|-R)_+]^{-2k}dxd\xi=\pi R^2+\mathcal O(R^{\max(1,-2k+2)})=\mathcal O(R^2),\]
  we get $\Vert \Em \tcRW \Vert_{\HS(\Ltx;\Ltxt)} \leq Ch^{-1/2}R$ and the Lemma is proved.
  \end{proof}

\begin{lemm}
  \label{lemma: trace class}
    Let $\Em,\Ep,\Epm$ be as in \eqref{eq:Grushin}. For $\Im z\neq 0$, both operators \[\tcRW\Ep \Epm^{-1}\Em \tcRW \text{~and~} \bcRW \Em \Ep \Epm^{-1}\bcRW \] are trace class as bounded linear operators $\mathcal{L}(\Ltx)$ and $\mathcal{L}(\Ltxt)$, respectively. 
\end{lemm}
\begin{proof}

  By Lemma \ref{lemma:E_-chi_R}, the fact that $\tcRW\Ep$ is the adjoint of $\Em\tcRW$ and boundedness of $\Epm$ from \eqref{eq:Reso Id}, we have
\[
  \begin{split}
    &\Tr_1(\tcRW\Ep \Epm^{-1}\Em \tcRW)\leq \frac{CR^2}{h^{\frac{3}{2}}|\Im z|}\text{ and }\Tr_2(\bcRW \Em \Ep \Epm^{-1}\bcRW)\leq \frac{CR^2}{h^{\frac{3}{2}}|\Im z|}.
  \end{split}
\]

\end{proof}

The second proposition allows us to change the position of $\Em$ in the averaging and limiting process in the proof of Lemma \ref{lemma: d.o.s.computable}.
\begin{lemm}
  \label{lemma: diffoftrace}
  Let $\Em,\Ep,\Epm$ be as in \eqref{eq:Grushin}, then
  \[
    \Tr_{L^2(\RR^2_{x};\CC^4)} (\cRW\Ep \Epm^{-1}\Em \cRW) - \Tr_{L^2(\RR_{x_2};\CC^2)}(\bcRW \Em \Ep \Epm^{-1}\bcRW ) \leq \frac{CR^{\frac{3}{2}}}{h |\Im z|}. 
  \]
\end{lemm}

\begin{proof} Since $\Tr(AB) = \Tr(BA)$ when $AB$ and $BA$ are both of trace class.
\[
  \begin{split}
    &\Tr_1 (\tcRW\Ep \Epm^{-1}\Em \tcRW) - \Tr_2(\bcRW \Em \Ep \Epm^{-1}\bcRW )\\
    =& \Tr_2 (\Em(\tcRW)^2\Ep \Epm^{-1}) - \Tr_2((\bcRW )^2\Em \Ep \Epm^{-1})\\
    =& \Tr_2\left[(\Em\tcRW - \bcRW  \Em )\tcRW\Ep \Epm^{-1}\right] + \Tr_2\left[\bcRW (\Em \tcRW - \bcRW \Em )\Ep \Epm^{-1}\right]\\
    =:& \Tr_2\left[[\Em, \indic_R]_w\tcRW\Ep \Epm^{-1}\right] + \Tr_2\left[\bcRW [\Em, \indic_R]_w\Ep \Epm^{-1}\right]\\
    =:& \Tr_2(A_1) + \Tr_2(A_2)
  \end{split}
\]
where $[\Em, \indic_R]_W := \Em\tcRW - \bcRW  \Em$. Then the following claim completes the proof. 
  \begin{claim}
    \label{claim: TraceAi}
  For $\Im z\neq0$, $A_1$, $A_2$ are trace class operators and there is a $C>0$ such that 
  \[\Tr_2(A_1),\Tr_2(A_2)\leq C h^{-1}|\Im z|^{-3/2}R^{3/2}.\]
  \end{claim}
\end{proof}
  \begin{proof}[Proof of Claim \ref{claim: TraceAi}]
  From Lemma \ref{lemma:E_-chi_R}, we already know
  \[
    \Vert [\Em,\indic_R]_W \Vert_{\HS^{\Weyl}} \leq Ch^{-1/2} R, \quad 
  \]
  where $\HS^{\Weyl} = \HS(\Ltx;\Ltxt)$.
  We will improve the upper bound from $Ch^{-1/2}R$ to $Ch^{-1/2}R^{1/2}$.
  
  Let $\bar{\chi}_R^c=1-\bar{\chi}_R$, $\tilde\indic_R^c = 1 - \tilde\indic_R$.  First notice that from the proof of Lemma \ref{lemma:E_-chi_R}, and replacing $\bar{\chi}_R$ by $\bar{\chi}_R^c$, we have
  \[
    \begin{split}
      &\Vert [\Em, \indic_R]_w(x_2,\xi_2)\Vert_{\HS} \leq \tfrac{C_k}{[1+(R - |(x_2,\xi_2)| )_+]^{k}}\text{ and }\Vert [\Em, \indic_R^c]_w(x_2,\xi_2)\Vert_{\HS} \le \tfrac{C_k}{[1+(|(x_2,\xi_2)| - R )_+]^{k}}
    \end{split}
  \]
  where $[\Em, \indic_R]_w(x_2,\xi_2) = \Em \# \tcRw - \bar\indic_R \# \Em$ is the symbol in $(x_2,\xi_2)$ of $[\Em, \indic_R]_W$ and $\HS = \HS(\Ltxo;\CC^2)$.  Since $[\Em, \indic_R]_w  = -[\Em, \indic_R^c]_w$, we have 
  \[
    \Vert [\Em, \indic_R]_w(x_2,\xi_2)\Vert_{\HS} \leq C_k[1+||(x_2,\xi_2)| - R |]^{-k}.
  \]
  Thus by Prop. \ref{prop: HS norm} and a straightforward computation of the following integral
  \[
    \int_{\RR^2_{x_2,\xi_2}} [1+||(x_2,\xi_2)| - R |]^{-2k}dx_2d\xi_2 = \frac{1}{(2k-2)(2k-1)} + \frac{R}{2k-1} = \mathcal O(R),
  \]
we find that
$    \Vert [\Em, \indic_R]_W\Vert_{\HS^{\Weyl}}\leq Ch^{-1/2}R^{1/2}.$
  Since $\tcRW\Ep$ is the adjoint of $\Em \tcRW$, this yields that
  \[
      \Tr(A_1) \leq Ch^{-3/2}R^{3/2},\quad \Tr(A_2)\leq Ch^{-3/2}R^{3/2}.
  \]
\end{proof}

In next Lemma, we state the averaging property of periodic symbols to reduce the regularized trace to a fundamental cell. 
\begin{lemm}
  \label{prop: periodicsymbol}
  Let $\Em,\Ep,\Epm$, $\bar \indic_R$ be as in \eqref{eq:Grushin}.  Then 
\[
   \lim\limits_{R\to\infty} \frac{1}{4 R^2} \int_{\RR^2}\Tr_{\CC^2}(\bcR \#\partial_z \Epm\#\Epm^{-1}\#\bcR )~d{x_2}~d{\xi_2}   = \frac{1}{|E|}  \int_{E} \partial_{\bar{z}}\tilde{f}  \Tr_{\CC^2}(\partial_z \Epm\#\Epm^{-1}) ~d{x_2} ~d{\xi_2}.
\]
\end{lemm}
The proof of this Lemma can be found in \cite[Prop.\@$3$]{W95}.

\end{document}